\documentclass[10pt]{article}
\usepackage[utf8]{inputenc}
\usepackage[a4paper, portrait, margin=1in]{geometry}
\usepackage{color,amsmath,amsfonts,fullpage,amsthm,mathrsfs,fancyhdr,verbatim,framed,hyperref,xspace,mdframed,tabularx,booktabs,enumitem,comment,float}
\usepackage[dvipsnames]{xcolor}
\usepackage{graphicx}
\usepackage{amsthm}
\usepackage{mathpazo}
\usepackage{amsmath}
\usepackage{amssymb}
\usepackage{mathtools}
\usepackage{thmtools,thm-restate}
\usepackage{fec} 
\usepackage{soul}
\usepackage{authblk}

\newtheorem{theorem}{Theorem}[section]
\newtheorem{lemma}{Lemma}[section]
\newtheorem{corollary}{Corollary}[section]
\newtheorem{definition}{Definition}[section]

\newtheorem{problem}{Problem}

\theoremstyle{definition}
\newtheorem{exmp}{Example}[section]

\renewcommand\bra[1]{{\langle{#1}|}}
\makeatletter
\renewcommand\ket[1]{%
  \@ifnextchar\bra{\k@t{#1}\!}{\k@t{#1}}%
}
\newcommand\k@t[1]{{|{#1}\rangle}}

\newcount\Comments
\newcommand{\kibitz}[2]{\ifnum\Comments=1\textcolor{#1}{#2}\fi}

\newcommand{\sh}[1]{\kibitz{orange}{\textbf{Shouvanik:} #1}}

\usepackage{pifont}

\let\Re\relax
\let\Im\relax
\DeclareMathOperator{\Re}       {Re}
\DeclareMathOperator{\Im}       {Im}
\DeclareMathOperator{\sgn}{sgn}

\title{On Speedups for Convex Optimization \\ via Quantum Dynamics}
\author{Shouvanik Chakrabarti \footnote{These authors contributed equally and are listed in alphabetical order.} \thanks{shouvanik.chakrabarti@jpmchase.com}}
\author{Dylan Herman \protect\footnotemark[1]}
\author{Jacob Watkins 
\protect\footnotemark[1]}
\author{\\\vspace{-.25cm}Enrico Fontana}
\author{Brandon Augustino}
\author{Junhyung Lyle Kim}
\author{Marco Pistoia}
\affil{Global Technology Applied Research, JPMorganChase, New York, NY 10001, USA}
\date{June 2025}

\begin{document}

\maketitle

\begin{abstract}
This work investigates the possibility of quantum speedups for continuous optimization through quantum Hamiltonian simulation. We establish the first rigorous query complexity bounds for unconstrained convex optimization via Quantum Hamiltonian Descent (QHD), a framework recently proposed by Leng et al. In the process, we derive novel resource estimates for simulating Schr\"odinger operators given black-box evaluation to the potential. These estimates apply for any $G$-Lipschitz potential of the form $b(t)f(x)$, and depend only on input simulation parameters. We then apply the simulation bounds to assess the complexity of optimization in the high-dimensional regime.

We demonstrate that QHD, with appropriately chosen schedules, can achieve \emph{arbitrarily fast} convergence rates in the evolution time. The speed limit for optimization is determined solely by the cost of discretizing the dynamics, reflecting a similar property of the classical dynamics on which QHD is based. Taking this cost into account, we show that a $G$-Lipschitz convex function can be optimized to an error of $\epsilon$ with $\widetilde{\Ocal}(d^{1.5} G^2 R^2/\epsilon^2)$ queries, given a starting point that is Euclidean distance $R$ from optimal. Under reasonable assumptions about the query complexity of simulating general Schr\"odinger operators and choice of initial state, we show that $\widetilde{\Omega}(d/\epsilon^2)$ queries are necessary. As a result, QHD does not appear to offer improvements over classical zeroth order methods when $f$ is accessed via exact black-box evaluations.

However, we show that the QHD algorithm can tolerate $\widetilde{\Ocal}(\epsilon^3 /d^{1.5} G^2 R^2)$ noise in function evaluation, and as a result, provides a super-quadratic query advantage over all known classical algorithms that tolerate this degree of evaluation noise in the high-dimensional setting. We leverage this to design a quantum algorithm for stochastic convex optimization that offers a super-quadratic speedup over all known classical algorithms in this regime. Our algorithms also outperforms existing zeroth-order quantum algorithms for noisy (with the same noise tolerance) and stochastic convex optimization in this setting. To our knowledge, these results represent the first rigorous quantum speedups for convex optimization obtained through a dynamical algorithm.
\end{abstract}

\clearpage

\tableofcontents

\clearpage

\section{Introduction}

Quantum computers have the potential to solve complex optimization problems more efficiently than classical methods. The degree of quantum algorithmic speedup in optimization has been the subject of many studies spanning over two decades (for recent surveys see \cite{abbas2024challenges, dalzell2023quantum}). Interest in this topic is driven both by the mathematical richness of problems in the field and by the many practical applications to machine learning, statistics, finance, and engineering. 

Perhaps the most natural optimization algorithms are those that directly simulate the continuous time dynamics of some physical system. These algorithms seek to leverage the fact that such systems experience forces in the direction of minimal potential energy or related quantities. If a system can be identified where these quantities correspond to an objective function, then the simulation of that system yields an optimization algorithm. These connections have been leveraged to great success in classical optimization theory \cite{su2016differential, wibisono2016variational, krichene2015accelerated}. Their most notable success has been in explaining the so-called \textit{acceleration phenomenon} \cite{polyak1964some, nesterov_acceleration} wherein certain fine tuned gradient methods have been shown to achieve asymptotically faster convergence than gradient descent in a variety of domains.

From the point of view of quantum algorithms, the study of optimization through simulating quantum dynamics has been ongoing since the beginning of the field. Perhaps the most famous example is the quantum adiabatic algorithm for combinatorial optimization~\cite{farhi2000quantum}, which leverages \textit{adiabaticity} \cite{born1928beweis}. This is the property in which slowly changing Hamiltonian systems remain in their ground state (optimum) even as the Hamiltonian varies. Due to the flexibility of this framework, the adiabatic algorithm and a heuristic variant called quantum annealing~\cite{kadowaki1998quantum} have been heavily studied from both theoretical and empirical perspectives~\cite{brooke1999quantum, brooke2001tunable, farhi2001quantum, van2001powerful, reichardt2004quantum, aharonov2008adiabatic, altshuler2010anderson, babbush2014adiabatic, boixo2014evidence}. For a comprehensive treatment of adiabatic quantum computing, we refer the reader to \cite{albash2018adiabatic} and the references therein. Other optimization algorithms were originally proposed in the discrete time regime but subsequently analyzed from a dynamical/optimal control perspective, such as the Quantum Approximate Optimization Algorithm (QAOA)~\cite{farhi2014quantum,yang2017optimizing, brady2021optimal, venuti2021optimal} and the Variational Quantum Eigensolver (VQE)~\cite{peruzzo2014variational, magann2021pulses, meitei2021gate, choquette2021quantum}. However, these algorithms are primarily intended for discrete optimization.

In contrast, this work considers the continuous setting, specifically unconstrained convex optimization.
\begin{restatable}{problem}{probnonsmoothconvex}\textup{(Unconstrained Convex Optimization).}
\label{prob:nonsmooth-convex}
    Let $f \colon \Rmbb^d \to \Rmbb$ be a convex\footnote{Our results also hold for \emph{star}-convex functions, all other assumptions kept the same. We only require convexity conditions to hold between a fixed global minimizer $x_\star$ and any point $y$ in the domain.} $G$-Lipschitz continuous function whose set of global minima is denoted by $\mathcal{X}^{\star}$. Given an input point $x_0$ such that $\ell_2(x_0,\mathcal{X}^{\star}) \leq R$, find a point $\tilde{x} \in \R{d}$ such that 
    \begin{equation*}
        f(\tilde{x}) - f(x_\star) \leq \epsilon, \quad \forall x_\star \in \mathcal{X}^{\star}.
    \end{equation*}
\end{restatable}
\noindent Thus, $R$ bounds the known distance from solution, and $\epsilon$ sets the acceptable error tolerance. Our focus on \emph{convex} optimization stems from both practical and theoretical considerations. Practical, because unconstrained convex optimization has plentiful applications to statistics and machine learning. Theoretical, because the rigorous study of the complexity of dynamical algorithms for optimization is an emerging topic, and it is instructive to focus on a domain where classical complexities are well understood. Unconstrained convex optimization provides such a domain, where the optimal query complexities in the standard regimes have been characterized for decades~\cite{nemirovskii1983problem}. There is also extensive research activity on variants of this setting, such as erroneous oracular access to the function or its derivatives~\cite{nemirovskii1983problem, belloni2015escaping}, stochastic optimization~\cite{cesa2004generalization, nemirovski2009robust, duchi2018introductory, shapiro2021lectures, sidford2023quantum}, and bandit optimization~\cite{agarwal2010optimal, agarwal2011stochastic, bubeck2012regret, shamir2017optimal, duchi2015optimal}.

A dynamical approach for characterizing acceleration in convex optimization was introduced by Su, Boyd, and Cand\`es~\cite{su2016differential}, and subsequently generalized by Wibisono, Wilson, and Jordan~\cite{wibisono2016variational}. In the latter work, the authors define a class of dynamical systems in terms of a so-called Bregman Lagrangian
\begin{equation*}
    \mathcal{L}(x,v,t) = e^{\alpha_t + \gamma_t}\left( D_h(x + e^{-\alpha_t}v,x) - e^{\beta_t}f(x)\right),
\end{equation*}
where $x$ denotes the position, $v = \dot{x}$ denotes the velocity, and $D_h$ is the Bregman divergence $D_h(x,y) \coloneqq h(x) - h(y) - \langle \nabla h(x), y - x \rangle$, where $h$ is convex and essentially smooth. Later, Leng et al.~\cite{leng2023quantum} quantized these dynamics in the Euclidean setting $h(x) = \norm{x}^2/2$, which they term \emph{Quantum Hamiltonian Descent (QHD)}. These dynamics are described by a time dependent Schr\"{o}dinger equation, with Hamiltonian
\begin{equation}\label{eq:QHD_Hamiltonian}\tag{QHD}
    H_{\text{QHD}}(t) \coloneqq -e^{\alpha_t - \gamma_t}\frac12\Delta  + e^{\alpha_t + \beta_t + \gamma_t}f(x).
\end{equation}
In both the classical and quantum frameworks, $\alpha_t,\beta_t,\gamma_t$ are differentiable real-valued functions, which we call \emph{schedules}. Wibisono et al.~\cite{wibisono2016variational} show that if the schedules satisfy so-called \emph{ideal scaling} conditions, then for convex $f$ the classical dynamics converge to a global minimum at rate $O(e^{-\beta_t}$). Leng et al.~\cite{leng2023quantum} show a similar result for quantum dynamics. 

Ideal scaling schedules are invariant under a reparameterization of time, and in principle, the rate of convergence can be arbitrarily fast in $t$. Clearly, however, the query complexity of convex optimization cannot be made arbitrarily small, as this would contradict known lower bounds in the classical~\cite{nemirovskii1983problem,nesterov2018lectures,bubeck2015convex} and quantum~\cite{chakrabarti2020quantum,van2020convex,garg2020no,garg2021near} settings. Wibisono et al.~\cite{wibisono2016variational} offer an explanation of this phenomenon by noting that while the convergence of continuous dynamics can be arbitrarily fast, algorithms for discretizing the dynamics become unstable beyond a speed limit. They additionally argue that the maximum speed where discretization is possible corresponds to the optimal (accelerated) classical query complexities. Leng et al.~\cite{leng2023quantum} hint at a similar speed/discretization tradeoff for quantum algorithms, but this tradeoff has yet to be theoretically characterized.

Our main technical contribution is to give a precise theoretical characterization of the dependence of the query complexity on schedules. This leads to rigorous bounds on the query complexity of convex optimization through the simulation of QHD dynamics. We identify precise speed-discretization tradeoffs and show that our results cannot be improved under believable assumptions about initial state and the complexity of simulating this class of Schr\"odinger operators. As a consequence, we identify regimes where QHD simulation does not offer any speedup over comparable classical algorithms, and yet others where (perhaps surprisingly) large speedups over the best known classical algorithms are possible. The next section discusses our results in detail.

\section{Results}
\label{sec:results}

We present our results in two parts. First, we informally state our main theorem regarding the simulation of Schr\"odinger operators, given evaluation access to a potential $f$. Then, based on these resource estimates, we discuss results for optimization using the QHD framework. See Sections~\ref{sec:simulating-schrodinger-dynamics} and~\ref{sec:Quantum_Simulation_Algorithms_Convex_Optimization} for details on each of these findings, and Section~\ref{sec:preliminaries_notation} for notation and relevant background.

\subsection{Real space Hamiltonian simulation}
Classically, the simulation of dynamical systems may be performed using various numerical integration techniques \cite{hairer2006structure, betancourt2018symplectic}. In the case where (closed) quantum dynamics are simulated on a quantum computer, these integrators correspond to algorithms for \emph{time dependent Hamiltonian simulation}, which synthesize the unitary time evolution operator as a quantum circuit. Such algorithms are well-studied in the discrete setting. A careful study of real space simulation complexity has only been recently undertaken~\cite{childs2022quantumsim}, though this analysis focuses primarily on Fourier truncation errors. In order to characterize the query complexity of optimization with QHD, we require optimized and explicit resource estimates for the complexity of simulation with black box potentials, making as few assumptions as possible and accounting for all sources of error.  

We present and analyze a quantum algorithm for real-space simulation based on a pseudo-spectral method known as collocation. The main novelty of the analysis is handling the dynamical propagation of discretization error. Existing results on Schr\"odinger operators and partial differential equations (PDEs) are often tailored to contexts in mathematical physics, e.g. applicable only to a few spatial dimensions, and are thus not immediately useful for complexity-theoretic statements. Our work fills this gap and provides a \emph{computational} result regarding the  digital simulation of Schr\"odinger equations in Euclidean space, which can be applied out-of-the-box. In particular, the result applies when the dimension of the problem is taken to be an asymptotic parameter, which is of crucial importance to high-dimensional settings such as optimization.  

The following informally summarizes our main simulation theorem.
\begin{theorem}[Informal version of Thm~\ref{thm:master_simulation_thm}] 
\label{thm:master_simulation_thm_informal}
    Consider the Schr\"odinger equation
    \begin{align*}
    i \partial_t \Phi(x, t) &= [- a(t) \Delta + b(t)f(x)]\Phi(x, t),
    \end{align*}
    subject to initial data $\Phi(x,0) = \Phi_0(x)$, where $a, b$ are sufficiently smooth functions of time and $f$ is a $G$-Lipschitz function with bound $\Lambda \geq \norm{f}_\infty$, accessed through a $\epsilon_f$-noisy binary quantum oracle $O_f$.  If $\epsilon_f = \widetilde{\Ocal}(\epsilon/\lVert b \rVert_1)$, then there is a digital quantum algorithm that outputs an $\epsilon$-approximation $\ket{\Psi_t}$ to $\Phi(\cdot, t)$ (in an appropriate sense defined later) using $\Ocal\left(\Lambda \lVert b \rVert_1\right)$ queries to $O_f$, 
    $\Ocal\left(d^2\cdot \polylog(1/\epsilon, \lVert a \rVert_1, \lVert b \rVert_1)\right)$ qubits and $\widetilde{\Ocal}\left(\poly(d, \Lambda, \lVert b \rVert_1)\right)$ gates.
\end{theorem}
\noindent To the best of our knowledge, this is the first rigorous result on the quantum computational complexity of real-space quantum dynamics given solely in terms of basic simulation parameters. We prove Theorem~\ref{thm:master_simulation_thm_informal} using multi-dimensional Fourier analysis and bounding various sources of error that typically occur in signal processing applications, such as spectral truncation and aliasing. These errors are connected to the so-called Sobolev norms of the functions involved, and occur because the digital quantum algorithm is unable to fully represent the infinite spectrum of the continuous problem. Our results apply to separable potentials $b(t)f(x)$ where $b$ is differentiable and $f$ is $G$-Lipschitz. These weak assumptions on smoothness are possible through a mollification argument, a common technique in signal processing and functional analysis.

The precise versions of our results on simulation are presented in Section \ref{sec:simulating-schrodinger-dynamics}. The results and proofs are presented in an independent manner and may be of interest beyond optimization.

\subsection{Convex Optimization}
With the simulation results of Theorem~\ref{thm:master_simulation_thm_informal}, we can calculate the query complexity of optimization via QHD. Our first result (Theorem~\ref{thm:QHD_convergence_rate}) generalizes the convergence of QHD under modified ideal scaling conditions. Additionally, we relax the conditions on the potential from being differentiable to Lipschitz continuous. We note that this relaxation is for \emph{any} schedule that obeys the ideal scaling conditions. Since our goal is to completely determine the query complexity, we keep track of (possibly dimension-dependent) factors in addition to the convergence rate. 

Our main theorem on unconstrained convex optimization via Hamiltonian simulation is as follows.
\begin{restatable}{theorem}{thmconvexmainspecialized}\label{thm:convex-main-specialized}
    There exists a quantum algorithm that solves Problem~\ref{prob:nonsmooth-convex} using $\widetilde{\Ocal}\big(d^{3/2}(GR/\epsilon)^2\big)$  queries to an evaluation oracle for $f$ with error $\widetilde{\Ocal}\left(\frac{\epsilon^3}{d^{3/2} G^2 R^2}\right)$.
\end{restatable}
\noindent A proof is provided in Section~\ref{sec:Quantum_Simulation_Algorithms_Convex_Optimization}, and follows from a general theorem (Theorem~\ref{thm:master_simulation_thm}) which also tallies qubit count and the gate complexity.\footnote{Our simulation algorithm is applicable to the simulation of QHD under a mild technical condition connected to the boundary conditions. We discuss these conditions and why they are very likely to hold in Section~\ref{sec:Quantum_Simulation_Algorithms_Convex_Optimization}. Whenever the simulation algorithm is applicable, the query complexities are as indicated in our theorems.} Our analysis considers a modified parametrization of the QHD Hamiltonian from~\eqref{eq:QHD_Hamiltonian}. In what follows we write 
\begin{equation} \label{eq:QHD_springmass}
    H_\mathrm{QHD}(t) = c_t \left(\frac{\hat{p}^2}{2 m_t} + m_t \omega_t^2 f(\hat{x})\right)
\end{equation}
where $m_t, \omega_t$ are positive differentiable functions satisfying the ideal scaling conditions
\begin{equation}\label{eq:ideal_scaling}\tag{Ideal Scaling}
        \frac{\dot{m}_t}{m_t} = \lambda c_t, \quad \frac{\dot{\omega}_t}{\omega_t} \leq \frac12 \lambda c_t
\end{equation}
for some $\lambda \in \Rmbb_+$. Besides providing greater physical intuition, this parametrization also ensures dimensional consistency, whereas prior expressions for Ideal Scaling left this point unclear. Evidently, the parameter $c_t = e^
{\alpha_t}$ merely plays the role of time parametrization. While such parametrization can affect algorithmic performance, the time dependent simulation algorithm we employ, based on the Dyson series method, will be robust to such reparametrizations.

Note that the query complexity in Theorem~\ref{thm:convex-main-specialized} has a dependence of $1/\epsilon^2$ on the precision, which is due to the query cost of discretization. In fact, the same complexity can be obtained from a schedule with \emph{arbitrarily fast} convergence (and consequently, arbitrarily small simulation time). As specific examples from prior literature,~\cite{wibisono2016variational,leng2023quantum} we construct schedules with arbitrary exponential or polynomial convergence rates.
\begin{definition}[Exponential Schedules with Convergence Rate $e^{-ct}$]
\label{defn:exponential-schedule}
    Let $c \geq 0$ be a parameter that determines the rate of convergence, and $m_0,\omega_0 \in \Rmbb_+$ be arbitrary. A family of \ref{eq:QHD_Hamiltonian} Hamiltonians with exponential convergence rate $e^{-ct}$ is defined by schedules $c_t = c$, $m_t = m_0 e^{ct}$, and $\omega_t = \omega_0 e^{ct/2}$. The resulting Hamiltonian is
    \begin{equation}
    \label{e:qhd-exponential} \tag{QHD-exponential}
        H_c(t) = c \left(-e^{-ct} \frac{1}{2 m_0} \Delta + e^{2ct} m_0 \omega_0^2 f\right)
    \end{equation}
    for $t \in [0,\infty)$.
\end{definition}

\begin{definition}[Polynomial Schedules with Convergence Rate $t^{-k}$]
\label{defn:polynomial-schedule}
    Let $k > 0$ be a parameter that determines the rate of convergence, and $m_0, \omega_0, t_0 \in \Rmbb_+$ be arbitrary. A family of \ref{eq:QHD_Hamiltonian} Hamiltonians with polynomial convergence rate $t^{-k}$ is defined by schedules $c_t = k/t$, $m_t = m_0 (t/t_0)^k$, and $\omega_t = \omega_0 (t/t_0)^{k/2}$. The resulting Hamiltonian is
    \begin{equation}
    \label{e:qhd-polynomial} \tag{QHD-polynomial}
        H_k(t) = \frac{k}{t}\left(-(t/t_0)^{-k} \frac{1}{2 m_0}\Delta + (t/t_0)^{2k} m_0 \omega_0^2 f\right)  
    \end{equation}
    for $t \in [t_0,\infty)$.
\end{definition}
The rate of convergence for these schedules, here provided within the definitions for clarity, are simple corollaries of Theorem~\ref{thm:QHD_convergence_rate}. We also show how any ideal-scaling schedule can be discretized to obtain the query complexity stated in Theorem~\ref{thm:convex-main-specialized}.
Finally, we note that, while the choice of convergence rate is relatively arbitrary and leads to the same query complexity, the remaining parameters of the schedule such as $m_0, \omega_0$ must be chosen carefully. This is because the convergence of QHD is shown via a Lyapunov ("nonincreasing") function $\Ecal_t$ which loosely tracks the ``energy" of the system. The convergence bounds are directly proportional to $\Ecal_0$, which may be dimension dependent. In fact, using QHD schedules already stated in previous works~\cite{leng2023quantum,leng2025quantumhamiltoniandescentnonsmooth} result in a dimension dependence of $\widetilde{\Ocal}(d^{2.5})$ in the query complexity, compared to the $\widetilde{\Ocal}(d^{1.5})$ obtained here. We will comment on this consideration in  the technical sections where relevant.

\subsubsection{Schedule Invariance and Lower Bounds on Query Complexity}
We now consider the possibility of improvements over our upper bounds within the QHD framework. The algorithmic results we provide are independent of ideal-scaling schedule, but perhaps there are optimizations in terms of the initial parameters, the initial state, or in terms of the domain of simulation. As discussed in the previous section, we can choose schedules that converge in continuous time at arbitrarily fast rates, and thus simulate for an arbitrarily short amount of time to reach the target precision. Lower bounds on the query complexity arise from the increasing cost of simulating these arbitrarily fast schedules. We make the following assumption on the complexity of simulating Schr\"{o}dinger Dynamics.
\begin{restatable}{assumption}{asmnofastforward}\label{asm:no-fast-forward}\textup{(No Fast-Forwarding for Time Dependent Schr\"{o}dinger Operator)} 
    A potential $f : \R{d} \to \R{}$ satisfies the ``no fast-forwarding" assumption if simulating the time dependent Hamiltonian $H(t) = -a(t)\Delta + b(t)f(x)$ for $t \in [t_1,t_2]$ requires $\Omega\left(\Lambda_f \cdot \int_{t_1}^{t_2} \abs{b(t)} dt  \right)$ evaluation queries to $f$, where $\Lambda_f$ is a linear functional of $f$ and independent of $t$.
\end{restatable}
\noindent We call this assumption ``no fast-forwarding" since, with constant $b(t)$, it reduces to the familiar notion for finite-dimensional, time-independent Hamiltonians, and because it is the most reasonable and natural generalization of this notion. It is possible in principle that, by restricting to convex $f$ and for some particular choices of schedules $b(t)$, this assumption may be violated. However, we argue that it is very likely to be true for simulation schemes that work in a black box setting, for several reasons. First, the assumption is satisfied by all general algorithms for time-dependent Hamiltonian simulation of Schr\"{o}dinger operators that we are aware of, including the one in this paper. Second, this dependence is inherent in the methods for time-dependent Hamiltonian simulation via the interaction picture, that these works are based on. Finally, Assumption~\ref{asm:no-fast-forward} reflects intuition in the following (non-rigorous) sense: in each time slice $[t,t+dt)$ where $dt$ is small enough that $b(t)$ is nearly constant, it follows from the no fast forwarding theorem for time independent potentials that the cost of simulating the Schr\"{o}dinger operator in this time slice is at least $b(t)f dt$. Approximating the whole evolution by these time slices results in the dependence encoded in Assumption~\ref{asm:no-fast-forward}.

We are now ready to state our lower bounds on the query complexity of convex optimization by simulating QHD.
\begin{restatable}{theorem}{lowerboundqhdmain}\label{thm:lower-bound-qhd-main}
    Let $f$ be a potential satisfying the no fast-forwarding assumption (Assumption~\ref{asm:no-fast-forward}), and the conditions of Problem~\ref{prob:nonsmooth-convex}. Then any optimization algorithm that solves Problem~\ref{prob:nonsmooth-convex} via the simulation of QHD under ideal scaling conditions, starting from a real initial state, must make $\Omega(d\Lambda_f GR/\epsilon^2)$ queries to an evaluation oracle for $f$, if the continuous-time convergence is as predicted by Theorem~\ref{thm:QHD_convergence_rate}.
\end{restatable}
\noindent The additional $\sqrt{d}$ factor in our upper bounds(Theorem~\ref{thm:convex-main-specialized} arises from a bound on $\Lambda_f$ over the hypercube. Together, our upper and lower bounds almost completely characterize the query complexity of convex optimization with QHD. 

We point out that we make no assumptions beyond Lipschitzness for the convex functions considered in Problem~\ref{prob:nonsmooth-convex}. As such, Problem~\ref{prob:nonsmooth-convex} gives rise to the ``nonsmooth" setting in convex optimization theory~\cite{bubeck2015convex, nesterov2018lectures}. It is also common to study settings with more assumptions on the function class. Concretely, it is standard to additionally assume Lipschitz continuity of gradients (this is often referred to as ``smoothness" in optimization theory \cite{rockafellar1994nonsmooth}),\footnote{Note that this definition of smoothness is incomparable to the typical one in analysis, where smoothness implies infinite differentiability.} and strong convexity. It is known that gradient methods are more efficient when we assume Lipschitz continuous gradients or, Lipschitz continuous gradients and strong convexity. Smoothness assumptions enable a convergence rate of $\Ocal(1/k)$, which improves upon the $\Ocal(1/\sqrt{k})$ rate possible for nonsmooth optimization. When one further assumes strong convexity, the rate of convergence is exponential. 

In contrast, our results are equally applicable to all the settings and only require the assumptions corresponding to the nonsmooth setting. Namely, from the point of view of convergence in continuous time, arbitrarily fast rates can be obtained from just these assumptions (matching the convergence of gradient methods in the smooth or smooth and strongly-convex case, with fewer assumptions).
In order to get improved quantum query complexities corresponding to the ones for gradient methods, our lower bounds show that we must use these additional assumptions to simulate the Schr\"{o}dinger equation faster than possible in general (specifically, the simulation must violate Assumption~\ref{asm:no-fast-forward}). As of this writing, no such algorithm exists but the investigation of whether one is possible is of independent interest.

\subsubsection{Prospects for Quantum Speedup for Convex Optimization Algorithms}

We now compare our query complexity bounds to classical algorithms for unconstrained optimization. Our results are summarized in Table~\ref{tab:compare-noisy}.

\paragraph{Comparison to Classical Algorithms for Noiseless Convex Optimization} 

 The first-order query complexity for the nonsmooth, smooth, and smooth and strongly convex settings is well-understood~\cite{bubeck2015convex, nesterov2018lectures}. When $f$ is $G$-Lipschitz, one can determine an $\epsilon$-approximate minimizer using $\Ocal ((GR/\epsilon)^2)$ queries. If $f$ has $L$-Lipschitz continuous gradients the bound improves to $\Ocal (LR^2/\epsilon)$, and when $f$ is additionally $\mu$-strongly convex the number of queries is at most $\Ocal (\kappa \log(1/\epsilon))$, where $\kappa := L/\mu$. Note that these bounds are for \textit{unaccelerated} methods. If one employs Nesterov's acceleration~\cite{nesterov_acceleration}, these results improve to $\Ocal (R \sqrt{L/\epsilon})$ for $L$-smooth functions, and $\Ocal (\sqrt{\kappa} \log(1/\epsilon))$ for $L$-smooth, $\mu$-strongly convex functions. 

The \textit{zeroth-order} complexity in these settings has only been established more recently. Most of the results along this line concern algorithms that construct gradient estimators by evaluating the function value at a fixed number of random points. The number of permitted points may vary, though the most frequent setups involve 2-point estimators~\cite{nesterov2017random, duchi2015optimal, shamir2017optimal} and single-point estimators~\cite{flaxman2005online, agarwal2011stochastic, bach2016highly, belloni2015escaping, bubeck2021kernel}. Compared to the first-order query complexities, the upper bounds obtained for zero-order optimization suffer a $\poly(d)$ slow-down~\cite{gasnikov2023randomized}. Specifically, for our setting of nonsmooth optimization, classical zero-order methods achieve an oracle complexity $\Ocal \left( d( GR/\epsilon)^2\right)$, see, e.g., \cite{nesterov2021implementable}. In comparison, we show query upper bounds of $\widetilde{\Ocal}\left(d^{1.5}(GR/\epsilon)^2\right)$, and if the fast forwarding assumption holds and the continuous time convergence rates match those of QHD~(Theorem~\ref{thm:QHD_convergence_rate}), a lower bound of $\widetilde{\Omega}\left(d(GR/\epsilon)^2\right)$ which matches the classical upper bound. Thus, if the potential corresponding to the objective function satisfies Assumption~\ref{asm:no-fast-forward}, and if the QHD convergence bound of Theorem~\ref{thm:QHD_convergence_rate} is sharp in general, \emph{there is no quantum speedup for zeroth order noiseless convex optimization by simulating QHD.}

\paragraph{Comparison to Classical Algorithms for Noisy Convex Optimization}

\begin{table}
\centering
\resizebox{\textwidth}{!}{\begin{tabular}{llll}
\toprule \hline 
\textbf{Classical Algorithms} & \textbf{Reference} &  \textbf{Query complexity} &   \textbf{Admissible noise}   \\ \hline \\[-1.5ex]
Stochastic Gradient Estimator & \cite{risteski2016algorithms} &  $\widetilde{\Ocal}\left(d^4 (GR/\epsilon )^6\right)$ & $\max \left\{  \frac{\epsilon^2}{\sqrt{d}GR}, \frac{\epsilon}{d}\right\}$\\[-1.5ex]\\  
Simulated annealing & \cite{belloni2015escaping} &  $\widetilde{\Ocal} \left(d^{4.5} \right)$ & $\epsilon/d$\\[-1.5ex]\\    
\toprule \hline
\textbf{Quantum Algorithms} & \textbf{Reference} &  \textbf{Query complexity} &   \textbf{Admissible noise}   \\ \hline \\[-1.5ex]
Quantum simulated annealing & \cite{li2022quantum} &  $\widetilde{\Ocal} \left(d^{3} \right)$ & $\epsilon/d$ \\[-1.5ex]\\ 
\textcolor{blue}{QHD simulation (this paper)} & \textcolor{blue}{Theorem~\ref{thm:convex-main-specialized}}  &  \textcolor{blue}{$\widetilde{\Ocal}\left(d^{1.5} (GR/\epsilon)^2\right)$} & \textcolor{blue}{$\widetilde{\Ocal}\left(\frac{\epsilon^3}{d^{1.5} G^2 R^2}\right)$} \\[-1.5ex]\\
Quantum subgradient method & \cite{augustino2025fast} & $\widetilde{\Ocal}\left((GR/\epsilon)^2\right)$ & $\Ocal \left( \frac{\epsilon^5}{d^{4.5} G^4 R^4 } \right)$ \\[-1.5ex]\\
\bottomrule
\end{tabular}}
\caption{Complexity of algorithms for noisy unconstrained zero-order convex optimization. Admissible noise is an upper bound on $\epsilon_f$, the error in the evaluation of $f$ using a binary oracle (see Definition~\ref{defn:binary_oracle}).}
 \label{tab:compare-noisy}
\end{table}

While our results for noiseless convex optimization are largely negative, we have a more positive outlook for the case of convex optimization with noisy oracles. By noisy, we assume that evaluations of $f$ are promised to be accurate to precision $\epsilon_f$ (see for instance Definition~\ref{defn:binary_oracle}). We do not assume that $\epsilon_f$ follows some probability distribution, and we will continue to use the term ``noisy" for this general setting. 

The best known classical algorithms with erroneous oracles are due to Belloni et al.~\cite{belloni2015escaping} and Ristetski and Li~\cite{risteski2016algorithms}. We discuss both results, since each presents advantages over the other. \cite{belloni2015escaping} presents an algorithm based on simulated annealing using a hit and run walk, which requires $\widetilde{\Ocal}(d^{4.5}\log(1/\varepsilon))$ evaluations of the objective, with error $\Ocal(\varepsilon/d)$. On the other hand, \cite{risteski2016algorithms} presents an algorithm based on a stochastic gradient estimator due to~\cite{flaxman2005online}. Their algorithm requires $\Ocal(d^4/\varepsilon^6)$ evaluations of the objective with error at most $\Ocal(\max(\varepsilon^2/\sqrt{d},\varepsilon/d))$. We note that Ristetski and Li only show the query complexity of their algorithm to be polynomial scaling in $d,\varepsilon$ but do not determine the polynomial. To make an accurate comparison, we determine this scaling in Section~\ref{sec:li-risteski-scaling}.

In contrast, our algorithm is shown to require $\Ocal(d^{1.5}/\varepsilon^2)$ quantum evaluation queries, each with error $\widetilde{\Ocal}(\varepsilon^3/d^{1.5})$. We note that if $\varepsilon = \Theta(d^{-\varrho})$ for some $\varrho > 0$, our query complexities are polynomially lower than the classical algorithms. Up to polylogarithmic factors, the exponent of the speedup over \cite{belloni2015escaping} is $1/3 + 4\varrho/9$ and that over \cite{risteski2016algorithms} is $(3 + 4\varrho)/(8+12\alpha)$. Therefore, if $\varrho$ is small, the speedup is super-quadratic over both classical algorithms. We note that the classical algorithms are robust to higher levels of noise than ours, so it is natural to ask if they can be improved if a lower noise threshold is taken into account. We argue in Appendix~\ref{app:prospects-reduction} that such an improvement is not immediate from existing considerations. It is of course possible that such an improvement can be attained with more effort as there are few applicable lower bounds in general, so it is important to note that our claimed speedups are over \emph{best known} classical algorithms, and not the best possible ones.

\paragraph{Comparison to Quantum Algorithms for Noisy Convex Optimization}
In the quantum model of computation, the best known algorithms in the noisy setting are attributed to Li and Zhang~\cite{li2022quantum} and Augustino et al.~\cite{augustino2025fast}. The algorithm in \cite{li2022quantum} can be viewed as a quantum analogue of the simulated annealing algorithm from Belloni et al.~\cite{belloni2015escaping} in that the classical hit and run walk is replaced with a quantum hit and run walk originally proposed in \cite{chakrabarti2023quantum}. The quantum hit and run walk requires quadratically fewer samples per iteration compared to its classical counterpart, which improves the overall query complexity from $\widetilde{\Ocal}(d^{4.5}\log(1/\varepsilon))$ to $\widetilde{\Ocal}(d^{3}\log(1/\varepsilon))$. As in \cite{belloni2015escaping}, their algorithm can tolerate error up to $\Ocal(\varepsilon/d)$. 

Augustino et al.~\cite{augustino2025fast} introduce a zero-order quantum subgradient method, wherein the subgradients are estimated via a generalization of Jordan's algorithm for gradient estimation~\cite{jordan2005fast} to subgradient estimation~\cite{van2020convex, chakrabarti2020quantum}. The quantum subgradient method proposed in~\cite{augustino2025fast} achieves a query complexity of $\widetilde{\Ocal}( (GR/\epsilon)^2)$, which is exponentially faster than all known zero-order methods in the problem dimension. This algorithm works in both the noiseless and noisy settings, tolerating function evaluation error up to $\Ocal (\epsilon^5/(d^{4.5} G^4 R^4 ))$ in the latter. 

These results and our analysis highlight a tradeoff between the query complexity and the robustness of quantum algorithms for noisy convex optimization. For the high-dimensional regime where one has $d \gg \poly(G, R, \epsilon^{-1})$, it is clear that the quantum subgradient method is the best performing algorithm with respect to query complexity, but is the least robust in terms of the amount of noise it can tolerate in the function evaluations. Meanwhile, QHD achieves a superior query complexity compared to quantum simulated annealing~\cite{li2022quantum}, and is significantly more robust than the quantum subgradient method~\cite{augustino2025fast}.

\paragraph{Comparison to Algorithms for Stochastic Convex Optimization}
QHD's combination of the superior query complexity compared to more robust algorithms, and higher noise tolerance compared to faster algorithms, leads to an additional speedup for the problem of \textit{unconstrained stochastic convex optimization}. 
\begin{restatable}{problem}{probstochasticconvex}\textup{(Unconstrained Stochastic Convex Optimization).}
\label{prob:stochastic-convex}
    Let $f \colon \R{d} \to \R{}$ be a (star) convex $G$-Lipschitz continuous function whose set of global minima is denoted by $\mathcal{X}$. Suppose $f$ is accessed via a random oracle (quantum or classical) which, given input $x$, returns $f(x) + \xi$ where $\xi$ has zero mean and bounded variance. Given an input point $x_0$ such that $\ell_2(x_0,\mathcal{X}) \leq R$, an algorithm solves the stochastic convex optimization problem if it returns a point $\tilde{x}(\xi) \in \R{d}$ such that 
    \begin{equation*}
        \Embb_{\xi^T \sim \Xi^T} [f(\tilde{x}(\xi))] - f(x_\star) \leq \epsilon, \quad \forall x_\star \in \mathcal{X}.
    \end{equation*}
    Here $\Xi$ is the distribution for $\xi$ and $T$ is the number of queries the algorithm makes to $f$.
\end{restatable}

The following result is an immediate consequence of combining the result in Theorem~\ref{thm:convex-main-specialized} with quantum mean estimation \cite{montanaro2015quantum}. 
\begin{restatable}{theorem}{thmconvexstochastic}\label{thm:convex-stochastic}
    There exists a quantum algorithm that solves Problem~\ref{prob:stochastic-convex} using $\widetilde{\Ocal}\left(d^{3}(GR/\epsilon)^5\right)$ queries to a stochastic evaluation oracle for $f$.
\end{restatable}
\noindent A proof is provided in Section~\ref{sec:stochastic-convex}, and can be informally summarized as follows. Our algorithm for convex optimization evaluates $f$ up to error $\widetilde{\Ocal}\left(\frac{\epsilon^3}{dG^2 R^2}\right)$. On the other hand, Montanaro~\cite{montanaro2015quantum} proves that if one has access to an oracle $O_f$ for evaluating $f$ up to error $\epsilon_f$, then $\widetilde{\Ocal}(1/\epsilon_f)$ applications of $O_f$ suffice to estimate $\Embb[f(x)]$.

\begin{table}
\centering
\begin{tabular}{lll}
\toprule \hline 
\textbf{Classical Algorithms} & \textbf{Reference} &  \textbf{Query complexity}  \\ \hline \\[-1.5ex]
Simulated annealing & \cite{belloni2015escaping} &  $\widetilde{\Ocal} \left( d^{7.5} (GR/\epsilon)^2\right)$\\[-1.5ex]\\    
\toprule \hline
\textbf{Quantum Algorithms} & \textbf{Reference} &  \textbf{Query complexity} \\ \hline \\[-1.5ex]
Quantum simulated annealing & \cite{li2022quantum} &  $\widetilde{\Ocal} \left(d^{5}GR/\epsilon\right)$ \\[-1.5ex]\\ 
\textcolor{blue}{QHD simulation (This paper)} & \textcolor{blue}{Theorem~\ref{thm:convex-stochastic}}  &  \textcolor{blue}{$\widetilde{\Ocal}\left(d^{3} ( GR/\epsilon )^5 \right)$} \\[-1.5ex]\\ 
Quantum subgradient method & \cite{augustino2025fast} & $\widetilde{\Ocal}\left(d^{4.5} (GR/\epsilon)^7 \right)$ \\[-1.5ex]\\
\bottomrule
\end{tabular}
\caption{Complexity of algorithms for  unconstrained zero-order stochastic convex optimization. See Problem~\ref{prob:stochastic-convex} for a statement of the problem
.}
 \label{tab:compare-stochastic}
\end{table}

Accordingly, the performance of QHD is favorable over the simulated annealing algorithm of Belloni et al.~\cite{belloni2015escaping} and the quantum algorithms \cite{li2022quantum, augustino2025fast} in the zero-order stochastic setting, as highlighted in Table~\ref{tab:compare-stochastic}. Indeed, QHD achieves a super-quadratic speedup in the dimension over the algorithm in \cite{belloni2015escaping}, and a sub-quadratic speedup in $d$ over quantum simulated annealing~\cite{li2022quantum}. The query complexity of QHD in this setting also outperforms that of the quantum subgradient method in every parameter. We further remark that the dependence on dimension is close to optimal, due to an $\Omega (d^2 (GR/\epsilon)^2)$ lower bound from Shamir~\cite{shamir2013complexity}.

\subsection{Related Work}
\subsubsection*{Comparison with Leng et al.~\cite{leng2025quantumhamiltoniandescentnonsmooth}}

Recently, Leng et al.~\cite{leng2025quantumhamiltoniandescentnonsmooth} published a preprint exploring the application of QHD to nonsmooth optimization problems, including strongly convex and nonsmooth nonconvex problems. Given the close relation to our work, we provide a detailed comparison.

Both our study and that of Leng et al. focus on the algorithmic performance of QHD for convex optimization. Leng et al. examine nonsmooth strongly convex, nonsmooth convex, and general non-convex optimization settings. They demonstrate specific QHD schedules with continuous time convergence rates of $e^{-\sqrt{\mu} t}$ and $t^{-2}$ for nonsmooth strongly convex and nonsmooth convex optimization, respectively. Their Trotter-based simulation scheme suggests discrete convergence rates of $t^{-1}$ and $t^{-2/3}$, indicating potential speedup for nonsmooth convex optimization. However, they do not provide explicit query complexities or analyze dimension dependence.

Our work diverges in several key technical aspects:
\begin{enumerate}
    \item \textbf{Generalization of Lyapunov Argument}: We extend the Lyapunov argument to nonsmooth potentials independently of the schedule, unlike Leng et al., who focus on specific schedules. This allows us to achieve arbitrary exponential or polynomial convergence rates in continuous time with QHD. It is notable that the $e^{-\sqrt{\mu} t}$ rate exhibited for strongly convex potentials in \cite{leng2025quantumhamiltoniandescentnonsmooth} can also be obtained for potentials that are only convex. For this reason, we make no distinction between convex and strongly-convex potentials in this work.

    \item \textbf{Dimensional Dependence}: Our analysis tracks dimension-dependent Lyapunov functions, leading to the first explicit query complexity bounds for convex optimization, considering both dimension and target precision. We also provide evidence that these rates are optimal unless the potential simulation bypasses a "no fast-forwarding" assumption. This illustrates, just as in the classical case, the discretized cost of optimization with QHD can be made independent of the continuous time convergence of the underlying schedule. Finally, we note that we have to further optimize our QHD schedules in order to get our upper bounds, both in terms of the initial values of the schedule functions, and the generalizations made to ideal scaling. Simply simulating the schedules given in~\cite{leng2025quantumhamiltoniandescentnonsmooth} results in polynomially worse dimension dependence.

    \item \textbf{Simulation Analysis}: Our simulation analysis is more comprehensive, addressing some error terms not covered in previous work. Detailed comparisons with other simulation algorithms are provided in Section~\ref{sec:comparison}.

    \item \textbf{Classical Algorithm Comparisons}: We uniquely compare QHD to classical algorithms in noiseless, noisy, and stochastic convex optimization settings, focusing on theoretical guarantees rather than numerical analysis.
\end{enumerate}

In summary, our work complements and extends the findings of Leng et al., confirming some results while providing tighter or more explicit bounds. We address a key open question by offering rigorous query complexity bounds for QHD in nonsmooth optimization.

An important distinction is the claimed $t^{-2/3}$ convergence rate for discretized QHD simulation by Leng et al. Our lower bound results suggest this rate is feasible if the potential simulation violates a "no fast-forwarding" condition or, more likely, if the convergence rate of QHD is faster than the Lyapunov arguments guarantee. Future research could explore which of these two cases the potentials of Leng et al, satisfy. In either case, it would be interesting to identify families of potentials with this property.

\subsubsection*{Other Related Works}

Catli, Simon and Wiebe~\cite{catli2025exponentially} also investigate dynamical algorithms for convex optimization. Their query complexity results are applicable to locally quadratic convex potentials and to strongly convex and smooth potentials. The query complexity is given in terms of an unfixed discretization spacing, and no specific speedups are identified as the goal of that work is primarily to compare to other quantum algorithms. Notably, our results on simulation decouple the discretization spacing from the oracle complexity, although explicit bounds on the discretization spacing are also given. Augustino et al.~\cite{augustino2025fast} give quantum algorithms for noiseless convex optimization via quantum subgradient methods. Likewise, Li and Zhang~\cite{li2022quantum} give quantum algorithms for noisy convex optimization and stochastic optimization via quantum hit and run walks. Explicit comparisons to the query complexities for these algorithms are given in Section~\ref{sec:results}. Sidford and Chen \cite{sidford2023quantum} provide quantum algorithms for first-order, stochastic optimization based on an unbiased version of quantum mean estimation.

Zhang et al. \cite{Zhang_2021} proposed a quantum version of the perturbed gradient descent algorithm, by Jin et al. \cite{jin2017escapesaddlepointsefficiently, jin2017acceleratedgradientdescentescapes}, for efficiently escaping saddle points in  continuous non-convex optimization. Their algorithm replaces the classical Gaussian perturbation by measurement and Schr\"odinger evolution, which disperses faster from the saddle point. These algorithms typically require some Hessian smoothness and non-degeneracy of the saddle points. Chen et al. \cite{chen2025quantum} proposed an open-system variant of QHD that they called Quantum Langevin Dynamics and is based on the quantum master equation. Via a Lyapunov based argument, they show the continuous-time convergence for convex and quasar-convex functions. The comparisons that they make to QHD are mostly in the non-convex case and empirically motivated.

Childs et al. \cite{childs2022quantumsim} were the first to propose and analyze an algorithm for simulating real-space quantum dynamics, where they focus mostly on the Fourier truncation error for pseudo-spectral methods. Our results generalize this algorithm to the black box setting by reducing the assumptions on the wavefunction. We also bound other error terms that are missing in previous analyses. We summarize these differences in Section \ref{sec:comparison}.

A few weeks after the first appearance of this draft on arXiv, Leng et al \cite{leng2025sub} posted a preprint describing a sub-exponential query speedup for optimization. The focus of~\cite{leng2025sub} is on nonconvex optimization while our focus in this manuscript is on convex optimization. Additionally, there is an important differences in the model from that which we consider here. The proof in \cite{leng2025sub} proceeds via a reduction from a known separation between stoquastic adiabatic evolution and classical algorithms for the problem of finding an exit node in a graph. The setup in~\cite{leng2025sub} encodes the exit node in the cost function and the graph in an adiabatic path that finds this node. In contrast, we work in the black box setting where the only function access we assume is through the evaluation oracle. As a consequence, these results are not directly comparable.

\section{Background}
\subsection{Preliminaries and Notation} \label{sec:preliminaries_notation}
We reserve $d\in\Zmbb_+$ to refer to the dimension of the Euclidean space $\Rmbb^d$, define $\Nmbb \coloneqq \Zmbb_+ \cup \{0\}$, and define $[m] \coloneqq \{1, 2,\ldots, m\}$ for $m\in \Zmbb_+$. Asymptotic notations such as $\Ocal, \Theta, \Omega$ have their traditional meaning, and expressions such as $\widetilde{\Ocal}$ hide subdominant polylogarithmic factors compared to those shown. For example, $\widetilde{\Ocal}(x \log y)$ may hide powers of $\log(x)$ or $\log\log(y)$. When $\Ocal$ is cumbersome due to large expressions, we may instead use $\lesssim$ in place of $\Ocal(\cdot)$.

For domain $D \subseteq \Rmbb^d$, a function $f:D \rightarrow \Rmbb$ is said to be \emph{Lipschitz continuous}, or $G$-Lipschitz if there exists a $G \in \Rmbb_+$ such that for all $x, y \in D$
\begin{equation*}
    \abs{f(x) - f(y)} \leq G \norm{x-y}.
\end{equation*}
When $D$ is open, $f$ is differentiable almost everywhere by Rademacher's theorem~\cite{rudin1987real}. Generalizing the definition to complex functions is straightforward.

We will be making use of both vector norms and functional norms, so opt to use concise notation. Suppose $\mathbb{F}$ is either $\Rmbb$ or $\Cmbb$. Let $(\mathcal{X}, \sigma(\mathcal{X}), \mu)$ denote a measure space. Then for measurable functions $f : \mathcal{X} \rightarrow \mathbb{F}$ we have the usual $p$-norm
\begin{equation*}
    \norm{f}_p \coloneqq \left(\int_{\mathcal{X}} \abs{f(x)}^p d\mu \right)^{1/p}
\end{equation*}
for $p \in [1,\infty)$ and
\begin{equation*}
    \norm{f}_\infty \coloneqq \text{ess.} \sup \abs{f(x)},
\end{equation*}
where both of these may be infinite. However, $L_{p}(\mathcal{X})$ will be the quotient space of $p$-norm finite, $\mu$-measurable functions. We will be primary concerned with two cases: (1) $\mu$ is the Lebesgue measure restricted to the sets $S$ or $[0, T]$ (zero elsewhere) with $\mathcal{X} = \mathbb{F}$, or (2) $\mu$ is the counting measure and $\mathcal{X} = [d]$, leading to the usual vector $p$-norms. The context will make it clear which we are referring to. Particularly relevant is the Euclidean norm $\norm{\cdot}_2$, whose subscript will often be omitted. 

The dot product on $\Rmbb^d$ ($\Cmbb^d$) is given by $x\cdot y \coloneqq \sum_i \overline{x}_i y_i$, where the overline denotes the complex conjugate. For real $d$-vectors the shorthand $v^2 \coloneqq v\cdot v$ will often be employed. Given $R \in \Rmbb_+$ we also define the (closed) $p$-ball $B_p(R) \subseteq \Rmbb^d$ of radius $R$ as 
\begin{equation*}
    B_p (R) := \{ x \in \Rmbb^d : \norm{x}_p \leq R \}.
\end{equation*}

When considering differentiable functions $f$ over variables $x = (x_1, x_2, \ldots, x_d) \in \Rmbb^d$, let $\partial_j = \partial_{x_j}$ denote the $j$th partial derivative. Let $\alpha = (\alpha_1, \alpha_2, \ldots, \alpha_d) \in \Nmbb^d$ be a multi-index of order $\norm{\alpha}_1$. Then denote
\begin{equation*}
     \partial^{\alpha} f \coloneqq \partial_1^{\alpha_1}\partial_2^{\alpha_2}\ldots \partial_d^{\alpha_d} f
\end{equation*}
assuming the derivative exists. These derivatives might only exist in the ``weak" sense, depending on the context, but the same notation will be used. Weak derivatives are understood in a distributional sense. If $f$ and $g$ are univariate functions, and for all smooth test functions $\varphi \in C_c^\infty(\Xcal)$ we have
\begin{equation*}
    \int_\Xcal f \varphi' dx = -\int_\Xcal g \varphi dx
\end{equation*}
then $g$ is a weak derivative for $f$. This notion generalizes to multivariate functions.

We also utilize standard notions of quantum oracles which compute continuous functions $f$. In the definitions that follow, $\ket{x}$ represents a computational basis state encoding some $x\in\Xcal$. The number of qubits, and hence precision to which $x$ is encoded in a finite-dimensional basis state, will be specified when necessary.
\begin{definition}[$\epsilon_f$-accurate binary oracle]
\label{defn:binary_oracle}
    Let $f : \mathcal{X} \rightarrow \Rmbb$, where $\mathcal{X} \subseteq \Rmbb^d$. The unitary $ O_{f}$ is an $\epsilon_f$-accurate binary oracle for $f$ if for all computational basis states $\ket{x}\ket{y}$
    \begin{align*}
        O_f\ket{x}\ket{y} = \ket{x}\ket{y\oplus \widetilde{f}(x)},
    \end{align*}
    and $\lVert \widetilde{f}(x) - f(x)\rVert_{\infty} < \epsilon_f$. 
\end{definition}
\begin{definition}[$\epsilon_p$-accurate phase oracle]
\label{defn:phase_oracle}
    For $\Xcal \subseteq \Rmbb^d$, let $f : \mathcal{X} \rightarrow \Rmbb$ be a bounded function with bound $\Lambda \geq \norm{f}_\infty$. The unitary $ O_{f}$ is an $\epsilon_p$-accurate phase oracle for $f$ if for all computational basis states $\ket{x}$,
    \begin{align*}
        O_{f}\ket{x} = e^{i\widetilde{h}(x)}|x \rangle,
    \end{align*}
    where $\lVert \widetilde{h} - f/\Lambda\rVert_{\infty} < \epsilon_p$.
\end{definition}

\subsection{Hilbert Spaces} \label{sec:Hilbert_spaces}
Given a measurable subset $S \subseteq \Rmbb^d$, let $\Hcal= 
L_2(S)$ be the Hilbert space of complex square-integrable functions on $\mathcal{X}$
\begin{equation*}
    L_2(S) \coloneqq \{\psi: S \rightarrow \Cmbb : \int_\mathcal{X} \abs{\psi(x)}^2 dx < \infty\}
\end{equation*}
equipped with the inner product 
\begin{equation*}
    \langle \psi, \phi\rangle \coloneq \int_\mathcal{X} \overline{\psi}(x) \phi(x) dx.
\end{equation*}
This Hilbert space is separable, hence admits a countable basis. Let $\norm{\psi} \coloneqq \sqrt{\langle\psi,\psi\rangle}$ be the induced norm on $\Hcal$ (note the consistency with the $L_p$ norm above). A vector $\psi$ is \emph{normalized} if $\norm{\psi} = 1$. In this case, $\abs{\psi}^2$ represents a probability distribution on $\Xcal$. In this work, we will primarily consider $S = [-1/2,1/2]^d$. Let $C^k(S) \subset L_2(S)$ be the subset of functions with $k$ continuous derivatives on $S$.

Often times $\Hcal$ will not admit enough structure necessary to guarantee well-behaved solutions to the Schr\"odinger equation. For our purposes, \emph{Sobolev spaces} are Hilbert subspaces with sufficient regularity to admit solutions to the differential equations of interest. For each $k\in\Zmbb_+$, let $\Hcal^{k} \coloneqq \{ \phi \in \Hcal : \lVert \phi \rVert_{\Hcal^k} < \infty\}$ denote the $k$th Sobolev space, with norm $\norm{\phi}_{\Hcal^k} \coloneqq \sqrt{\langle\phi,\phi\rangle_{\Hcal^k}}$ induced by the inner product 
\begin{equation*}
    \langle\phi,\psi\rangle_{\Hcal^k} \coloneqq \langle\phi,\psi\rangle + \sum_{\alpha \in \Nmbb^{n} : \norm{\alpha}_1 = k} \langle \partial^{\alpha}\phi, \partial^{\alpha}\psi\rangle.
\end{equation*}
More precisely, $\Hcal^k$ consists of $\phi \in \Hcal$ whose (weak) derivatives up to order $k$ exist and are in $L_2$. In most cases, we will be more interested in the Sobolev seminorm
\begin{align}
\label{eq:real_space_sob_seminorm}
    \lvert \phi\rvert_{\Hcal^k}^2 \coloneqq \sum_{\alpha \in \Nmbb^{n} : \norm{\alpha}_1 = k} \langle \partial^{\alpha}\phi, \partial^{\alpha}\phi\rangle.
\end{align}

For $\Hcal = L_2([-1/2,1/2]^d)$, we can define the notion of Fourier series. Specifically, the plane waves $\chi_{\nmbf}({x}) = e^{i2\pi\nmbf\cdot x}$ with $\nmbf \in \Zmbb^d$, form an orthonormal basis for $\Hcal$. Hence, we can express every element $\phi \in \Hcal$ via its Fourier series
\begin{align} \label{eq:fourier_series}
    \phi(x) = \sum_{\nmbf\in\Zmbb^d} \langle\chi_{\nmbf},\phi\rangle \, \chi_{\nmbf}(x).
\end{align}
Here convergence of the series is in the $L_2$ norm and the partial sums need to be somewhat carefully considered (though $L_2$-convergence allows for less care, and several reasonable choices work). We will be interested in the following subspaces of $\Hcal$ consisting of truncated Fourier polynomials. For
    \begin{equation} \label{eq:truncation_indices}
        \Ncal \coloneqq \{-N, -N+1, \ldots, N-2, N-1\}^d
    \end{equation}
we define a family of subspaces of finite Fourier polynomials parameterized by $N$ by 
\begin{align} \label{eq:truncated_fourier_subspace}
 \Hcal_N \coloneqq \mathrm{span}\{\chi_\nmbf\}_{\nmbf\in\Ncal},
\end{align} with dimension $(2N)^d$. Convergence of~\eqref{eq:fourier_series} may then be understood in the following sense \cite[Theorem 4.1]{weisz2012summability}
\begin{align*}
\lim_{N \rightarrow \infty} \lVert \phi_{N} - \phi\rVert  = 0,
\end{align*}
where
\begin{align*}
    \phi_{N}(x) = \sum_{\nmbf \in \Ncal}\langle\chi_{\nmbf},\phi\rangle \, \chi_{\nmbf}(x).
\end{align*}

It is typically the case that the Hamiltonian dynamics are initially over the Hilbert space $L_2(\Rmbb^d)$. However, Schr\"odinger dynamics are intrinsically local, so whenever an initial state is supported primarily in some bounded region, this support can only "grow so much" during the evolution. For our applications of interest, the situation is that the support of the wavefunction will tend to contract, not expand, due to the dynamics. Thus, for some suitable characteristic size $R > 0$, we can expect to faithfully capture the dynamics on $\Rmbb^d$ by truncating it to a rectangular region $[-R, R]^d$. We then rescale by $R$ to work within the space $\Hcal$ of interest:
\begin{align*}
    g(x) \coloneqq f(2Rx).
\end{align*}
Here $g: S \rightarrow \Rmbb$ is bounded in any $L_p$ norm when $f$ is continuous ($p \in [1,\infty]$). 

In this work, an \emph{unbounded} operator $A$ on a Hilbert space $\Hcal$ means ``not necessarily bounded." The domain of unbounded operators is not necessarily all of $\Hcal$, but for most of the unbounded operators of interest in quantum mechanics, such as the ``position" and ``momentum" operators $\hat{x}$ and $\hat{p}$ on $L_2(\Rmbb)$, the domains are dense in $\Hcal$. The notion of the \emph{adjoint} $A^*$ of a densely-defined operator $A$ is more subtle in the general unbounded setting. Suppose $\phi \in \Hcal$ is such that $\langle \phi, A \cdot \rangle$ is a bounded linear functional into $\Cmbb$. By the Riesz representation theorem, there is a unique $\chi$ such that $\langle \phi, A\cdot\rangle = \langle \chi,\cdot\rangle$. For such $\phi$, we say $\phi \in \Dom(A^*)$ and $A^*\phi \coloneqq \chi$~\cite{hall2013quantum}. An unbounded linear operator $A$ is \emph{symmetric} if 
\begin{align*}
    \langle \phi, A\psi\rangle =  \langle A\phi, \psi\rangle,
\end{align*}
where $\psi, \phi \in \Dom(A)$. Such operators satisfy $\Dom(A) \subseteq \Dom(A^*)$. A symmetric operator is \emph{self-adjoint} if $\Dom(A) = \Dom(A^*)$. Unbounded self-adjoint operators admit a \emph{spectral theorem}, which generalizes the notion of an orthonormal eigenbasis and allows for a functional calculus~\cite{hall2013quantum}. 

In quantum mechanics, self-adjoint operators are referred to as \emph{observables}, as they correspond with experimentally observable quantities in quantum physics. The spectral theorem leads to a rigorous notion of measuring an observable. For any $\psi \in \Dom(A)$, let
\begin{equation*}
    \langle A\rangle_{\psi} \coloneqq \langle \psi, A \psi\rangle
\end{equation*}
be the \emph{expectation value} of $A$ with respect to $\psi$. If the context is clear, then the subscript $\psi$ will be dropped. One class of observables that is of particular interest is the class of \emph{multiplication operators} acting as $[V\phi](x) \coloneqq V(x)\phi(x)$ for some $V:\Rmbb^d \rightarrow \Rmbb$ with suitable regularity assumptions (we use $V$ interchangeably for the operator and function itself). The domain is $\Dom(V) := \{ \phi \in \Hcal | \lVert V\phi\rVert < \infty\}$.

Another important example of an observable for this work is the \emph{kinetic energy} (or Laplacian)
\begin{equation*}
    K = -\Delta,
\end{equation*}
with domain 
\begin{align*}
\Dom( -\Delta) := \{\phi \in L_2(S) | \lVert \phi \rVert_{\Hcal^2} < \infty\} = \Hcal^2.
\end{align*}
The name of the operator derives from its connections to the energy of a free quantum particle in physics. On an unbounded domain such as $\Rmbb^d$, the Laplacian is naturally defined using Fourier transforms. As we will instead be concerned primarily with $L_2(S)$, we define the Laplacian in terms of Fourier series:
\begin{align}
\label{eq:laplacian_defn}
    -\Delta\phi(x) := \sum_{\mathbf{n} \in \Zmbb^d} 4\pi^2\lVert \mathbf{n} \rVert^2 \langle \chi_{\nmbf}, \phi\rangle \chi_{\nmbf}(x), \quad \forall \phi \in \Hcal^2,
\end{align}
where $\Delta$ now acts as a multiplication operator on square-summable sequences $(\langle \phi, \chi_{\nmbf}\rangle)_{\mathbf{n} \in \Zmbb^d}$. One can verify that we also have
\begin{align*}
    \norm{\phi}_{\Hcal^2}^2 = \sum_{\mathbf{n}\in\Zmbb^d} (1 + (2\pi)^4\lVert\mathbf{n}\rVert^4)\lvert \langle \chi_{\nmbf}, \phi\rangle\rvert^2, \qquad \abs{\phi}_{\Hcal^2}^2 = \sum_{\mathbf{n}\in\Zmbb^d} (2\pi)^4\lVert\mathbf{n}\rVert^4\lvert \langle \chi_{\nmbf}, \phi\rangle\rvert^2
\end{align*}
which leads to the reasoning for defining the domain of the Laplacian to be $\Hcal^2$. This Fourier-based definitions of $-\Delta$ is equivalent to one formulated in terms of weak derivatives, and more suitable for our purposes. It is straightforward to verify that $-\Delta$ is self-adjoint on $\Hcal^2$ because the $\chi_\nmbf$ form a true orthonormal eigenbasis for $\Delta$ on $L_2(\Hcal)$, i.e. $-\Delta$ is unitarily equivalent to a self-adjoint, multiplication operator.

A function $\psi:S \rightarrow \Cmbb$ is said to satisfy \emph{periodic boundary conditions} (PBC) if
\begin{equation*}
    \psi\left(\frac12, x_2, \ldots, x_d\right) = \psi\left(-\frac12, x_2, \ldots, x_d\right)
\end{equation*}
and similarly in the other coordinates. A continuous function which satisfies such conditions can be considered continuous on the $d$-torus $\mathbb{T}^d$.

\subsection{Time Dependent Hamiltonian Simulation} \label{sec:time_dependent_Ham_Sim}
    In nature, the motion of closed quantum mechanical systems is governed by the Schr\"odinger equation. Given a quantum state described by a vector $\ket{\psi_0} \in \Hcal$ at time $t = 0$, the evolution of the state over time is given by continuous, unitary evolution generated via the differential equation
    \begin{equation} \label{eq:schrodinger}\tag{Schr\"odinger}
        i\partial_t \ket{\Psi_t} = H(t) \ket{\Psi_t}
    \end{equation}
    where $H(t)$ is a Hermitian operator called the \emph{Hamiltonian}. In the context of quantum computing, and for the purposes of this work, \emph{Hamiltonian simulation} refers to generating a state on a collection of quantum registers which, in some sense, approximates $\ket{\Psi_t}$ for certain observables of interest. Remarkably, Hamiltonian simulation as a technique possesses applications outside of the natural sciences, and forms an important subroutine for generic quantum algorithms~\cite{childs2003exponential, harrow2009quantum}.

    Provided that $H(t)$ is given in some access model on a quantum computer, a direct approach to solving the~\ref{eq:schrodinger} equation is to produce, as a quantum circuit, the \emph{unitary propogator} $U(t,0)$ generated by $H$. This operator is defined as the unique solution to the operator Schr\"odinger equation
    \begin{equation} \label{eq:operator_SE}
        i\partial_t U(t,0) = H(t) U(t,0), \qquad U(0,0) = I,
    \end{equation}
    and the solution can, formally, be written as a so-called time ordered exponential
    \begin{equation*}
        U(t,0) = \exp_\Tcal\left\{-i\int_0^t H(\tau) d\tau\right\}.
    \end{equation*}
    For bounded $H$, there are several equivalent definitions of $\exp_\Tcal$, one of which is via a Dyson series expansion.
    \begin{equation*}
        \exp_\Tcal\left\{-i \int_0^t H(\tau) d\tau \right\} = \mathbbm{1} + \sum_{k=1}^\infty (-i)^k \int_0^t d\tau_1 \int_0^{\tau_1} d\tau_2 \cdots \int_0^{\tau_{k-1}} d\tau_k H(\tau_1) H(\tau_2)\ldots H(\tau_k).  
    \end{equation*}
    Truncating this series and implementing the approximate unitary is one approach to Hamiltonian simulation~\cite{low2019interaction}, and the one we employ in this work.

    An important feature of Hamiltonian simulation is invariance under time reparametrization, which we will take advantage of in several different contexts throughout this work. Let $\ket{\psi_t}$ the the state generated under Hamiltonian $H(t)$, and consider a new state $\ket{\phi_s} = \ket{\psi_{\xi(t)}}$, where $\xi:[0,T]\rightarrow[0,T']$ is a differentiable bijective function with $\xi(0) = 0$ and $\xi(T) = T'$. Then $\ket{\phi_s}$ is evolved by the Hamiltonian $\widetilde{H}(s)$, where
    \begin{equation} \label{eq:time_reparametrization}
        \widetilde{H}(s) = \frac{d\xi^{-1}(s)}{ds} H(\xi^{-1}(s)).
    \end{equation}
    Thus, up to change of variables and rescaling of time, the dynamics remain unchanged.

\section{Simulating Schr\"odinger Dynamics in Euclidean Space}
\label{sec:simulating-schrodinger-dynamics}

One of the most fundamental and natural quantum simulation problems considers the time evolution of wavefunctions $\Phi(x,t)$ on a Euclidean space, evolved under a kinetic-plus-potential type Hamiltonian
\begin{equation}\label{eq:real_schrodinger}\tag{Real Space Schr\"odinger}
    i \partial_t \Phi(x, t) = [- a(t) \Delta + V(x,t)]\Phi(x, t),
\end{equation}
from some initial state $\Phi(x,0) = \Phi_0(x)$. Here $a(t)$ is a Lesbesgue-measurable, positive function of $t$ (often a constant) and $V(x,t)$ is a multiplicative operator called the \emph{potential}. We frequently refer to $t$ as ``time." The time-dependent vector $\Phi(\cdot, t)$ is called the \emph{wave function}. 
        
In this section we provide, to the authors' knowledge, the first rigorous performance guarantees for quantum simulation of Schr\"odinger equations in arbitrary-dimensional Euclidean space. Specifically, we determine the level of discretization required for a digital quantum simulation algorithm to output a state $\ket{\Psi_t}$ that is, in an appropriate continuous representation, close in $L_2$ distance to $\Phi$, and also that certain observables on the continuous space can be well-approximated with the discrete state. These two conditions imply that the digital output of the algorithm is a sufficient proxy for the real solution on the continuum and for further downstream algorithmic processing.    

Our work is by no means the first to investigate the digital simulation of continuous quantum systems. It has long been recognized that Trotterization, in conjunction with the Quantum Fourier Transform, provides a direct means for simulating kinetic-plus-potential Hamiltonians on a quantum computer~\cite{nielsen2002quantum}. Indeed, our discretization is essentially the same, somewhat obvious choice. Significantly more challenging is providing rigorous error bounds on discretization, and this precludes a true complexity analysis. Recently, Childs et al.~\cite{childs2022quantumsim} provided a major step in characterizing real-space Hamiltonian simulation, considering several leading simulation algorithms for a ``pseudo-spectral" framework. However, these performance bounds are provided in terms of the \emph{solution itself}, which cannot be known a priori. More crucially, the authors consider only errors resulting in discretization of the state, and appear to ignore that discretization also affects the \emph{dynamics} in a way that must, in principle, be accounted for. 

The consequences of neglecting these error sources are significant, as they may lead to a distorted view of how quantum algorithms perform when using Euclidean Hamiltonian simulation as a subroutine. For example, as we discuss further in Section~\ref{sec:Quantum_Simulation_Algorithms_Convex_Optimization}, Hamiltonian simulation has been proposed as a method for unconstrained continuous optimization. In order to facilitate sturdy comparisons with existing quantum and classical algorithms for the task, and to make unequivocal statements on optimality, a precise estimate of resources required for controlling discretization errors is essential. In Section \ref{sec:comparison}, we provide additional details comparing existing results on real-space quantum simulation to ours.
        
We now begin to establish the setting of our results. We primarily consider separable potentials of the form
\begin{equation} \label{eq:separable_potential}
    V(x,t) = b(t) g(x)
\end{equation}
where $b:[0,\infty)\rightarrow \Rmbb_{\geq 0}$ is Lesbesgue-integrable, though some aspects of our simulation algorithm apply to more general $V$. Here $g$ is a natural restriction of some Lipschitz function $f$ to a hypercube, which we term an $R$-restriction.
\begin{definition}[$R$-restriction] \label{defn:R_restriction}
    Suppose $f : \mathcal{X} \rightarrow \Rmbb$ is Lipschitz continuous, with $x_0 + [-R,R]^d\subseteq \mathcal{X} \subseteq \Rmbb^d$ for some $R>0$. We say that $g:S \rightarrow \Rmbb$ defined by $g(x) \coloneqq f(2Rx + x_0)$ is the $R$-restriction of $f$.
\end{definition}
\noindent For simplicity, we will typically just set $x_0 =0$ by a change of coordinates, which does not affect any results. The separability assumption enables efficient simulation while simplifying analysis, and is more relative to our applications of interest.

We now formalize the above~\ref{eq:real_schrodinger} equation into a fully-defined boundary value problem.
\begin{problem}[$R$-restricted Schr\"odinger Problem]
\label{prob:restricted_schr}
Let $f : \mathcal{X} \rightarrow \Rmbb$ be Lipschitz continuous, with $[-R,R]^d \subseteq \mathcal{X} \subseteq \Rmbb^d$ for some $R > 0$. Let $g$ be an $R$-restriction of $f$. Suppose $a,b:[0,T]\rightarrow \Rmbb$ are continuously differentiable. We define the $R$-restricted Schr\"odinger Problem to be following 
boundary-value problem
\begin{align*}
    i \partial_t \Phi(x, t) &= [- a(t) \Delta + b(t)g(x)]\Phi(x, t), \\
    \Phi(x, 0) &= \Phi_0(x) \in \Hcal^2, \lVert \Phi_0 \rVert = 1
\end{align*}
for all $t\in[0,T]$ and $x\in S$.
\end{problem}
\noindent By specifying from the onset that the domain of interest is bounded (the hypercube), we tacitly neglect any errors resulting from the restriction of $f$, assuming one is actually interested in larger domains such as $\Rmbb^d$. This could be dealt with by additional assumptions of probability leakage. Regardless, this is not a problem for our applications of interest, and so  does not sacrifice any rigor. 

Additionally, we have the following result on the existence of solutions to Problem \ref{prob:restricted_schr}.
\begin{theorem}
\label{thm:solution_existence}
If $\Phi_0 \in \Hcal^2$ then Problem~\ref{prob:restricted_schr} has a strongly (time-) differentiable solution, i.e. there is a unique $\Phi(t)$ that satisfies the PDE in the classical sense. In fact, for any $s\geq 2$, $\Phi(t) \in \Hcal^s$ whenever $\Phi_0$ is.
\end{theorem}
\noindent See Appendix \ref{subapp:existence_of_dynamics} for a proof, which follows from standard results in mathematical physics.

Our main result of this section will provide sufficient conditions on the number of grid points per dimension, $N$, that suffices to ensure $\epsilon$ error on Problem \ref{prob:restricted_schr}. This result will also include the number of queries to a noisy-oracle for $f$ and the overall gate complexity and qubit count. To be able to provide guarantees on the performance of the simulation algorithm, we will only require the potential to be $G$-Lipschitz in some $p$-norm for $p\geq 1$ and satisfy periodic boundary conditions. There is also a version of the result without assuming PBC (Theorem \ref{thm:collocation_trace_distance_aperiodic} in Section \ref{subsec:error_analysis_hamiltonian}), and some with alternative regularity assumptions in Appendix \ref{app:improved_regularity}. Note that prior results on real-space simulation had complexities that depended on parameters that were not bounded. 

Before stating our main result, we introduce the notion of Fourier interpolation to quantify the accuracy of our discrete approximation compared to the continuous solution. A rigorous definition is given in Definition~\ref{defn:fourier_interpolant}, but an informal version is easy to state. Given an approximation $\ket{\Psi_t} = (\Psi(x_\jmbf, t))_\jmbf$ of the exact wavefunction at $(2N)^d$ grid points $x_\jmbf$, the Fourier interpolation $\Psi(t)$ is the unique low-degree Fourier polynomial, i.e. in $\Hcal_N$, which matches $\ket{\Psi_t}$ at these grid points. 

Having established the necessary terminology, we are ready to state the main result of this section which provides rigorously-proven guarantees on the performance of the real-space simulation algorithm (Algorithm \ref{alg:sim_schro}) for the \ref{eq:real_schrodinger} equation formulated as Problem \ref{prob:restricted_schr}, under mild assumptions on the potential.

\begin{algorithm}  
    \caption{Euclidean Hamiltonian Simulation via Truncated Dyson Series} \label{alg:sim_schro}
    \hspace*{\algorithmicindent}
    \textbf{Input}: Desired evolution time $T$; precision $\epsilon$; dimension $d$; upper bound $\Lambda \geq \norm{g}_\infty$; scaling schedules $a(t), b(t)$ and their integrals given as explicitly computable functions; radius of simulation $R$; $\epsilon_f$-accurate quantum oracle $O_f$ for real function $f$ on $[-R,R]^d$ with $\epsilon_f$ obeying the bound in Thm~\ref{thm:master_simulation_thm}; state preparation oracle $U_0$ on discretization of $[-R,R]^d$.
    
    \hspace*{\algorithmicindent}
    \textbf{Output}: With probability $1-\Ocal(\epsilon)$, a quantum state $\ket{\Psi_T}$ whose Fourier interpolation (see Definition \ref{defn:fourier_interpolant}) $\epsilon$-approximates solution $\Phi(T)$ of Problem~\ref{prob:restricted_schr}.

    \hspace*{\algorithmicindent}
    \textbf{Procedure}:
    \begin{algorithmic}[1]
        \State Set discretization $N \leftarrow N(\epsilon, d, T, a, b, R)$ using \eqref{eq:N-master-thm} in Thm~\ref{thm:master_simulation_thm}.
        \State Allocate $d \lceil\log(2N)\rceil$ qubits for main quantum register, and auxiliary registers per Thm~\ref{thm:master_simulation_thm}.
        \State Initialize $\ket{\Psi_0} = U_0 \ket{0}^{\otimes d}$ on main register.
        \State Define subroutine $\HamT(\widetilde{H}_I) \leftarrow \text{HAMT}(a, b, O_f, \Lambda, \epsilon)$ using Lemma~\ref{lem:HamT_construction}.
        \State Set $T' \leftarrow \norm{b}_1$.
        \State Set $\ket{\widetilde{\Psi}_T} \leftarrow  \text{TruncDyson}(\HamT(\widetilde{H}_I), \ket{\Psi_0}, T', \epsilon)$ (see~\cite{low2019interaction}).
        \State \textbf{Return} $\ket{\Psi_T} \leftarrow U_K^\dagger(T) \ket{\widetilde{\Psi}_T}$. 
    \end{algorithmic}
\end{algorithm}
\begin{theorem}[Digital Simulation of Schr\"odinger Equations]
\label{thm:master_simulation_thm}
Suppose $f : \mathcal{X} \rightarrow \Rmbb$ is a $G$-Lipschitz continuous function in $p$-norm with $\mathcal{X} \subseteq \Rmbb^d$ and $p \geq 1$. Let $g:S\rightarrow\Rmbb$ be the $R$-restriction of $f$ (Definition~\ref{defn:R_restriction}), and let $\Lambda \geq \lVert g \rVert_{\infty}$ be a promised upper bound. Furthermore, suppose that $g$ satisfies periodic boundary conditions. Choose $N$ such that
\begin{equation}\label{eq:N-master-thm}
    \log_2(N) = \Ocal\left(d\log \bigg( \frac{\lVert a \rVert_1^{1/d}\big[\max\big\{\lvert\Phi_0 \rvert_{\Hcal^d}, \lvert\Phi_0\rvert_{\Hcal^{3}}\big\}^{1/d}  +\lVert b \rVert_1GRd\big]}{\epsilon}\bigg)\right),
\end{equation} 
and suppose that the initial wavefunction satisfies $\Phi_0 \in \Hcal_N$ (defined in \eqref{eq:truncated_fourier_subspace}). Then there exists a quantum algorithm (Algorithm~\ref{alg:sim_schro}) that solves\footnote{There is one technicality that we also require that the true solution have a periodic extension that is in $C^{1}(\Rmbb^d)$. The main issue being ensuring that we have continuous-differentiability across the boundary. Given that we can ensure arbitrarily-high regularity on the true solution by tuning our initial state, this assumption is believable. Additionally, we require t $\rightarrow \partial^{\alpha}_x\Phi(x, t)$ to be continuous $\forall x \in S$ and $\alpha \in \Nmbb^{d}$ for which the derivatives exist, which is believable given the regularity of $a(t)$ and $b(t)$.} Problem \ref{prob:restricted_schr} to precision $\epsilon$, in the following sense. With probability $1 - \Ocal(\epsilon)$ the algorithm outputs a $d \lceil\log_2(2N)\rceil$-qubit quantum state $|\Psi_T\rangle$ of the form
\begin{align*}
    |\Psi_T\rangle = \sum_{x_\jmbf \in \Gcal}(2N)^{-d/2}\Psi_T(x_\jmbf)|\jmbf\rangle
\end{align*}
where $x_\jmbf\in \Gcal$ are uniform grid points on $S = [-1/2,1/2]^d$, and $\abs{\Gcal} = (2N)^d$. The state $\ket{\Psi_T}$ has the following guarantees:
\begin{itemize}
    \item Let $\Psi(T) : S \rightarrow \Cmbb$ be the Fourier interpolant (Definition \ref{defn:fourier_interpolant}) of $|\Psi_{T}\rangle$. Then,
    \begin{align*}
        \lVert \Psi(T) - \Phi(T) \rVert = \Ocal(\epsilon).
    \end{align*}
    \item Let $h : \mathcal{X}\rightarrow \Rmbb$ be any $G$-Lipschitz function \footnote{Note this result is from Theorem \ref{thm:energy_error_bound}, where the guarantee is in Big-O, here we absorb the constant into $\epsilon_{\mathrm{alg}}$. If $h$ has a Lipschitz constant other than $G$, then we just replace $G$ in \eqref{eq:N-master-thm} with the max of the Lipschitz constant of $h$ and the potential $g$.}. Then,
    \begin{align*}
        \bigg\lvert (2N)^{-d}\sum_{x_\jmbf \in \Gcal} {h}(2Rx_\jmbf)\lvert \Psi_{T}(x_\jmbf)\rvert^2 - \int_S{h}(2Ry)\lvert \Phi(y, T)\rvert^2 dy \bigg\rvert \leq \lVert h \rVert_{\infty}\epsilon.
    \end{align*}
\end{itemize}
Finally, the algorithm uses the following resources:
\begin{itemize}
    \item \textbf{Queries:} $Q_f$ calls to an $\epsilon_f$-accurate quantum binary oracle or an $(\epsilon_f/\Lambda)$-accurate phase oracle, where 
    \begin{equation*}
        \epsilon_f^{-1} = \Ocal\left(\frac{\norm{b}_1}{\epsilon} \frac{\log(\Lambda \norm{b}_1/\epsilon)}{\log\log(\Lambda \norm{b}_1/\epsilon)}\right).
    \end{equation*}
        \begin{align}
           Q_f =  \begin{cases}
                 \Ocal\left(\Lambda \norm{b}_1 \frac{\log(\Lambda \norm{b}_1/\epsilon)}{\log\log(\Lambda \norm{b}_1/\epsilon)}\right) & \text{(binary oracle)} \\
                {\Ocal}\left(\frac{\Lambda^2 \norm{b}_1^2}{\epsilon} \frac{\log^3(\Lambda \norm{b}_1/\epsilon)}{\log^2\log(\Lambda \norm{b}_1/\epsilon)}\right) & \text{(phase oracle)} 
            \end{cases}
        \end{align}
        \item \textbf{Qubits}: $d\lceil\log_2 N\rceil +\Ocal\left(\log dN^2+\log(\norm{a}_1)+ \log(\Lambda\lVert b\rVert_1) + \log(1/\epsilon)\right)$
        \item \textbf{Gates}: 
        \begin{equation}
            \begin{aligned}
            \Ocal&\Bigg(\frac{\Lambda^2\norm{b}_1^2}{\epsilon_\mathrm{alg}} \frac{\log^3(\Lambda\norm{b}_1/\epsilon_\mathrm{alg})}{\log^2\log(\Lambda\norm{b}_1/\epsilon_\mathrm{alg})} \\
            &\times\left[d\log N\log dN^2 + d\log N\log\left(\frac{1}{\epsilon_\mathrm{alg}}\left[\norm{a}_1 + \frac{\Lambda \norm{b}_1}{dN^{2}}\right]\right) + \log^3\left(\frac{\Lambda\norm{b}_1}{\epsilon_\mathrm{alg}}\right)\right]\Bigg).
        \end{aligned}
        \end{equation}
    \end{itemize}
\end{theorem}
We note that we can easily remove  the assumption of $\Phi_0 \in \Hcal_N$ and instead consider general $\Phi_0 \in \Hcal^2$. However, there may be a slight qubit overhead, which is handled by Lemma \ref{lem:gen_periodic_initial_condition}. This will be handled by considering a proxy $\widetilde{\Phi}_0 \in \Hcal_N$ for $\Phi_0 \in \Hcal$, such that the two wavefunctions are $\Ocal(\epsilon)$ close in $L_2$ distance. Hence, the time-evolved wavefunctions remain $\Ocal(\epsilon)$ close by unitarity. 

One will also note that we have shown that the query complexity of Algorithm~\ref{alg:sim_schro} does not depend on the amount of grid spacing, like it did in \cite{catli2025exponentially}. This is particularly important for applications to optimization, where we do not want to pay additional dimension factors.

For many applications, we are interested in measuring certain observables on the solution to \ref{eq:real_schrodinger} equation. Theorem \ref{thm:master_simulation_thm} shows that observables measured on the outputted digitized wavefunction are close to those measured on the true continuous solution. Thus we provide a complete discretization analysis of the simulation algorithm. 

The rest of this section is dedicated to proving Theorem \ref{thm:master_simulation_thm} and some extensions. We first introduce our general approach, the \emph{pseudo-spectral} method, in Section~\ref{subsec:pseudo-spectral_methods}. In Section~\ref{subsec:simulation_algorithm}, we develop this approach into a concrete Hamiltonian simulation algorithm based on the Dyson series method, and quantify its costs. Finally, in Section~\ref{subsec:error_analysis_hamiltonian}, we deal with the difficult question of discretization errors resulting from the pseudo-spectral method.
\subsection{The Pseudo-Spectral Method} \label{subsec:pseudo-spectral_methods}
    \emph{Spectral methods} are a broad family of methods for solving partial differential equations which consider the problem in terms of the natural eigenfunctions of the operators involved. Specific truncation schemes are used to provide a finite representation, which can then be worked with numerically. 
    
    As written, the~\ref{eq:real_schrodinger} is not directly amenable to numerical calculation. Spectral methods begin by imposing some truncation scheme on the solution, such as 
    \begin{equation} \label{eq:pseudo-spectral_Phi}
        \Psi(x,t) = \sum_{\nmbf\in\Ncal} \widehat{\Psi}_\nmbf(t) \chi_\nmbf(x)
    \end{equation}
    and stipulating some conditions on the modified coefficients $\widehat{\Psi}_\nmbf$ so that $\Psi$ approximates $\Phi$ for the observable of interest. Here $\Ncal$ is as defined in Eq.~\eqref{eq:truncation_indices} and
    characterizes Fourier modes of low wavenumber $\nmbf$. The reason for leaving out $N$ but not $-N$ above will become clearer shortly. 

    A spectral method effectively uses the space  $\Hcal_{N}$ \eqref{eq:truncated_fourier_subspace} as a solution ansatz, and computes the Fourier coefficients that minimize the error in solving the \ref{eq:real_schrodinger} equation. For example, the \emph{Galerkin} spectral methods seeks to find $\Psi(t) \in \Hcal_N$ such that
    \begin{align*}
        \langle \chi_\nmbf, [i\partial_t - H(t)]\Psi(t)\rangle = 0,\qquad \forall \nmbf \in \Zmbb^d\setminus\Ncal.
    \end{align*}
    Unfortunately, such a computation involves high-dimensional integrals, i.e. inner products $\langle \cdot, \cdot \rangle$, and transiitons to Fourier modes outside of $\Ncal$ due to the potential $V$. There may be settings where we have sufficient knowledge about the PDE to obtain closed-form expressions; however, in our case the problem instance will be provided by black-box query access, thus preventing this approach. In contrast to spectral methods, \emph{pseudo}-spectral methods replace $\langle \cdot, \cdot \rangle$ with numerical integration to deal with the aforementioned issues.
    
    One popular and practical pseudo-spectral method is the so-called \emph{collocation} method, which seeks to replicate the Schr\"odinger equation exactly on a predefined set of grid points $\Gcal \subset S$. More precisely, define the \emph{residual} of $\Phi(x,t)$ as 
    \begin{equation*}
        Q_\Phi (x,t) \coloneqq [(i\partial_t - H(t))\Phi](x,t).
    \end{equation*}
    Then the collocation method seeks a $\Psi\in\Hcal_N$ such that the residual is zero on $\Gcal$:  
    \begin{equation} \label{eq:zero_residue}
        Q_{\Psi}(x_\jmbf, t) = 0 \qquad \forall x_\jmbf \in \Gcal.
    \end{equation} 
    For the pseudo-spectral method on the box $[-\frac12,\frac12]^d$, there is a very natural grid $\Gcal$ to consider, which also happens to be the simplest one. If we take a uniform rectangular grid $\Gcal = \{x_\jmbf\}$ such that
    \begin{equation} \label{eq:def_grid}
        x_\jmbf = \jmbf/2N, \qquad \jmbf \in \Ncal
    \end{equation}
    and define the sesquilinear form
    \begin{equation} \label{eq:discrete_inner_product}
        \inner{\psi, \phi}_{\Gcal} \coloneqq \frac{1}{(2N)^d} \sum_{\jmbf\in \Ncal} \overline{\psi}(x_\jmbf) \phi(x_\jmbf),
    \end{equation}
    then $\langle\cdot,\cdot\rangle_\Gcal$ is in fact an inner product on $\Hcal_N$, and equals $\langle\cdot,\cdot\rangle$ on this subspace. Note that~\eqref{eq:discrete_inner_product} is essentially a Riemann approximation to the integral given by $\langle\cdot,\cdot\rangle$. Moreover, $\langle\cdot,\cdot\rangle_\Gcal$ extends naturally to a symmetric, nonnegative sesquilinear form on $\Hcal$; however, it fails to be a true inner product since there exists nonzero $\psi\in\Hcal$ for which $\langle\psi,\psi\rangle_\Gcal = 0$. $\Hcal_N$ is the maximal subspace on which $\langle\cdot,\cdot\rangle$ and $\langle\cdot,\cdot\rangle_\Gcal$ agree (and for which $\langle\cdot,\cdot\rangle_\Gcal$ is an inner product) among all subspaces containing the constant function (i.e., among ``low-frequency" subspaces). In what follows we let $\norm{\cdot}_\Gcal$ refer to the seminorm induced by $\langle\cdot,\cdot\rangle_\Gcal$.

    It turns out that there is a relatively natural Schr\"odinger equation which satisfies these matching conditions on the grid.
    \begin{lemma} \label{lem:FT_pseudo_equations}
        Let $H(t)$ be the Hamiltonian in the~\ref{eq:real_schrodinger} equation, and $\Gcal \subset S$ be the uniform grid of equation~\eqref{eq:def_grid}. Then the pseudo-spectral wavefunction $\Psi(\cdot,t)\in \Hcal_N$ in~\eqref{eq:pseudo-spectral_Phi} has zero residual on $\Gcal$ if and only if the Fourier coefficients $\widehat{\Psi}_\mmbf(t)$ for $\mmbf\in\Ncal$ satisfy the following Schr\"odinger system of ODEs
        \begin{equation} \label{eq:collocation_Schrodinger}
            i \partial_t \widehat{\Psi}_\mmbf(t) = a(t) (2\pi \mmbf)^2 \widehat{\Psi}_\mmbf(t) +  \sum_{\nmbf \in \Ncal} W_{\mmbf\nmbf}(t) \widehat{\Psi}_\nmbf(t),
        \end{equation}
        where
        \begin{equation*}
            W_{\mmbf\nmbf}(t) \coloneqq \inner{\chi_\mmbf, V(t) \chi_\nmbf}_{\Gcal} = \frac{1}{(2N)^d} \sum_{\jmbf\in \Ncal} \overline{\chi}_\mmbf(x_\jmbf) V(x_\jmbf,t) \chi_\nmbf(x_\jmbf).
        \end{equation*}
    \end{lemma}
    \begin{proof}
        Substituting~\eqref{eq:pseudo-spectral_Phi} into the zero-residue condition~\eqref{eq:zero_residue}, and moving the Hamiltonian term $H(t)\Psi$ to the right hand side, we see that the residue condition holds if and only if
        \begin{equation*}
            \sum_{\nmbf \in \Ncal} i \partial_t \widehat{\Psi}_\nmbf(t) \chi_\nmbf(x_\jmbf) = \sum_{\nmbf \in \Ncal} a(t) 4\pi^2 \nmbf^2 \widehat{\Psi}_\nmbf(t) \chi_\nmbf(x_\jmbf) + \sum_{\nmbf \in \Ncal} V(x_\jmbf,t) \widehat{\Psi}_\nmbf(t) \chi_\nmbf(x_\jmbf)
        \end{equation*}
        for all $x_\jmbf \in \Gcal$. Taking the discrete inner product~\eqref{eq:discrete_inner_product} with some $\chi_\mmbf$ for some $\mmbf \in \Ncal$ and invoking orthonormality, this is true if and only if, for all such $\mmbf$
        \begin{equation*}
            i \partial_t \widehat{\Psi}_\mmbf(t)  = a(t) (2\pi \mmbf)^2 \widehat{\Psi}_\mmbf(t) + \sum_{\nmbf \in \Ncal} \widehat{\Psi}_\nmbf(t) W_{\mmbf\nmbf}(t)
        \end{equation*}
        where $W_{\mathbf{mn}}$ is as defined in the lemma statement.
    \end{proof}
    Further insight is gained by explicitly evaluating $W$ at the grid points $x_\jmbf$:
    \begin{equation}\label{eq:W_hat}
        W_{\mmbf\nmbf}(t) = \frac{1}{(2N)^d} \sum_{\jmbf\in\Ncal} e^{-2\pi i \mmbf\cdot \jmbf/(2N)} V(x_\jmbf, t) e^{2\pi i \nmbf\cdot \jmbf/(2N)}.
    \end{equation}
    Upon inspection, \eqref{eq:W_hat} is directly related to the \emph{discrete} Fourier transform on $d$ registers of length $2N$. To represent this discrete space we switch to braket notation. Consider the tensor product space $(\Cmbb^{2N})^{\otimes d}$ and let \begin{equation}
        \ket{\jmbf} \coloneqq \ket{j_1}\ket{j_2}\ldots \ket{j_d}
    \end{equation}
    with standard basis vectors $\ket{j_i}$ and implicit tensor product. Let $V(t)$ be the diagonal operator defined by
    \begin{equation}
        V(t)\ket{\jmbf} \coloneqq V(x_\jmbf,t) \ket{\jmbf}.
    \end{equation}
    Finally, let $\ket{\widehat{\Psi}_t} = \sum_\jmbf \widehat{\Psi}_\jmbf(t) \ket{\jmbf}$. Then Lemma~\ref{lem:FT_pseudo_equations} can be recast as
    \begin{equation} \label{lem:braket_c_Schrodinger}
        i\partial_t \ket{\widehat{\Psi}_t} = a(t) (2\pi \hat{\nmbf})^2 \ket{\widehat{\Psi}_t} + \QFT^{^\dagger \otimes d} V(t) \QFT^{\otimes d} \ket{\widehat{\Psi}_t}
    \end{equation}
    where $\QFT$ is the quantum (discrete) Fourier transform on $2N$ variables
    \begin{equation}
        \QFT \ket{j} \coloneqq \frac{1}{\sqrt{2N}} \sum_{k = -N}^{N-1} e^{2\pi i j k/(2N)} \ket{k}
    \end{equation}
    and $\hat{\nmbf}^2 \ket{\nmbf} \coloneqq \nmbf^2 \ket{\nmbf}$. Transforming the problem through a $\QFT$ leads to the following consequence of the above.
    \begin{corollary} \label{coro:Our_actual_H}
        Let
        \begin{equation}
        \label{eq:discrete_quantum_state}
            \ket{\Psi_t} \coloneqq \QFT^{\otimes d}\ket{\widehat{\Psi}_t} = \frac{1}{\sqrt{(2N)^d}}\sum_{\jmbf \in \Ncal} \Psi_\jmbf(t) \ket{\jmbf}.
        \end{equation}
        Then $\ket{\Psi_t}$ satisfies $i\partial_t \ket{\Psi_t} = H(t) \ket{\Psi_t}$, where
        \begin{equation} \label{eq:ham-pseudospec}
             H(t) \coloneqq a(t) (2\pi)^2 \QFT^{\otimes d}\hat{\nmbf}^2 \QFT^{\dagger \otimes d} + V(t).
        \end{equation}
        Moreover, $\Psi_\jmbf (t) = \Psi(x_\jmbf, t)$ for all $\jmbf \in \Ncal$.
    \end{corollary}
    \begin{proof}
        The evolution equation follows almost immediately from transforming the Schr\"odinger equation~\eqref{lem:braket_c_Schrodinger} using the $\QFT$. As for the second claim, we start by evaluating $\Psi(x_\jmbf, t)$.
        \begin{align}
        \begin{aligned}
            \Psi(x_\jmbf, t) &= \sum_{\nmbf \in \Ncal} \widehat{\Psi}_\nmbf (t) \chi_\nmbf(x_\jmbf) \\
            &= \sum_{\nmbf \in \Ncal} e^{2\pi i \nmbf\cdot\jmbf/(2N)} \widehat{\Psi}_\nmbf (t) \\
            &= \sqrt{(2N)^d} \,\bra{\jmbf}\QFT^{\otimes d}\ket{\widehat{\Psi}_t} = \sqrt{(2N)^d} \langle\jmbf\vert\Psi_t\rangle.
        \end{aligned}
        \end{align}
        Some minor rearranging then gives the result.
    \end{proof}
    \noindent To summarize, the pseudo-spectral wavefunction, represented by $\ket{\Psi_t}$, is generated by the dynamics of $H(t)$, which in some ways is the most obvious discretization of the real-space Hamiltonian from a Fourier perspective. 

    We strongly emphasize that the Fourier coefficients of the exact solution $\Phi$ and the pseudo-spectral approximation $\Psi$ are not related via closed-form transformations, such as Fourier truncation or interpolation. Specifically, $\Psi$ and $\Phi$ will not even match on $\Gcal$ unless $\Phi$ has a finite Fourier series, which is essentially never the case for nonconstant $V$. Thus, it is a nontrivial task to bound the $L_2$ distance between $\Phi$ and $\Psi$ based on the number of grid points. We delay consideration of this important analysis to Section~\ref{subsec:error_analysis_hamiltonian}, and for the moment move on to consider the algorithm in the discretized setting.

\subsection{Simulation Algorithm} \label{subsec:simulation_algorithm}
In Corollary~\ref{coro:Our_actual_H}, we arrived at a pseudo-spectral Hamiltonian of the form 
\begin{equation} \label{eq:colloc_Ham}
    H(t) = K(t) + V(t)
\end{equation}
where
\begin{equation*}
    K(t) \coloneqq a(t) \QFT^{\otimes d}(2\pi \hat{\nmbf})^2 \QFT^{\dagger \otimes d}
\end{equation*}
is the discretized ``kinetic" operator, diagonal in the discrete Fourier basis, and 
\begin{equation*}
    V(t) \coloneqq \sum_{\jmbf\in\Ncal} V(x_\jmbf,t) \ket{\jmbf}\bra{\jmbf}
\end{equation*}
is the ``potential" term, diagonal in the computational basis. The simulation of~\eqref{eq:colloc_Ham} can be performed through a number of existing methods, such as product formulas~\cite{wiebe2011simulating}, Dyson series~\cite{kieferova2019dyson}, multiproduct formulas, or qubitization~\cite{watkins2024time}. For this work, we will ultimately use a so-called ``rescaled" Dyson series~\cite{berry2020time}, but first we perform some preconditioning by moving into the interaction picture (IP). Simulation in the IP is advantageous when $H$ is a sum of two terms and one (or both) of these terms has large norm but is fast-forwardable on its own~\cite{low2019interaction}. In principle, one could work either in the frame of $K$ or $V$. We find it advantageous to work in the $K$-frame, because it turns out that this choice decouples the query complexity of $f$ from the discretization parameter $N$. This is beneficial for a number of reasons, one of which being it is very difficult to give a precise analysis of the errors from discretization. 

Let $U_K$ stand for the unitary propagators under $K(t)$. Given a Schr\"odinger wavefunction $\ket{\Psi_t}$ generated by $H(t)$, we define the $\ket{\Psi^K_t}$ in the IP as
\begin{equation}
    \ket{\Psi_t^K} \coloneqq U_K^\dagger(t,0) \ket{\Psi_t}.
\end{equation}
From here on we drop the $0$ argument to $U_K$. It can be checked that
\begin{equation} \label{eq:IP_Schrodinger}
    i \partial_t \ket{\Psi^K_t} = H_I(t) \ket{\Psi^K_t}
\end{equation}
where
\begin{equation} \label{eq:unscaled_IP_Ham}
    H_I(t) \coloneqq U_K^\dagger(t) V(t) U_K(t).
\end{equation}
By solving~\eqref{eq:IP_Schrodinger}, the original solution of interest $\ket{\Psi_t}$ can be obtained by applying $U_K(t)$, which generally presumed efficient and is indeed so in our case. In fact,
\begin{align} \label{eq:kinetic_diag_discrete}
    U_K(t) &= \QFT^{\otimes d} \exp\left\{-i (2\pi\hat{\nmbf})^2 \int_0^t a(\tau) d\tau \right\} \QFT^{\dagger \otimes d}
\end{align}
which can be implemented using $d$ $\QFT$s and inverses, along with standard techniques for diagonal phase operations. In many cases of interest, such as the applications to optimization we consider in Section~\ref{sec:Quantum_Simulation_Algorithms_Convex_Optimization}, $a(t)$ will be a simple function and thus the integral is immediate. From here on we neglect approximation errors and computational resources associated with computation of $\int_0^t a(\tau) d\tau$ in quantum phases.

The Hamiltonian~\eqref{eq:unscaled_IP_Ham} is almost the one we wish to simulate, but first we perform time reparametrization (see Section~\ref{sec:time_dependent_Ham_Sim} for a brief discussion) as another step of preconditioning. It was recognized in~\cite{berry2020time} that, for reparametrizations that fix the ``size" of $H$ over time, the cost of simulating time dependent Hamiltonians depends on the time-averaged Hamiltonian size rather than the largest size in the simulation. This relies on knowledge of the appropriate reparametrization. In our specific context, where the potential $V(t,x) = b(t) g(x)$ is separable, this comes almost for free. Let $\Lambda \geq \norm{g}_\infty$ be an upper bound on $g$ in the domain of simulation (assuming bounded $g$), and define the mapping
\begin{equation*}
    t \mapsto \xi(t) = \int_0^t \Lambda b(\tau) d\tau = \Lambda B(t)
\end{equation*}
where $B:[0,T]\rightarrow[0,\norm{b}_1]$ is differentiable and bijective. Then the reparameterized Hamiltonian $\widetilde{H}_I$ is given by 
\begin{equation}
    \widetilde{H}_I(t) = \frac{1}{\Lambda b(\xi^{-1}(t))} H_I(\xi^{-1}(t)) = U^\dagger_K(\xi^{-1}(t)) \frac{G}{\Lambda} U_K(\xi^{-1}(t))
\end{equation}
where $G = \sum_{\jmbf\in\Ncal} g(x_\jmbf)\ket{\jmbf}\bra{\jmbf}$ is the discretized potential. Whereas in~\cite{berry2020time} the cost of calculating reparametrizations was considered via oracles, in our case this is essentially unnecessary. We assume that $\xi^{-1}(t)$ can be precomputed at negligible cost, which is typically the case in our applications of interest (see Section~\ref{sec:Quantum_Simulation_Algorithms_Convex_Optimization}). We thus write 
\begin{equation} \label{eq:interaction_Ham}
    \widetilde{H}_I(t) = U_K^\dagger(t) (G/\Lambda) U_K(t)
\end{equation}
with the redefinition $U_K(t) \mapsto U_K(\xi^{-1}(t))$ understood. Note that by construction $\|\widetilde{H}_I(t)\| \leq 1$ and is constant in time.

Simulation by Dyson series will require a certain block encoding of $H_I$, which we now define. 
\begin{definition}[\HamT~Oracle~\cite{low2019interaction}] \label{defn:HamT}
    Let $H(t)$ be a time-dependent Hamiltonian such that for some $\Lambda\in \Rmbb_+$, $\lVert H(t) \rVert \leq \Lambda$ for all $t \in [0,T]$. Let $\{t_j\}_{j=1}^M \subset [0,T]$ be a discrete set of points. The associated $\HamT$ oracle is a unitary block encoding such that
    \begin{align*}
        (\bra{0}_a\otimes I)\HamT(H) (\ket{0}_a \otimes I) = C(H)/\Lambda, 
    \end{align*}
    where
    \begin{equation*}
        C(H) \coloneqq \sum_{j=0}^{M} \ket{j} \bra{j} \otimes H(t_j)
    \end{equation*}
    is the time-controlled Hamiltonian and $a$ is an auxiliary register.
\end{definition}
\noindent Our Hamiltonian~\eqref{eq:interaction_Ham} admits such a block encoding almost immediately. 
\begin{lemma}[\HamT~Construction] \label{lem:HamT_construction}
    Let $\widetilde{H}_I$ be the interaction picture Hamiltonian of equation~\eqref{eq:interaction_Ham} with $G\ket{\jmbf} \coloneqq g(x_\jmbf)\ket{\jmbf}$, and let $\Lambda \geq \norm{g}_\infty$ be a promised upper bound. Then there exists an $\epsilon_H$-accurate block encoding $U_g$ of $G$, such that
    \begin{equation}
        \norm{\HamT(\widetilde{H}_I) - C(U_K)^\dagger U_g C(U_K)} \leq \epsilon_H,
    \end{equation}
    with probability of failure at most $\delta$, and 
    where $\HamT(\widetilde{H}_I)$ is given in Definition~\ref{defn:HamT}. The total costs are as follows.
    \begin{itemize}
        \item \textbf{Queries}:  $Q_f$ queries to a quantum oracle $O_f$, where 
        \begin{align}
           Q_f =  
        \begin{cases}
                 2 & \text{if $O_f$ is a $(\Lambda \epsilon_H)$-accurate binary oracle (Definition~\ref{defn:binary_oracle})},\\
                 \Ocal(\log(1/\delta)/\epsilon_H) & \text{if $O_f$ is a $\epsilon_H$-accurate controlled phase oracle~(Definition~\ref{defn:phase_oracle})}.
        \end{cases}
        \end{align}
        \item \textbf{Qubits}: 
        $\Ocal(\log(1/\epsilon_H) + \log(M) +\log(1/{\delta}))$ auxiliary qubits.
        \item \textbf{Gates}: $\Ocal(\log(1/{\delta})\log^2(1/\epsilon_H) + d\log(N)^2 + d\log(N)\log(M))$ additional elementary quantum gates.
    \end{itemize}
\end{lemma}
\begin{proof}
    If $O_f$ is a phase oracle, we first convert it to an $(\Lambda\epsilon_H)$-accurate binary oracle with $\Ocal(\delta)$-failure probability using phase estimation (see Lemma~\ref{lem:phase_to_binary}). This requires ${\Ocal}(\log(1/\delta)/\epsilon_H)$ calls to $O_f$. Additionally, $\Ocal(\log(1/\delta))$ auxiliary qubits and $\Ocal(\log(1/\delta)\log^2(1/\epsilon_H))$ elementary gates are required for this conversion (including a $\log(1/\delta)$ factor for failure probability). 
    
    Now suppose we have an $(\Lambda\epsilon_H)$-accurate binary oracle $O_f$. The construction of $U_g$ comes from Lemma~\ref{lem:bitwise_to_amp_block_encoding}. This requires two $O_f$ calls, $1 + \lceil\log_2(1/\epsilon_H)\rceil$ qubits to hold the $f$ values rescaled to the range $[-1,1]$) and precision $\epsilon_H$, a single additional qubit, $O(\log^2(1/\epsilon_H))$ gates to apply the $\arccos$ transformation, and $\lceil\log_2(1/\epsilon_H)\rceil$ controlled rotations from the values register to the single ancilla.

    It remains to quantify the gate and qubit costs associated with $C(U_K)$. In terms of qubits, there are an additional $\lceil\log(M)\rceil$ required to store the time points $t_j$. We express $C(U_K)$ as
    \begin{equation}
        C(U_K) = (\QFT \;C(D_K)\; \QFT^{\dagger})^{\otimes d},
    \end{equation}
    where $D_K(t_j)$ is a diagonal operation acting on a $2N$-state register as
    \begin{equation*}
        D_K(t_j) \ket{n} = e^{-i (2\pi n)^2 \int_0^{t_j} a(\tau) d\tau} \ket{n}.
    \end{equation*}
    There are $2d$ $\QFT$s (and inverses) for a total gate cost of $\Theta(d \log(N)^2)$ using the standard approach. The $D_K$ operation can be implemented with $\Ocal(\log N)$ controlled phase gates, each of with costs $\Ocal(1)$. Controlling each of these on the time register gives, using standard techniques, a gate count of 
    $$
    \Ocal(\log M \times \log N)
    $$
    plus $\Ocal(\log M)$ work qubits. These same work qubits can be reused for each copy of $C(D_K)$. The total resources are thus figured in the lemma statement.
\end{proof}
\noindent We emphasize an important feature of this block encoding: the independence of $Q_f$ on the discretization parameter $N$. This is the key reason our algorithm has an $N$-independent query complexity, and is a major benefit of working in the $K$-frame. 

Once the $\HamT(\widetilde{H}_I)$ block encoding is constructed, it can be fed directly into the truncated Dyson series algorithm.
\begin{theorem}[Corollary 4 from~\cite{low2019interaction}]\label{thm:low_main}
    Let $H:[0,T]\rightarrow \Cmbb^{2^{n_s}\times 2^{n_s}}$ be a time dependent Hamiltonian whose derivative is bounded in $t$. Let $\Lambda \geq \max_t \norm{H(t)}$ be a promised upper bound. Let $U(T) = U(T,0)$ be the unitary propagator for $H(t)$ on $[0,T]$. For $\tau = T/\lceil 2\Lambda T\rceil$ assume $H_j(t) \coloneqq H((j-1) \tau + t)$, $t\in[0,\tau)$ admits a block encoding $\HamT(H_j)$ as per Definition~\ref{defn:HamT}. Then for all $T> 0$ and $\epsilon_\mathrm{alg} > 0$, there is a quantum algorithm producing an operation $\widetilde{U}(T)$, with failure probability $\Ocal(\epsilon_\mathrm{alg})$ such that $\lVert\widetilde{U}(T) - U(T)\rVert \leq \epsilon_\mathrm{alg}$. The algorithm requires
    \begin{equation*}
        M \in \Ocal\left(\frac{T}{\Lambda \epsilon_\mathrm{alg}}\left[\frac{1}{T}\int_0^T \norm{\dot{H}(t)}dt + \max_{t\in[0,T]} \norm{H(t)}^2\right]\right)
    \end{equation*}
    for $M\in \Zmbb_+$ in Definition~\ref{defn:HamT}, and admits the following asymptotic complexity.
    \begin{itemize}
        \item \textbf{Queries}: $Q_H \in \Ocal\left(\Lambda T \frac{\log(\Lambda T/\epsilon_\mathrm{alg})}{\log\log(\Lambda T/\epsilon_\mathrm{alg})}\right)$ to $\HamT(H)$.
        \item  \textbf{Qubits}: $n_s + \Ocal\left(n_a + \log\left(\frac{T}{\Lambda \epsilon_\mathrm{alg}}\left[\frac{1}{T}\int_0^T \norm{\dot{H}(t)}dt + \max_t \norm{H(t)}^2\right]\right)\right)$
        \item \textbf{Gates:} $\Ocal \left(\Lambda T\frac{\log(\Lambda T/\epsilon_\mathrm{alg})}{\log\log(\Lambda T/\epsilon_\mathrm{alg})}\left(n_a + \log\left(\frac{T}{\Lambda \epsilon_\mathrm{alg}}\left[\frac{1}{T}\int_0^T \norm{\dot{H}(t)}dt + \max_t \norm{H(t)}^2\right]\right)\right)\right)$
    \end{itemize}
\end{theorem}
\noindent Combining this general result with the block encodings of Lemma~\ref{lem:HamT_construction}, we are ready to state our result for simulating the pseudo-spectral Hamiltonian $H$ to discretization $N$.

\begin{theorem}\label{thm:simulation_complexity}
    Let $O_f$ be a $\epsilon_f$-accurate quantum oracle for $f$, and let $g$ be an $R$-restriction of $f$ (Definition~\ref{defn:R_restriction}). Let $U(T)$ be the unitary propagator generated by the $d$-dimensional pseudo-spectral Hamiltonian
    \begin{equation*}
        H(t) = a(t) (2\pi)^2 \QFT^{\otimes d}\hat{\nmbf}^2 \QFT^{\dagger \otimes d} + b(t)G
    \end{equation*}
    for time $T$ with discretization parameter $N$. Finally, suppose one is given a promised upper bound $\Lambda \geq \norm{g}_\infty$ ($\Lambda \geq 1$ for asymptotic purposes). Then there exists a quantum algorithm that, with probability of failure $\Ocal(\epsilon_\mathrm{alg})$, produces a quantum operation $\widetilde{U}(T)$ approximating $U(T)$ to spectral distance $\epsilon_\mathrm{alg}$. The algorithm makes $Q_f$ queries to an $\epsilon_f$-accurate quantum binary oracle or an $(\epsilon_f/\Lambda)$-accurate phase oracle such that
    \begin{equation*}
        \epsilon_f^{-1} \in \Ocal\left(\frac{\norm{b}_1}{\epsilon_\mathrm{alg}} \frac{\log(\Lambda \norm{b}_1/\epsilon_\mathrm{alg})}{\log\log(\Lambda \norm{b}_1/\epsilon_\mathrm{alg})}\right).
    \end{equation*}
    The asymptotic complexity is as follows.
    \begin{itemize}
        \item \textbf{Queries}:
        \begin{align}
           Q_f =  \begin{cases}
                 \Ocal\left(\Lambda \norm{b}_1 \frac{\log(\Lambda \norm{b}_1/\epsilon_\mathrm{alg})}{\log\log(\Lambda \norm{b}_1/\epsilon_\mathrm{alg})}\right) & \text{(binary oracle)} \\
                {\Ocal}\left(\frac{\Lambda^2 \norm{b}_1^2}{\epsilon_\mathrm{alg}} \frac{\log^3(\Lambda \norm{b}_1/\epsilon_\mathrm{alg})}{\log^2\log(\Lambda \norm{b}_1/\epsilon_\mathrm{alg})}\right) & \text{(phase oracle)} 
            \end{cases}
        \end{align}
        \item \textbf{Qubits}: $d\lceil\log_2 N\rceil +\Ocal\left(\log dN^2+\log\left(\frac{\norm{a}_1}{\epsilon_\mathrm{alg}}\right)+ \log\left(\frac{\Lambda\lVert b\rVert_1}{\epsilon_\mathrm{alg}}\right)\right)$
        \item \textbf{Gates}: 
        \begin{equation}
            \begin{aligned}
            \Ocal&\Bigg(\frac{\Lambda^2\norm{b}_1^2}{\epsilon_\mathrm{alg}} \frac{\log^3(\Lambda\norm{b}_1/\epsilon_\mathrm{alg})}{\log^2\log(\Lambda\norm{b}_1/\epsilon_\mathrm{alg})} \\
            &\times\left[d\log N\log dN^2 + d\log N\log\left(\frac{1}{\epsilon_\mathrm{alg}}\left[\norm{a}_1 + \frac{\Lambda \norm{b}_1}{dN^{2}}\right]\right) + \log^3\left(\frac{\Lambda\norm{b}_1}{\epsilon_\mathrm{alg}}\right)\right]\Bigg).
        \end{aligned}
        \end{equation}
    \end{itemize}
\end{theorem}
\begin{proof}
    Following the discussion earlier in this section, we simulate the rescaled IP Hamiltonian $\widetilde{H}_I$ in~\eqref{eq:interaction_Ham} for time $T' = \Lambda\norm{b}_1$ using the truncated Dyson series method. Converting back from the IP using $U_K(T)$ costs a $\QFT^{\otimes d}$, its inverse, and $d$ phase gates, all of which are asymptotically subdominant.

    The costs of the IP Dyson algorithm are given by Theorem~\ref{thm:low_main} in terms of a $\HamT(\widetilde{H}_I)$ oracle. Our procedure involves constructing an $\epsilon_H$-accurate version of this unitary, using the quantum oracle $O_f$ with accuracy $\epsilon_f = \Lambda\epsilon_H$ as per Lemma~\ref{lem:HamT_construction}. For analysis purposes, we begin by assuming errorless oracle access to $\HamT(\widetilde{H}_I)$ on any subinterval of $[0,T']$, and derive the consequences of Theorem~\ref{thm:low_main} for our particular context. Indeed, we have the promise that
    \begin{align*}
        \max_{t'\in[0,T']} \norm{\widetilde{H}_I(t)} \leq 1 \equiv \Lambda'.
    \end{align*}
    Moreover
    \begin{align*}
        \frac{1}{T'}\int_0^{T'} \norm{\dot{\widetilde{H}}_I} dt = \frac{\norm{a}_1}{\Lambda^2\norm{b}_1} \norm{[\QFT^{\otimes d}(2\pi \hat{\nmbf})^2 \QFT^{\dagger\otimes d}, G]}
    \end{align*}
    where the commutator is loosely upper bounded by $2d (2\pi N)^{2} \norm{g}_\infty$ using standard norm inequalities. Plugging these into Theorem~\ref{thm:low_main}, with $\Lambda, T$ replaced by $\Lambda', T'$,  gives the number of $\HamT$ queries $Q_H$. Then Lemma~\ref{lem:HamT_construction} relates these to the number of required calls to $O_f$. Putting these together gives
    \begin{equation}
        Q_f = \begin{cases}
            \Ocal\left(\Lambda \norm{b}_1 \frac{\log(\Lambda \norm{b}_1/\epsilon_\mathrm{alg})}{\log\log(\Lambda \norm{b}_1/\epsilon_\mathrm{alg})}\right) & \text{(binary oracle)} \\
            \Ocal\left(\frac{\Lambda^2 \norm{b}_1^2}{\epsilon_f} \frac{\log(\Lambda \norm{b}_1/\epsilon_\mathrm{alg})}{\log\log(\Lambda \norm{b}_1/\epsilon_\mathrm{alg})}\right) & \text{(phase oracle)}
        \end{cases}
    \end{equation}
    This gives our query complexity results up to the required precision $\epsilon_f$. To determine this, let $\widetilde{V}$ be the quantum operation implementing the Dyson protocol with exact $\HamT(\widetilde{H}_I)$, and let $\widetilde{U}(T)$ be the approximate version. Then
    \begin{equation*}
        \norm{\widetilde{V} - \widetilde{U}(T)} \leq \epsilon_H Q_H
    \end{equation*}
    where $Q_H$ is the number of $\HamT$ queries. Thus $\epsilon_H \leq \epsilon_\mathrm{alg}/ Q_H$ suffices to achieve an overall error $\epsilon_\mathrm{alg}$. Taking $Q_H$ scaling as in Theorem~\ref{thm:low_main} but with our particular parameters, we find that $\epsilon_f^{-1}$ scaling as in the statement of the theorem suffices to achieve error $\Ocal(\epsilon_\mathrm{alg})$. We then plug this scaling into the phase oracle complexity to obtain the overall query scaling.

    Next we consider qubit counts. From Lemma~\ref{lem:HamT_construction}, the number of qubits for 
    $\HamT$ scales as 
    $$\Ocal(\log(\Lambda^2/\epsilon_f) + \log M)$$ 
    plus $d \lceil\log_2(N)\rceil$ qubits needed for the simulation domain. Again using Theorem~\ref{thm:low_main}, it suffices to choose $M$ in Definition~\ref{defn:HamT} such that
    \begin{align*}
        M &\in \Ocal \left(\frac{\norm{b}_1}{\Lambda\epsilon_\mathrm{alg}}\left(\frac{\norm{a}_1}{\norm{b}_1}dN^{2}\norm{g}_\infty + \norm{g}_\infty^2\right)\right) \\
        \implies \log M &\in \Ocal\left(\log dN^2 + \log\left(\frac{1}{\epsilon_\mathrm{alg}}\left[\norm{a}_1 + \frac{\norm{b}_1\norm{g}_\infty}{dN^{2}}\right]\right)\right).
    \end{align*}
    Combining these factors gives the stated qubit costs in the theorem.

    The number of elementary gates can be tallied by (a) those part of the Dyson protocol itself and (b) those involved in the construction of the $\HamT$ oracles. It is essentially an exercise of combining previous results appropriately, and upon keeping only the dominant terms one obtains a scaling as shown.

    For the case of a phase oracle, we have failure probability of $\delta$. To achieve a failure probability of $\Ocal(\epsilon_\mathrm{alg})$ it suffices to set $\delta = \Ocal(\epsilon_\mathrm{alg}/ Q_H)$, which adds an additional $\Ocal(\log(\Lambda\lVert b\rVert_1/\epsilon_\mathrm{alg}))$ factor to the gate complexity and additive factor to the qubit count.
\end{proof}

\subsection{Error Analysis and Discretization Spacing} \label{subsec:error_analysis_hamiltonian}
    Despite its importance to quantum simulation of continuous systems, concrete statements of the error and complexity of the pseudo-spectral method appear hard to find, if not absent, in the literature.
    Many  existing results only consider the one or few dimensional setting \cite{lubich2008quantum}. There are some results tackling spectral methods for arbitrary dimensions, but do not focus on evolution equations \cite{Childs2021highprecision, gross2023sparsespectralmethodssolving}. This may not be surprising as most results in applied mathematics are not concerned with solving PDEs with asymptotic dimension. However, quantifying the $d$ dependence is essential for important applications such as many-body physics, and also for our applications to high-dimensional optimization (Section~\ref{sec:Quantum_Simulation_Algorithms_Convex_Optimization}).

    In this section, we provide a rigorous analysis of the simulation errors arising from the pseudo-spectral method, under some reasonable regularity assumptions of the potential and wavefunction. We then reduce the dependencies on the wave function into mild dependencies on the potential.
    This leads to a proof of the qubit/grid-size requirements specified in Theorem \ref{thm:master_simulation_thm}. 

    We begin by discussing two ways to project a vector $\phi\in\Hcal$ into the space $\Hcal_N$ of equation~\eqref{eq:truncated_fourier_subspace}, both of which are relevant for this work. Let $P_N:\Hcal\rightarrow\Hcal_N$ denote the orthogonal projection onto $\Hcal_N$ with respect to the inner product $\langle\cdot,\cdot\rangle$
    \begin{equation*}
        P_N[\phi] \coloneqq \sum_{\nmbf\in\Ncal} \langle \chi_{\nmbf}, \phi\rangle\, \chi_{\nmbf},
    \end{equation*}
    which a truncated Fourier series. Note that the analogous linear map $I_N:\Hcal\rightarrow\Hcal_N$ equipped with $\langle \cdot, \cdot \rangle_\Gcal$ in \eqref{eq:discrete_inner_product} defined as
    \begin{equation} \label{eq:interpolant}
        I_N[\phi] \coloneqq \sum_{\nmbf\in\Ncal} \langle  \chi_{\nmbf}, \phi\rangle_\Gcal \; \chi_{\nmbf}
    \end{equation}
    is indeed a projection, but is not orthogonal (its not even symmetric as an operator on $\Hcal$). The square brackets are often omitted for convenience in what follows. These projections can be equivalently understood as minimizing the ``distance" in the subspace $\Hcal_N$ with respect to their corresponding induced (semi)norms:
    \begin{equation*} 
        P_N\phi = \argmin_{\psi \in \Hcal_N} \norm{\phi - \psi}, \qquad I_N\phi  = \argmin_{\psi \in \Hcal_N} \lVert \phi - \psi \rVert_{\Gcal}.
    \end{equation*}
    The output of $I_N$ is called the \textit{Fourier interpolant}. Observe that $I_N \phi$ only depends on the values of $\phi$ on the grid $\Gcal$, and in particular $I_N \phi = I_N \tilde{\phi}$ whenever $\phi$ and $\tilde{\phi}$ match on $\Gcal$. 

    One can understand this choice of name by looking in the coordinate representation. Substituting in the expression for $\langle\cdot,\cdot\rangle_\Gcal$ directly into~\eqref{eq:interpolant}, 
    \begin{align*}
        I_N[\phi](x) &=\sum_{\nmbf \in \Ncal}\left((2N)^{-d}\sum_{\jmbf \in \Ncal} \overline{\chi}_{\nmbf}(x_\jmbf) \phi(x_\jmbf)\right) \chi_{\nmbf}(x)\\
        &=\sum_{\jmbf \in \Ncal}\phi(x_\jmbf)\left((2N)^{-d}\sum_{\nmbf \in \Ncal}\chi_{\nmbf}(x-x_\jmbf)\right)\\
        &=\sum_{\jmbf \in \Ncal}\phi(x_\jmbf)N^{-d}\prod_{k=1}^d\sin\left(\pi N\frac{x_k- x_{j_k}}{2}\right)\cot\left(\pi\frac{x_k - x_{j_k}}{2}\right).
    \end{align*}
    One can show that on the grid $\prod_{k=1}^d\sin\left(\pi N\frac{x_k- x_{j_k}}{2}\right)\cot\left(\pi\frac{x_k - x_{j_k}}{2}\right)$ is the Kronecker delta for $x_\jmbf$ or equivalently a Lagrange basis function. Hence the interpolant agrees with the original function on the grid $\Gcal$, while attempting to interpolate over $S$.

This leads to a natural definition of interpolating a set of values on $\Gcal$.
\begin{definition} \label{defn:fourier_interpolant}
    Let $\ket{F} \in \Cmbb^\Ncal$ be a normalized quantum state of the form
    \begin{align*}
        \ket{F} = \frac{1}{\sqrt{(2N)^d}}\sum_{\jmbf\in\Ncal} F(\jmbf) \; \ket{\jmbf},
    \end{align*}
    where $F : \Ncal \rightarrow \Cmbb$.
    Then the \emph{Fourier interpolant} $I_N\ket{F} : S \rightarrow \Cmbb$ of $\ket{F}$, is a function defined as
    \begin{align} \label{eq:fourier_interpolated_state}
        [I_N\ket{F}](x) \coloneqq \sum_{\nmbf\in\Ncal}\langle \chi_{\nmbf}, F\rangle_{\Gcal} \, \chi_{\nmbf}(x)
    \end{align}
    where $\langle\chi_\nmbf, F\rangle_\Gcal$ is given by~\eqref{eq:discrete_inner_product}. Equivalently, $I_N:\Cmbb^\Ncal \rightarrow \Hcal_N$ is the unique unitary operation such that $I_N\ket{F} = \chi_\nmbf$ for $\ket{F}$ such that $F(\jmbf) = \chi_\nmbf(x_\jmbf)$.
\end{definition}
\noindent Note that for $x_\jmbf \in \Gcal$, we have $I_N[\ket{F}](x_\jmbf) = F(\jmbf)$, which motivates the term ``interpolation." This notion gives us a reasonable mechanism to compare a discretized wavefunction with a continuous counterpart, and thus provide a natural means of quantifying errors associated with digital simulation. However, it should be stressed that one cannot produce this interpolation computationally without performing state tomography, which scales poorly with system size. Thus, this latter notion of interpolation is a theoretical tool rather than an algorithmic one.

The main difference between $P_N$ and $I_N$ is whether the Fourier transform $\langle \cdot, \chi_{\nmbf}\rangle$ or discrete Fourier transform $\langle \cdot, \chi_{\nmbf}\rangle_{\Gcal}$ is used. The error between the two is due to aliasing, which results from $\langle \cdot, \chi_{\nmbf}\rangle_{\Gcal}$ being unable to distinguish modes that differ by $2N$ in a coordinate.
This leads us to define three notions of approximation error for $\phi \in \Hcal$:
\begin{align} \label{eq:alias_truncation}
    \underbrace{\lVert \phi - I_N\phi\rVert}_{\text{(Interpolation Error)}} \leq \underbrace{\lVert \phi- P_N\phi\rVert}_{\text{(Truncation Error)}} + \underbrace{\lVert P_N\phi - I_N\phi\rVert}_\text{(Aliasing Error)},
\end{align}
to which we derive the following general result.
\begin{theorem}
\label{thm:total_interpolation_error_l2}
     Suppose $\phi \in \Hcal^m$ for positive integer $m > \max\{d/2,2\}$, and that the periodic extension of $\phi$ is in $C^{1}(\Rmbb^d)$. Then, 
     \begin{align*}
        \lVert \phi - I_N\phi\rVert \lesssim\left(\frac{\pi}{4}\right)^{d/4} \frac{1}{\sqrt{(m-d/2) \Gamma(d/2)} N^{m}} \abs{\phi}_{\Hcal^m} + \frac{\lvert \phi \rvert_{\Hcal^1}}{N},
    \end{align*}
    where the first term bounds the aliasing error and the second term bounds the truncation error (equation~\eqref{eq:alias_truncation}).
\end{theorem}
\noindent The proof of this result is in Appendix~\ref{subapp:state_interpolation_bound}.

As shown in Corollary~\ref{coro:Our_actual_H} of the prior section, the pseudo-spectral method approximately solves Problem~\ref{prob:restricted_schr} by simulating a discretized Hamiltonian. In the analysis that follows, we will require a slight strengthening of this result that applies in the continuous setting. Specifically, we use the tools introduced above to form an interpolated version of the discretized equation (from Corollary \ref{coro:Our_actual_H}) solved by the pseudo-spectral method.
\begin{lemma}
\label{lem:interpolating_collocation}
    Let $\Psi(t)$ be the Fourier interpolation of the output quantum state $|\Psi_t\rangle$ in \eqref{eq:discrete_quantum_state} of the pseudo-spectral method (Corollary~\ref{coro:Our_actual_H}). Then, $\Psi$ satisfies
    \begin{align}
    \label{eq:continuous_collocation_eqn}
    i\partial_t\Psi(x, t) = -a(t)(\Delta\Psi)(x, t) + I_N[V\Psi](x, t).
    \end{align}
\end{lemma}
\begin{proof}
We begin with the solution $\widehat{\Psi}$ to the pseudo-spectral Hamiltonian in Fourier space from Lemma~\ref{lem:FT_pseudo_equations}. 
We have that
\begin{align*}
    \ket{\widehat{\Psi}_t} = \sum_{\nmbf \in \Ncal} \widehat{\Psi}_\nmbf(t) \ket{\nmbf}.
\end{align*} 
Using Definition~\ref{defn:fourier_interpolant} and the statement of Lemma~\ref{lem:FT_pseudo_equations}, which we recall below:
\begin{align*}
    &i\partial_t\widehat{\Psi}_{\nmbf}(t) = a(t) (2\pi \mathbf{n})^2 \widehat{\Psi}_{\mathbf{n}}(t) + \sum_{\mathbf{m} \in \Ncal}\langle \chi_{\mathbf{n}}, V(t)\chi_{\mathbf{m}}\rangle_{\Gcal}\widehat{\Psi}_{\mathbf{m}}(t),
\end{align*}
we have
\begin{align*}
    \sum_{\nmbf \in \Ncal}i\partial_t\widehat{\Psi}_{\nmbf}(t)\chi_{\nmbf}(x) &= a(t)\sum_{\nmbf\in \Ncal}(2\pi \mathbf{n})^2\widehat{\Psi}_{\mathbf{n}}(t)\chi_{\nmbf}({x}) + \sum_{\mathbf{n} \in \Ncal}\sum_{\mathbf{m} \in \Ncal}\langle \chi_{\mathbf{n}}, V(t)\chi_{\mathbf{m}}\rangle_{\Gcal}\widehat{\Psi}_{\mathbf{m}}(t)\chi_{\nmbf}(x)\\
    \implies i\partial_t\Psi(x, t) &= -a(t)(\Delta\Psi)(x, t) + \sum_{\mathbf{n} \in \Ncal}\langle \chi_{\mathbf{n}}, V( t)\Psi\rangle_{\Gcal}\chi_{\nmbf}(x).
\end{align*}
Using the definition of $I_N[V\Psi]$ in \eqref{eq:fourier_interpolated_state} in the last line gives the statement of the lemma.
\end{proof}
\noindent Thus, the discretized dynamics furnish an approximate Hamiltonian evolution on the continuous space via interpolation. Comparing Lemma~\ref{lem:interpolating_collocation} with the exact~\ref{eq:real_schrodinger} equation, we see that the only difference is the interpolated potential operator $I_N[V \cdot]$ rather than $V$ itself. 

The term $I_{N}[V \Psi]$ presents an additional challenge that is not present in some well-known applications of pseudo-spectral methods, e.g., Elliptic PDEs.
Even if $V$ is, say, a Fourier polynomial of degree $\leq N$, it is still not the case that the pseudo-spectral method is exact, because $V$ will generate transitions to higher Fourier modes in the state $\Psi$.

The following theorem bounds the $L_2$ distance between the (interpolated) output of the pseudo-spectral method and the true solution $\Phi$ in terms of the solution regularity. Here, one sees that the key quantities are the Sobolev norms, where the dominating term comes from the $\Hcal^k$ norm with $k = \Theta(d)$. We emphasize that it is critical to have a result that is valid when \emph{the dimension of the PDE is an asymptotic parameter}, as the one below.
\begin{theorem}[Multi-dimensional Generalization of Theorem 1.8 \cite{lubich2008quantum}]
\label{thm:collocation_bound_real_space}
Let $\Phi(t)$ denote the solution to 
\begin{equation}
    i \partial_t \Phi(x, t) = [- a(t) \Delta + V(x,t)]\Phi(x, t),
\end{equation}
over $S = [-\frac12, \frac12]^d$,
with initial condition $\Phi(x, 0) = \Phi_0 \in \Hcal_{N}$. If $\forall t \in [0, T]$,  $\Phi(t) \in \Hcal^{m+2}$ for $m > \max(d/2, 2)$, and $\Phi(t)$  has a  periodic extension, in $x$, that is in $C^{1}(\Rmbb^d)$, then the $L_2$-distance between the true solution $\Phi(t)$ and the pseudo-spectral approximation $\Psi(t)$ is bounded as
\begin{align*}
    \lVert \Phi(t) - \Psi(t)\rVert \lesssim\norm{a}_{1}\cdot \max_{t \in [0 ,T]}\left[\left(\frac{\pi}{4}\right)^{d/4}\frac{2\lvert \Phi(t) \rvert_{\Hcal^{m+2}} + \lvert \Phi(t)\rvert_{\Hcal^m}}{\sqrt{(m-d/2) \Gamma(d/2)} N^{m}} + \frac{2\lvert \Phi(t) \rvert_{\Hcal^3} + \lvert \Phi(t) \rvert_{\Hcal^1}}{N}\right].
\end{align*}
\end{theorem}

The above result by itself already provides a rigorous guarantee on the output of the pseudo-spectral method for the Schr\"odinger equation. However, as currently written, it still depends on the regularity of the true solution $\Phi$ to the problem, limiting practical applicability. We show that the explicit dependence on $\Phi$ can be removed if the potential is  assumed to be Lipschitz continuous. These results (Theorems \ref{thm:collocation_trace_distance} \& \ref{thm:collocation_trace_distance_aperiodic}) lead to bounds that are in terms of parameters that are significantly easier to bound and are usually even assumed to be given. Still, an important point to note  is that for Theorem \ref{thm:collocation_bound_real_space} to ensure $\Ocal\left(\epsilon\right)$ error with $N$ going only as poly $d$, it suffices for $\Phi$'s $d$-th Sobolev norm to go at most as $d^{cd}$ for some constant $c$.

The reader may also notice that Theorem \ref{thm:collocation_bound_real_space} requires the initial wavefunction $\Phi_0$ to be an element of the space of finite Fourier polynomials $\Hcal_{N}$, where $N$ is the grid size in a single dimension. This restriction can  be easily removed, as we show below.

\begin{lemma}
\label{lem:gen_periodic_initial_condition}
Consider the setting of Problem~\ref{prob:restricted_schr}. 
There exists a modified problem with a ``proxy'' initial condition $\widetilde{\Phi}_{0} \in \Hcal_N$ such that the solutions to Problem~\ref{prob:restricted_schr} at time $T$ with and without the modified initial condition  are $\Ocal\left(\epsilon\right)$ in $L_2$ distance. 

The requirement on $N$ for the ``proxy'' initial wave function $\widetilde{\Phi}_0$ can be either of the following:
\begin{align*}
    N = 
    \begin{cases}
        \Ocal\left(\frac{\lvert \Phi_0\rvert_{\Hcal_1}}{\epsilon}\right) & \text{ if}\quad \widetilde{\Phi}_0 = \frac{P_N\Phi_0}{\lVert P_N\Phi_0\rVert}, \\
        \Ocal\left( \frac{d^{-1/4}\lvert \Phi_0\rvert_{\Hcal^d}^{1/d} + \lvert \Phi_0 \rvert_{\Hcal^1}}{\epsilon} \right) & \text{ if}\quad \widetilde{\Phi}_0 = \frac{I_N\Phi_0}{\lVert I_N\Phi_0\rVert}.
    \end{cases}
\end{align*}
\end{lemma}

The reason for considering two potential candidate proxy wavefunctions is that one or the other may be easier to load. Specifically, if $\Phi_0$ is such that we have closed form expression for its Fourier transform, then $P_N\Phi_0$ might be tractable. Otherwise, we can just load $\Phi_0$ on a grid to get $I_N\Phi_0$. Obviously for large $d$, we will want the initial wave function to be separable, so that preparing the state is tractable.

As mentioned earlier, our main Theorem \ref{thm:master_simulation_thm} applies  to potentials that are only assumed to be Lipschitz continuous on $S$ with continuous periodic extensions over the Torus. The Lipschitz-continuity assumption is generally considered a mild assumption for black-box problems, e.g. optimization problems. We can swap out the continuous periodic-extension assumption for an alternative assumption that is believable in practice (Theorem \ref{thm:collocation_trace_distance_aperiodic}).

Also, a footnote in Theorem \ref{thm:master_simulation_thm} mentions that, due to a technicality, we assume that the true wave function has a periodic extension that is in $C^{1}(\Rmbb^d)$. However, since by Theorem \ref{thm:solution_existence} we can guarantee high-regularity by ensuring a highly-regular initial condition, the assumption that the initial wave has 
a periodic extension in $C^{1}(\Rmbb^d)$ is also believable in practice. The footnote also mentions that we need $\rightarrow \partial^{\alpha}_x\Phi(x, t)$ to be continuous $\forall x \in S$ and $\alpha \in \Nmbb^{d}$. This to be able to exchange some limits and apply Lemma \ref{lem:sobolev_growth_wavefunc} below. This condition is also believable given the regularity of $a(t), b(t) \in C^{1}([0, T])$.

The following result is the Sobolev-growth bound and generalizes the results of \cite{bouland2023quantum} to our setting. 
We emphasize that, in our setting, it is important to account for scaling in $d$ for parameters that, in other applications, are considered to be ``absolute constants.'' Note that in the below result,  we will require $\Phi_0 \in \Hcal^{(2m + 2)+\lceil \frac{d}{2} \rceil + 1}$. However, this is trivially satisfied for all $m$, when $\Phi_0$ is a Fourier polynomial, i.e. in $ \Hcal_N$ for any $N$. The previous lemma allows for a reduction from weaker initial state assumptions.
\begin{lemma}[Sobolev Growth Bound -- Multidimensional Generalization of \cite{bourgain1999growth}]
\label{lem:sobolev_growth_wavefunc}
    Suppose $S$ is $d$ dimensional and $\Phi(t)$ is the  classical solution to Problem \ref{prob:restricted_schr}.  Furthermore, suppose that $g(x) \in \Hcal^{(2m + 2)+\lceil \frac{d}{2} \rceil + 1}$,  $\Phi_0 \in \Hcal^{(2m + 2)+\lceil \frac{d}{2} \rceil + 1}$, and that $\forall \alpha \in \Nmbb^{d}$ with $\lVert \alpha \rVert_1 \leq 2m+2$, we have 
 $t \rightarrow \partial^{\alpha}_x\Phi(x, t)$ is continuous $\forall x \in S$. Then
    \begin{align*}
         \lvert \Phi(t)\rvert_{\Hcal^s} \leq 2{\lvert \Phi_0\rvert_{\Hcal^{s}}} + \left(\frac{2\sum_{k=1}^{s}\binom{s}{k}\lvert g\rvert_{\Hcal^k}\lVert b\rVert_1}{s}\right)^s
    \end{align*}
    for all $s \leq m$.
\end{lemma}
The above, by itself, already rephrases the quantities in Theorem \ref{thm:collocation_bound_real_space} in terms of quantities that are usually inputs to the problem, i.e. the regularity of $f$. Combining the two results, leads to a bound on the pseudo-spectral performance for highly-regular potentials (Appendix \ref{subapp:potential_regularity}). 
However, if $f$ is only assumed to be Lipschitz, then it may not be immediately clear that this makes the problem any easier. Specifically, the above appears to require $f$ to have finite Sobolev norms of high degree.

In the subsection that follows, we will remove the regularity assumption on $f$ and prove the qubit-count component of Theorem \ref{thm:master_simulation_thm}. Note Section \ref{subsec:simulation_algorithm} already covered the query complexity and gate count. The qubit-count result will appear as a trivial consequence of multiple theorems (specifically, Theorems \ref{thm:collocation_trace_distance} \& \ref{thm:energy_error_bound}) that we prove separately. The main mathematical tool we utilize to cover the Lipschitz case is mollification. Admittedly, we note that the result appears to be overly pessimistic, i.e. number (qu)bits per dimension growing with $d$. However, it seems challenging to get a better theoretical bound using such a mild assumption as Lipschitzness of the potential. If we can assume some addition regularity on the problem, then we can improve the results significantly (Appendix \ref{app:improved_regularity}).

\subsubsection{Guarantees for Lipschitz Potentials}
\label{sec:lipschitz-guarantee}

We start by recalling the definition of the rectangular mollifier function.
\begin{definition}[Rectangular Mollifier]
\label{defn:mollifer}
Suppose $0 < \sigma < \frac12$. The $d$-dimensional rectangular mollifier over $S$ is the following $C^{\infty}(S)$ function:
\begin{align*}
    \mathcal{M}_{\sigma, d}(x) = \begin{cases}
        (\vartheta\sigma)^{-d}\prod_{j=1}^{d}e^{-\frac{1}{1-x_j^2/\sigma^2}} & x \in (-\sigma,\sigma)^d\\
        0 & \text{otherwise}
    \end{cases},
\end{align*}
and satisfies $\lVert \mathcal{M}_{\sigma, d}\rVert_1 = 1$. The value of $\vartheta$ is the normalization constant of the Mollifier with $\sigma =1$.
\end{definition}
It should be apparent that periodic extension of $\mathcal{M}_{\sigma, d}$ is in $C^{\infty}(\Rmbb^d)$.

Our approach for controlling large Sobolev norms of the potential is to instead utilize a smooth proxy function that can be shown to be close to the original potential. Note that we do not need to explicitly construct this function, its existence suffices. This function is the result of the periodic convolution of the $R$-restriction of the potential and $\mathcal{M}_{\sigma, d}$. The bound on the Fourier coefficients of the resulting function comes from the convolution theorem. This is a common technique in signal processing to combat the Gibbs phenomenon \cite{tadmor2007filters}.

\begin{lemma}[Periodic Mollification]
\label{lem:mollification}
Consider the $R$-restriction $g$ of a $G$-Lipschitz function in $p$-norm $f$, $p\geq 1$. Also define the set
\begin{align*}
 D := \begin{cases}
      S & \text{if}~$g$ ~\text{satisfies periodic boundary conditions,} \\ 
     B_{\infty}(0, \frac12 - \sigma) & \text{otherwise}
\end{cases}
\end{align*}
for any $0 < \sigma < \frac12$. 

There exists a $C^{\infty}(S)$ function $g_{\sigma}$ satisfying the following conditions:
\begin{itemize}
    \item $\lvert g(x) - g_{\sigma}(x) \rvert = \Ocal\left(d^{1/p}GR\sigma\right), \forall x \in D,$
    \item  if $f$ is convex, $g_{\sigma}$ is convex on $D$, and
    \item  for $m \leq d + \Ocal(1)$,
    $$\lvert g_{\sigma}\rvert_{\Hcal^m} =\Ocal\left(\lVert g \rVert_{\infty}\left(\frac{Cd}{\sigma}\right)^{3d}\right),$$
    where $C$ is an absolute constant.
\end{itemize}
\end{lemma}

Hence, if we use $g_{\sigma}$ in place of $g$, then we obtain an asymptotic bound on Sobolev norms of $d+\Ocal(1)$-order, which is required for Lemma \ref{lem:sobolev_growth_wavefunc} to be useful for Theorem \ref{thm:collocation_bound_real_space}. If $g$ satisfies PBC, then we can arbitrarily drive down the uniform error between the two over $S$. This leads to the following result, which proves one half of the discretization error component of  Theorem \ref{thm:master_simulation_thm}.

\begin{theorem}[Pseudo-spectral Approximation for Lipschitz, PBC Potentials]
\label{thm:collocation_trace_distance}
    Suppose $f : \mathcal{X} \rightarrow \Rmbb$ is a $G$-Lipschitz continuous function in $p$-norm with $[-R,R]^d\subseteq \mathcal{X} \subseteq \Rmbb^d$ and $p \geq 1$, and let $g$ be its $R$-restriction such that $g$ satisfies periodic boundary conditions. Suppose the assumptions of Lemma \ref{lem:sobolev_growth_wavefunc} and Theorem \ref{thm:collocation_bound_real_space} are satisfied. Let $\Psi(T)$ be the pseudo-spectral approximation to $\Phi(T)$ at time $T$, where $\Phi$ solves Problem \ref{prob:restricted_schr}. Then $\norm{\Psi(T) - \Phi(T)} = \Ocal(\epsilon)$ for a choice of discretization $N$ such that
    \begin{align*}
        \log(N) = \Ocal\left(d\log_2(\lVert a \rVert_1^{1/d}[\max(\lvert\phi \rvert_{\Hcal^{d}}, \lvert\phi \rvert_{\Hcal^{3}})^{1/d}  +\lVert b \rVert_1GRd]/\epsilon)\right).
    \end{align*}
\end{theorem}
\begin{proof}
    Let $\norm{a}_1= \int_0^T a(t)dt,  \norm{b}_1= \int_0^T b(t)dt$. If $\sigma = \frac{\epsilon}{d^{1/p}GR\lVert b \rVert_1}$, then by Lemma \ref{lem:mollification} there exists an $C^{\infty}(S)$ function $g_{\sigma} : [-\frac12, \frac12]^d \rightarrow \Rmbb$  such that
    \begin{align}
    \lVert g_{\sigma} - g \rVert_{\infty} &= \Ocal\left(\frac{\epsilon}{\lVert b \rVert_1}\right)\label{eq:infty_norm_g_bound}\\
    \lvert g_{\sigma}\rvert_{\Hcal^m} &= \Ocal\left(\lVert g \rVert_{\infty}\left(\frac{C\lVert b\rVert_1GRd^{1+2/p}}{\epsilon}\right)^{3d}\right), &\text{for} ~ m \leq d +\Ocal(1). \label{eq:sob_seminorm_bound}
    \end{align}
    In addition, we periodically extend both $g$ and $g_{\sigma}$ in the natural way, leading to a notion of Fourier series.
    
    We start by considering the following sequence of equations:
    \begin{subequations}
    \begin{align}
    \label{eq:first_pde}
    &i\partial_t\Phi = -a(t)\Delta\Phi + b(t)g\Phi\\
    \label{eq:second_pde}
    &i\partial_t\Phi_{\sigma} = -a(t)\Delta\Phi_{\sigma} + b(t)g_{\sigma}\Phi_{\sigma}\\
    \label{eq:thid_pde}
    &i\partial_t\Psi_{\sigma} = -a(t)\Delta\Psi_{\sigma} + b(t)I_N[g_{\sigma}\Psi_{\sigma}]\\
    \label{eq:fourth_pde}
    &i\partial_t\Psi = -a(t)\Delta\Psi + b(t)I_N[g\Psi],
    \end{align}
    \end{subequations}
    where  the last two correspond to interpolated pseudo-spectral equations using Lemma \ref{lem:interpolating_collocation}. From a triangle inequality, we have
    \begin{align}
    \label{eq:triangle_inequality}
     \lVert \Phi - \Psi\rVert\leq & \lVert \Phi - \Phi_{\sigma}\rVert +\lVert \Phi_{\sigma} - \Psi_{{\sigma}}\rVert+\lVert \Psi_{\sigma} - \Psi\rVert.
    \end{align}
    
    An immediate consequence of Equation~\eqref{eq:infty_norm_g_bound}, when combined with Lemma \ref{lem:evolution_bound}, is that the first and last distances on the right-hand side of \eqref{eq:triangle_inequality}  are both $\Ocal(\epsilon)$. Note that for the last term we apply \eqref{eq:triangle_inequality} in Fourier space and use that $\Psi_{{\sigma}}, \Psi \in \Hcal_N$, so $L_2$ distance is preserved by discrete FT.
    
    For the middle term, we can apply Theorem~\ref{thm:collocation_bound_real_space} using $g_{\sigma}$.
    The resulting bound is
    \begin{align}
        \label{eq:collocation_bound}
        \lVert \Phi_{\sigma}(T) - \Psi_{{\sigma}}(T)\rVert \lesssim \norm{a}_{1}\cdot \max_{t \in [0 ,T]}\left[\left(\frac{\pi}{4}\right)^{d/4}\frac{2\lvert \Phi_{\sigma}(t) \rvert_{\Hcal^{m+2}} + \lvert \Phi_{\sigma}(t)\rvert_{\Hcal^m}}{\sqrt{(m-d/2) \Gamma(d/2)} N^{m}} + \frac{2\lvert \Phi_{\sigma}(t) \rvert_{\Hcal^3} + \lvert \Phi_{\sigma}(t) \rvert_{\Hcal^1}}{N}\right]
    \end{align}
    If we combine Lemma~\ref{lem:sobolev_growth_wavefunc} and Equation~\eqref{eq:sob_seminorm_bound}, then for all $t \in [0, T]$:
    \begin{align}
        \lvert \Phi_{\sigma}(t)\rvert_{\Hcal^m} &\lesssim{\lvert \Phi_0\rvert_{\Hcal^{m}}} + \left(\frac{2^{m}\sum_{k=1}^{m}\lvert  g_{\sigma}\rvert_{\Hcal^k}\lVert b\rVert_1}{m}\right)^m  \nonumber\\
        \label{eq:sobolev_sigma_wavefunc_bound}
        &\lesssim\lvert \Phi_0\rvert_{\Hcal^{m}} + \left(\frac{2^m (C\lVert b\rVert_1GRd^{1+2/p})^{3d}}{\epsilon^{3d}}\lVert b \rVert_1\right)^m,
    \end{align}
    for $m \leq  d + \Ocal(1)$. Accordingly, applying Lemma~\ref{lem:gns_inequality} and the previous asymptotics:
    \begin{align*}
        &\lvert \Phi_{\sigma} \rvert_{\Hcal^{d/2 + 1}}  \leq \lvert \Phi_{\sigma} \rvert_{\Hcal^{d/2 + 3}} \leq  \lvert \Phi_{\sigma} \rvert_{\Hcal^{d}} \lesssim\lvert \Phi_0 \rvert_{\Hcal^{d}} + \left(\frac{ (C'\lVert b\rVert_1Gd^{1+2/p})^{3d}}{\epsilon^{3d}}\lVert b \rVert_1\right)^d \\
        &\lvert \Phi_{\sigma} \rvert_{\Hcal^{1}}\leq \lvert \Phi_{\sigma} \rvert_{\Hcal^{3}} \lesssim\lvert \Phi_0\rvert_{\Hcal^{m}} + \left(\frac{ (C\lVert b\rVert_1GRd^{1+2/p})^{3d}}{\epsilon^{3d}}\lVert b \rVert_1\right)^3.
    \end{align*}
    Thus 
    \begin{align*}
        \lVert \Phi_{\sigma} - \Psi_{{\sigma}}\rVert &= \Ocal\left(\left[\frac{\lVert a \rVert_1^{1/d}\lvert\Phi_0 \rvert_{\Hcal^{d}}^{1/d} + \lVert a \rVert_1^{1/d}\left(\frac{ (C'\lVert b\rVert_1GRd^{1+2/p})^{3d}}{\epsilon^{3d}}\lVert b \rVert_1\right)}{d^{1/4}N}\right]^{d} \right. \\ &+ \left. \frac{\lVert a \rVert_1{\lvert \Phi_0\rvert_{\Hcal^{3}}} + \lVert a \rVert_1\left(\left(\frac{ (C\lVert b\rVert_1GRd^{1+2/p})^{3d}}{\epsilon^{3d}}\lVert b \rVert_1\right)^3\right)}{N}\right).
    \end{align*}
    
     Consequently, if the number of qubits per dimension is
    \begin{align*}
    n = \log_2(N) = \Ocal\left(d\log_2(\lVert a \rVert_1^{1/d}[\max(\lvert\Phi \rvert_{\Hcal^{d}}, \lvert\Phi \rvert_{\Hcal^{3}})^{1/d}  +\lVert b \rVert_1GRd]/\epsilon)\right),
    \end{align*} then the error in \eqref{eq:collocation_bound} is $\Ocal\left(\epsilon\right)$.
    The overall result then follows from \eqref{eq:triangle_inequality}.
\end{proof}

If the $R$-restriction does not satisfy periodic boundary conditions, then we can still bound the pseudo-spectral error $\lVert \Phi - \Psi \rVert$ if we assume the following.
\begin{restatable}{definition}{lowleakage}[$\delta_0$ Low-Leakage Evolution]
\label{defn:low-leakage-ev}
We say that Problem \ref{prob:restricted_schr} with potential $g$ results in $\delta_0$ Low-Leakage Evolution if there exists an $0 < \delta_0 < 1$ s.t. for all $0<\delta < \delta_0$:
\begin{align*}
\int_{[B_{\infty}( \frac12-\delta)]^{c}} \lvert \Phi(t) \rvert^2  \leq \frac{\delta}{2\lVert g \rVert_{\infty}\lVert b \rVert_1}.
\end{align*}
\end{restatable}
The above is inspired by the fact that Lemma \ref{lem:mollification} does not provide uniform error guarantees if $g$ does not satisfy PBC. The above is relatively believable if the initial state is chosen to be tightly concentrated inside the ``good''' region, $B_{\infty}(\frac12-\delta)$. If the potential is also, highlty confining, then it will be unlikely that the state leaves the good region throughout time. This leads to the following result that removes PBC assumption.

\begin{theorem}[Pseudo-spectral Approximation for Lipschitz Potentials]
\label{thm:collocation_trace_distance_aperiodic}
    Suppose $f : \mathcal{X} \rightarrow \Rmbb$ is a $G$-Lipschitz continuous function in $p$-norm with $[-R,R]^d\subseteq \mathcal{X} \subseteq \Rmbb^d$ and $p \geq 1$, and let $g$ be its $R$-restriction. Suppose that $g$ results in a $\Ocal\left(\frac{1}{R}\right)$ Low-Leakage Evolution (Definition \ref{defn:low-leakage-ev}), then the consequences of Theorem \ref{thm:collocation_trace_distance} apply with same specified value of $\log_2(N)$.
\end{theorem}

The above results shows that the \emph{Fourier interpolant} of $|\Psi\rangle$ output by the pseudo-spectral method is close in $L_2(S)$ distance to the true solution $\Phi$. Unfortunately, we only have access to the Fourier interpolant over the grid. So, it is not an immediate consequence that $N$ is large enough such that
\begin{align*}
\Big\lvert \sum_{x_\jmbf \in \Gcal} g(x_\jmbf )\lvert \Psi(x_\jmbf )\rvert^2(2N)^{-d} 
 -  \int_S g(y) \lvert \Phi ({y})\rvert^2 dy \Big\rvert
\end{align*}
is small. The above is effectively a Riemann summation error. The below results shows that this still can be bounded and the grid spacing chosen in the previous theorem ends up being sufficient. The first result below presents the Riemann summation error in terms of quantities that we know how to bound.

\begin{lemma}[Pseudo-spectral Method Expectation Error General]
\label{lem:energy_error_bound_gen}
Suppose we have the hypotheses of Theorem \ref{thm:collocation_trace_distance}.
Let $h : \mathcal{X} \rightarrow \Rmbb$ be $G$ Lipschitz in $\ell_p$ norm with $\mathcal{X} \subseteq \Rmbb^d$ and $g$ be its $R$-restriction. Then
\begin{align*}
\big\lvert \sum_{x_\jmbf \in \Gcal} g(x_\jmbf )\lvert \Psi(x_\jmbf )\rvert^2(2N)^{-d} 
 &- \int_S g(y) \lvert \Phi ({y})\rvert^2 dy\big\rvert \\ &\lesssim \frac{\lVert g\rVert_{\infty}d\sup_{S}\lVert \nabla\lvert \Phi_M \rvert^2\rVert_{\infty} + GRd^{1/p}}{N} + \lVert g\rVert_{\infty}\lVert \Psi - \Phi\rVert + \frac{\lVert g\rVert_{\infty}\lvert \Phi\rvert_{\Hcal^1}}{M} 
\end{align*}
for any $M \leq N$, where $\Phi_M \coloneqq P_M\Phi$.
\end{lemma}

The previous result displays the true dependence of the difference in expectation value on the grid and continuum. These quantities may be more well-behaved in practice, e.g. Theorem \ref{thm:energy_diff_analytic}. The following result does provide an end-to-end guarantee and does not require any additional assumptions beyond Lemma \ref{lem:energy_error_bound_gen}. We have separated the results to indicate the potential looseness of the bound. The main problematic quantity is $\sup_{S}\lVert \nabla\lvert \Phi_M \rvert^2\rVert_{\infty}$.

\begin{theorem}[Pseudo-spectral Method Expectation Error]
\label{thm:energy_error_bound}
    Assuming the hypotheses of Lemma \ref{lem:energy_error_bound_gen}, then we also have
    \begin{align*}
    \left\lvert \sum_{\mathbf{x_\jmbf} \in \Gcal} g(\mathbf{x_\jmbf})\lvert \Psi(\mathbf{x_\jmbf})\rvert^2(2N)^{-d} 
     -  \int_S g(y) \lvert \Phi ({y})\rvert^2 dy \right\rvert = \Ocal(\lVert g\rVert_{\infty}\epsilon).
    \end{align*}
\end{theorem}
\begin{proof}
    We have the following generic bound for any $M$
    \begin{align*}
        \lVert \nabla(\lvert \Phi_M\rvert^2) \rVert_{\infty} &\leq 2\lVert \Phi_M \rVert_{\infty} \lVert \nabla \Phi_M\rVert_{\infty} \\
        &\leq 2\lVert \Phi_M \rVert_{\infty} \sum_{\nmbf \in \mathcal{R}}\lVert \nmbf \rVert_{\infty}\lvert c_{\nmbf}\rvert \\
        &\leq 2M\left(\sum_{\nmbf \in \mathcal{R}} \lvert c_{\nmbf}\rvert\right)^2 \\
        &\leq 2M^{d+1}\left(\sum_{\nmbf \in \mathcal{R}} \lvert c_{\nmbf}\rvert^2\right)\\
        &\leq 2M^{d+1}.
    \end{align*}
    
    From Lemma \ref{lem:energy_error_bound_gen}, we want to take $M = \Ocal\left(\lvert \Phi \rvert_{\Hcal^1}/\epsilon\right)$. From Eq.~\eqref{eq:sobolev_sigma_wavefunc_bound} we have that via mollification 
  \begin{align*}
        \lvert \Phi \rvert_{\Hcal^1} = \Ocal\left(\lvert \Phi_0\rvert_{\Hcal^{1}} + \left(\frac{(C\lVert b\rVert_1GRd^{1+2/p})^{3d}}{\epsilon^{3d}}\lVert b \rVert_1\right)\right).
    \end{align*}
    Theorem \ref{thm:collocation_trace_distance} provides a sufficient value of $N$ such that $\lVert \Phi - \Psi\rVert = \Ocal(\epsilon)$. One will note  that this value of $\log_2(N)$ is of the same order as $\log_2(N)$ required to suppress the $M^{d+1}$ term for our choice of $M$. Hence the same $N$ from Theorem \ref{thm:collocation_bound_real_space}
    ensures that 
    \begin{align*}
        \left\lvert \sum_{x \in \Gcal} g({x}_{\jmbf})\lvert \Psi({x}_{\jmbf})\rvert^2(2N)^{-d} 
         -  \int_S g(y) \lvert \Phi ({y})\rvert^2 dy \right\rvert = \Ocal(\lVert g\rVert_{\infty}\epsilon).
     \end{align*}
\end{proof}
Together, Theorems \ref{thm:collocation_trace_distance} and \ref{thm:energy_error_bound} provide a proof of the qubit count for Theorem \ref{thm:master_simulation_thm}. Section \ref{sec:simulation-cost} provided a proof of the query complexity and gate count. Thus, the proof of Theorem \ref{thm:master_simulation_thm} is complete.

\subsection{Comparison of Simulation Result to Childs et al. \cite{childs2022quantumsim}}
\label{sec:comparison}

As mentioned above, our work is not the first to analyze quantum algorithms for real space simulation. Indeed, the pseudo-spectral framework of Section~\ref{subsec:pseudo-spectral_methods} was previously considered in~\cite{childs2022quantumsim}, which made a significant step towards a rigorous characterization of real-space quantum Hamiltonian simulation. The framework and analysis has been subsequently used for various applications \cite{Zhang_2021,leng2023quantum,leng2023quantum2,leng2025quantumhamiltoniandescentnonsmooth, augustino2023quantum}. However, we argue that existing literature leaves out a critical piece of analysis that is essential to obtain a complete end-to-end analysis of quantum real-space algorithms is absent. 

As highlighted in Section \ref{sec:simulating-schrodinger-dynamics} in great detail, the output of the pseudo-spectral method is \emph{neither a truncated Fourier series of the true solution nor a trigonometric interpolation}. However, \cite{childs2022quantumsim} provides only a proof of the required grid spacing to ensure low Fourier truncation error, which does not by itself rigorously imply good pseudo-spectral method performance. The key missing link is showing that the Fourier interpolant is close to output of the pseudo-spectral method in $L_2$ distance.

For the truncation error bound, the results of~\cite{childs2022quantumsim} are given in terms of a quantity $g'$, which depends on properties of the solution wavefunction and are left unbounded. Note that no such quantity appears in our main simulation result (Theorem \ref{thm:master_simulation_thm}). As mentioned above, this is essential, since we demand a simulation theorem that is applicable to black-box potentials.

The Fourier truncation result of \cite{childs2022quantumsim} was  only proven for one-dimension, which we generalized to multiple dimensions in Section \ref{subapp:state_interpolation_bound}. We remark that results in one-dimensional Fourier analysis do not always trivially generalize to multi-dimensions, and in some cases may not even be true \cite{weisz2012summability}.

Inspired by the one-dimensional proofs of \cite{lubich2008quantum} and \cite{shen2011spectral}, we generalize their analysis of pseudo-spectral methods to arbitrary dimensions. This makes use of the additional notion of the Fourier interpolant, which enables bounding the $L_2$ error between the output of pseudo-spectral method and the Fourier interpolant of the true solution. A notable consequence of our work is rigorously proving that pseudo-spectral real-space quantum simulation algorithms work for the applications to which they have been previously applied.

\section{Convex Optimization via Hamiltonian Simulation} 
\label{sec:Quantum_Simulation_Algorithms_Convex_Optimization}

In this section, we present our results on quantum algorithms for optimization via Quantum Hamiltonian Descent, providing a rigorous and almost optimal characterization of the query complexity. We recall the unconstrained convex optimization problem.
\probnonsmoothconvex*
\noindent Recall also that QHD dynamics are specified by evolution under a time dependent Hamiltonian
\begin{equation*}\tag{QHD}
    H_{\text{QHD}}(t) \coloneqq c_t \left(-\frac{1}{2 m_t}\Delta  + m_t \omega_t^2 f(x)\right)
\end{equation*}
where $c_t, m_t, \omega_t \in C^1([t_0,\infty))$ define our QHD ``schedule." Throughout, we will assume the \ref{eq:ideal_scaling} conditions, which are a generalization of conditions under the same name in \cite{wibisono2016variational,leng2023quantum} and which we restate below.
\begin{definition}[Ideal Scaling Condition]
The schedules $c_t, m_t,\omega_t \in C^1([t_0,\infty))$ satisfy ideal scaling conditions if
    \begin{equation*}
\tag{Ideal Scaling}
    \dot{m_t} = \lambda c_t m_t, \quad \dot{\omega}_t \leq \frac12 \lambda c_t \omega_t.
\end{equation*}
for some $\lambda > 0$. If the second expression holds with equality, we use the term ``ideal scaling with equality."
\end{definition}
\noindent The rest of the section is laid out as follows:
\begin{itemize}
    \item In Section~\ref{sec:continuous-time} we prove the convergence of QHD dynamics in continuous time, and demonstrate that arbitrary convergence rates can be attained for nonsmooth convex optimization. We also prove that the rates we obtain are essentially optimal.
    \item In Section~\ref{sec:simulation-cost} we apply our analysis of time-dependent Hamiltonian simulation in real space to obtain query upper bounds on the complexity of convex optimization with QHD, proving Theorem~\ref{thm:convex-main-specialized}.
    \item In Section~\ref{sec:schedule-invariance} we prove Theorem~\ref{thm:lower-bound-qhd-main}, showing that while the query complexity upper bounds can be attained via schedules with arbitrarily fast continuous-time convergence, this does not improve the query complexity much compared to our upper bounds.
\end{itemize}

\subsection{Continuous-Time Convergence}
\label{sec:continuous-time}
Our first result proves the continuous time convergence of QHD under ideal scaling conditions. The proof is a slight strengthening of \cite[Theorem 1]{leng2023quantum}. In addition to the generalization of ideal scaling, we remark on two additional points: firstly, we retain the pre-factor of the error dependence in addition to the convergence rate. This is crucial in order to determine the dimension dependence of the final query complexity (as we do in the next section). Secondly, we show that the convergence requires only the Lipschitz continuity of the potential, not necessarily differentiability.

\begin{theorem}[QHD convergence for Lipschitz Potentials]\label{thm:QHD_convergence_rate}
    Let $f:\Rmbb^d \rightarrow \Rmbb$ be convex and Lipschitz continuous, with global minimizer $x_\star \in \Xcal$. Let $H_\mathrm{QHD}(t)$ be the~\ref{eq:QHD_Hamiltonian} Hamiltonian with schedules $c_t,m_t,\omega_t$ obeying the ideal scaling conditions (\ref{eq:ideal_scaling}). Then, for any $\Phi_t \in \Dom(-i\nabla, x, f(x))$ solving the Schr\"odinger equation with Hamiltonian $H_\mathrm{QHD}$ and initial condition $\Phi_0$,
    \begin{equation*}
        \Embb_{\abs{\Phi_t}^2} f(x) - f(x_\star) \leq \Ecal_0 \omega_t^{-2}, \quad \forall x_\star \in \mathcal{X}^{\star}
    \end{equation*}
    where 
    \begin{equation*}
        \Ecal_0 \coloneqq \frac12 \langle(-i m_0^{-1} \nabla + \lambda(x - x_\star))^2\rangle_{\Phi_0} + \omega_0^2 \langle f(x) - f(x_\star)\rangle_{\Phi_0}.
    \end{equation*}
\end{theorem}
\begin{proof}
    As in~\cite{leng2023quantum}, we proceed with a quantum Lyapunov argument. Let $p \coloneqq -i\nabla$ and 
    \begin{align}
    \widehat{\Ecal}_t &\coloneqq \frac12 (p/m_t + \lambda(x -x_\star))^2 + \omega_t^2 (f(x) - f(x_\star)).
    \end{align}
    Observe that $\widehat{\Ecal}_t$ is the sum of two positive semi-definite operators (on the appropriate subdomain), hence positive. Define $\Ecal_t \coloneqq \langle \widehat{\Ecal}_t \rangle_t$, where $\langle\cdot \rangle_t$ is shorthand for $\langle\cdot\rangle_{\Phi_t}$ hereon. Then $\Ecal_t \in C^1$ is a non-negative function. We wish to prove that $\Ecal_t$ is a \emph{Lyapunov function}, meaning $\partial_t \Ecal_t \leq 0$. This would imply
    \begin{equation*}
        \omega_t^2 \langle f(x) - f(x_\star)\rangle_t \leq \Ecal_T \leq \Ecal_0
    \end{equation*}
    which amounts to the theorem statement.

    Toward this end, we note the well-known fact that
    \begin{equation*}
        \partial_t \langle \widehat{\Ecal}_t\rangle_t = i \langle [H_\mathrm{QHD}(t),\widehat{\Ecal}_t]\rangle_t + \langle \partial_t \widehat{\Ecal}_t \rangle_t.
    \end{equation*}
    For the first term on the right, we evaluate the commutator as
    \begin{align*}
        i[H_\mathrm{QHD}(t),\widehat{\Ecal}_t] = \frac{i}{4} \frac{c_t}{m_t} [p^2, P_t^2] + \frac{i}{2} \frac{c_t \omega_t^2}{m_t} [p^2, f] + \frac{i}{2} c_t m_t \omega_t^2 [f,P_t^2]
    \end{align*}
    where $P_t \coloneqq p/m_t + \lambda(x -x_\star)$. We handle each term on the right separately, making frequent use of the derivation property of the commutator: $[AB,C] = [A,C] B + A[B,C] $. We also use
    \begin{equation*}
        [p,f] = -i \nabla f
    \end{equation*}
    where $\nabla f$ is understood as the gradient of $f$ where it exists, and $0$ elsewhere. By Rademacher's theorem, the ``elsewhere" set is of measure zero~\cite{rudin1987real}. The first term is
    \begin{equation*} 
        \frac{i}{4} \frac{c_t}{m_t} [p^2, P_t^2] = \lambda \frac{c_t}{m_t}\left(\frac{p^2}{m_t} + \lambda \Re(p\cdot \Delta x)\right)
    \end{equation*}
    where $\Re(A)$ denotes the Hermitian part and $\Delta x = x - x_\star$. Next,
    \begin{equation*} 
        \frac{i}{2} \frac{c_t \omega_t^2}{m_t}[p^2,f] = \frac12 \frac{c_t \omega_t^2}{m_t}(p\cdot \nabla f + \nabla f \cdot p).
    \end{equation*}
    Finally,
    \begin{equation*}
        \frac{i}{2} c_t m_t \omega_t^2 [f,P_t^2] = -\frac12 \frac{c_t \omega_t^2}{m_t}(p\cdot\nabla f + \nabla f \cdot p) - \lambda c_t \omega_t^2 \Delta x\cdot\nabla f.
    \end{equation*}
    The last two are easiest to combine, as there is a cancellation of one of the terms. Adding all the terms together one obtains
    \begin{equation*}
        i\langle[H(t),\widehat{\Ecal}_t]\rangle_t = \lambda c_t \frac{\langle p^2\rangle_t}{m_t^2} + \lambda^2 c_t \frac{\langle\Re(\Delta x \cdot p)\rangle_t}{m_t} - \lambda c_t \omega_t^2 \langle \Delta x \cdot \nabla f\rangle_t.
    \end{equation*}
    Meanwhile, the partial derivative term is given by
    \begin{equation*}
        \langle \partial_t \widehat{\Ecal}_t\rangle_t = -\frac{\dot{m}_t}{m_t} \frac{\langle p^2\rangle_t}{m_t^2} - \lambda \frac{\dot{m}_t}{m_t}  \frac{\langle\Re(\Delta x\cdot p)\rangle_t}{m_t} + 2 \dot{\omega}_t \omega_t \langle f- f(x_\star)\rangle_t.
    \end{equation*}
    Adding the previous two lines term by term, one can see that the $\langle p ^2\rangle_t$ and $\langle \Re(\Delta x \cdot p)\rangle_t$ terms cancel if $c_t$ and $m_t$ are chosen such that
    \begin{equation*}
        \lambda c_t - \frac{\dot{m}_t}{m_t} = 0
    \end{equation*}
    which amounts to the first ideal scaling condition. Taking this from here on, we have
    \begin{align}
    \begin{aligned}
        \partial_t \Ecal_t &= 2 \dot{\omega}_t \omega_t \langle f- f(x_\star)\rangle_t - \lambda c_t\omega_t^2 \langle \Delta x \cdot \nabla f\rangle_t \\
        &= \omega_t^2 \left(2 \frac{\dot{\omega}_t}{\omega_t}\Embb_{x\sim \abs{\Phi_t}^2}[f - f(x_\star)] - \lambda c_t \Embb_{x\sim \abs{\Phi_t}^2}[\Delta x \cdot \nabla f]\right).
    \end{aligned}
    \end{align}
    Next, we make use of the (star-)convexity of the function $f$, which implies that
    \begin{equation} \label{eq:use_convex_convergence}
        0\leq f(x) - f(x_\star) \leq \Delta x \cdot \nabla f(x)
    \end{equation}
    almost everywhere. Hence the expression holds unconditionally under the expectation value, leaving
    \begin{equation}
    \label{eq:lyapunov_inequality}
        \partial_t \Ecal_t \leq \omega_t^2 \Embb_{x\sim\abs{\Phi_t}^2}[\Delta x \cdot \nabla f]\left(2 \frac{\dot{\omega}_t}{\omega_t} - \lambda c_t\right).
    \end{equation}
    By~\eqref{eq:use_convex_convergence}, the expectation value is nonnegative, and thus $\partial_t \Ecal $ is non-positive provided that the second ideal scaling condition holds.
\end{proof}
It is simple to see that Theorem~\ref{thm:QHD_convergence_rate} implies that QHD for nonsmooth convex optimization converges at an arbitrarily fast rate in $t$. To see this, suppose that $r(t)$ is a monotonically increasing function such that $e^{-r(t)}$ is the desired convergence rate. Set $\omega_t^2 = e^{r(t)}$, and set $c_t, m_t$ by ideal scaling with equality. 

For completeness, we document general families of schedules that attain any exponential or polynomial convergence rates.
\begin{corollary}[QHD schedules with exponential convergence]
    \label{cor:exponential-schedules-exist}
    Let $c, m_0, \omega_0 \in \Rmbb_+$. A family of \ref{eq:QHD_Hamiltonian} Hamiltonians with exponential convergence rate $e^{-ct}$ is defined by schedules $c_t = c, m_t = m_0 e^{ct}, \omega_t^2 = \omega_0^2 e^{ct}$, resulting in the Hamiltonian
    \begin{equation*}
    \tag{QHD-exponential}
        H_c(t) = c \left(-e^{-ct} \frac{1}{2 m_0} \Delta + e^{2ct} m_0 \omega_0^2 f\right) 
    \end{equation*}
    for $t \in [0,\infty)$. Let $\Phi_t$ be obtained by simulating the action of $H_c(t)$ on some initial state $\Phi_0$ starting from time $0$. The expectation value $\Embb_{\abs{\Phi_t}^2} [f]$ converges to its minimum at rate $\Ocal(e^{-ct})$.
\end{corollary}

\begin{corollary}[QHD schedules with polynomial convergence]
    \label{cor:polynomial-schedules-exist}
    Let $k, m_0, \omega_0 \in \Rmbb_+$. A family of \ref{eq:QHD_Hamiltonian} Hamiltonians with polynomial convergence rate $t^{-k}$ is defined by schedules $c_t = k/t, m_t = m_0 (t/t_0)^k, \omega_t^2 = \omega_0^2 (t/t_0)^k$, resulting in the Hamiltonian
    \begin{equation*}
    \tag{QHD-polynomial}
        H_k(t) = \frac{k}{t}\left(- (t/t_0)^{-k} \frac{1}{2 m_0}\Delta + (t/t_0)^{2k} m_0 \omega_0^2 f\right) 
    \end{equation*}
    for $t \in [t_0,\infty)$. Let $\Phi_t$ be obtained by simulating the action of $H_k(t)$ on some initial state $\Phi_0$ starting from time $t_0$. The expectation value $\Embb_{\abs{\Phi_t}^2} [f]$ converges to its minimum at rate $\Ocal(t^{-k})$.
\end{corollary}
\noindent The proofs of these lemmas follow by verifying that they satisfy the ideal scaling conditions, and applying Theorem~\ref{thm:QHD_convergence_rate}.

\begin{exmp}[Harmonic oscillator] \label{ex:harmonic_oscillator}
    Let $d = 1$ and $f(x) = x^2/2$. Consider an exponential ideal scaling Hamiltonian $H_c(t)$. This simple example can be solved exactly. The Heisenberg-evolved position operator $x_t$ satisfies
    \begin{equation}
        \ddot{x}_t + c \dot{x}_t + c^2 e^{ct} x_t = 0
    \end{equation}
    whose general solution is given by Bessel functions of order one.
    \begin{equation}
        x_t = k_1 e^{-ct/2}J_1(2 e^{ct/2}) + k_2 e^{-ct/2} Y_1(2 e^{ct/2})
    \end{equation}
    Here $k_1, k_2$ are constant operators whose value depends on the initial state $\Phi_0$. For large $u$, $J_1(u), Y_1(u)$ are of size $u^{-1/2}$. Thus, the distance from optimality $\langle f - f_\star\rangle_t = \langle x_t^2\rangle$ scales as $\Ocal(e^{-3/2 ct})$, which is a factor of $3/2$ better in convergence than Theorem~\ref{thm:QHD_convergence_rate} guarantees. Observe that these dynamics are underdamped (oscillatory), and thus the expectation value $\langle x_t\rangle$ crosses $x_\star =0$ multiple times. These results were observed on the classical side by Su, Boyd, and Cand\`es (indeed, the math is essentially identical in the quadratic case)~\cite{su2016differential}.      
\end{exmp}

\subsection{Cost of Simulation: Upper Bounds on Query Complexity}
\label{sec:simulation-cost}

\begin{algorithm}  
    \caption{Convex Optimization via Simulation of Quantum Hamiltonian Descent} \label{alg:optimize_QHD}
    \hspace*{\algorithmicindent}
    \textbf{Input}: Evaluation oracle $O_f$ for $f$ satisfying conditions of Theorem~\ref{thm:convex-main-general} with parameters $R, R_\infty,\Lambda,\Lambda_\infty$; Schedule $c_t,m_t,\omega_t$ satisfying \ref{eq:ideal_scaling} with equality, with initial conditions $\omega_0= 1$, $m_0 = R^{-1} \sqrt{d/(GR + \Lambda)}$ and $\lambda = G^{-1} R^{-2}(\Lambda + GR)^{3/2}$ (See Definitions~\ref{defn:exponential-schedule},\ref{defn:polynomial-schedule} for examples); Target error parameter $\epsilon$.
    
    \hspace*{\algorithmicindent}
    \textbf{Output}: With probability at least $2/3$, a point $\tilde{x}$ such that $f(\tilde{x}) - f(x_\star) \leq\epsilon, \forall x_\star \in \mathcal{X}^{\star}$.

    \hspace*{\algorithmicindent}
    \textbf{Procedure}:
    \begin{algorithmic}[1]
        \State Compute minimum time $T_\epsilon$ such that $\omega_{T_\epsilon}^{-2} \leq\epsilon/24$.
        \State Prepare initial state $|\Psi_0\rangle$ as an appropriately discretized (as in Lemma~\ref{lem:initial_state_discretization guarantee}) symmetric Gaussian wavefunction
        \begin{equation*}
            \Phi_0(x) = \frac{1}{(2\pi R^2)^{d/4}} e^{-(x - x_0)^2/4R^2}.
        \end{equation*}
        \State Simulate, using Theorem~\ref{thm:master_simulation_thm}, centered at $x_0$ with $\epsilon_{\mathrm{alg}} = \epsilon/6\Lambda_\infty$ and $R = 4R_\infty$, the evolution of $\ket{\Psi_0}$ under $H_{\mathrm{QHD}}$ as in \eqref{eq:QHD_Hamiltonian}, for time $T_\epsilon$. Obtain state $|\Psi_{T_\epsilon}\rangle$.
        \State Measure $|\Psi_{T_\epsilon}\rangle$ in the computational basis to obtain $\tilde{x}$; return $\tilde{x}$.
    \end{algorithmic}
\end{algorithm}

We now prove Theorem~\ref{thm:convex-main-specialized} by combining the above convergence results with the Hamiltonian simulation results of Section~\ref{sec:simulating-schrodinger-dynamics}. These results together lead to a complete resource analysis of discretized QHD. We will actually prove the following stronger theorem.
\begin{theorem}
    \label{thm:convex-main-general}
    Let $f : \Rmbb^d \to \Rmbb$ be a $G$-Lipschitz convex function with $\mathcal{X}^{\star}$ its (non-empty, convex) set of global minima. Let there further be a known point $x_0 \in \Rmbb^d$ and parameters $R, R_\infty, \Lambda, \Lambda_\infty \in \Rmbb_+$ such that for some minimum $x_\star \in \mathcal{X}^{\star}$ it holds that $\norm{x_0 - x_\star} \leq R$, $\norm{x_0 - x_\star}_\infty \leq R_\infty$, $f(x_0) - f(x_\star) \leq \Lambda$, and $\max_{x \in B_\infty(x_0,4R_\infty)} f(x) \leq \Lambda_\infty$.\footnote{Not all of $R,R_\infty,\Lambda ,\Lambda_\infty$ must be supplied in advance, as bounds on these parameters can be obtained from others and Lipschitz continuity. We state the result in terms of the parameters to accommodate cases where stronger bounds than these are available.}
    There exists then a quantum algorithm (Algorithm~\ref{alg:optimize_QHD}) that solves Problem~\ref{prob:nonsmooth-convex} using $\widetilde{\Ocal} \left( \frac{dGR\Lambda_\infty}{\epsilon^2}\right)$  queries to a $\widetilde{\Ocal}\left(\frac{\epsilon^3}{d GR \Lambda_\infty}\right)$-accurate evaluation oracle, or $\widetilde{\Ocal} \left( \left[\frac{d GR \Lambda_\infty}{\epsilon^2}\right]^2\frac{1}{\epsilon}\right)$ queries to a phase oracle.
\end{theorem}

We first prove Theorem~\ref{thm:convex-main-specialized} as a special case of assuming Theorem~\ref{thm:convex-main-general}. We refer the reader to Theorem \ref{thm:master_simulation_thm} for the qubit and gate counts.

\thmconvexmainspecialized*
\begin{proof}
    We note that $f$ satisfies all the conditions of Theorem~\ref{thm:convex-main-general}. By the conditions of Problem~\ref{prob:nonsmooth-convex} we are given a point $x_0$ such that $\norm{x_0 - x_\star} \leq R$, hence, $\norm{x_0 - x_\star}_\infty \leq R$. Furthermore, since $f$ is $G$-Lipschitz, we have that $f(x_0) - f(x_\star) \leq GR$ and $\max_{x \in B_\infty(x_0,4R_\infty)} f(x) \leq GR\sqrt{d}$. We can therefore apply Theorem~\ref{thm:convex-main-general} with $R_\infty = R$, $\Lambda = GR$, and $\Lambda_\infty = GR\sqrt{d}$. The result follows.
\end{proof}

It remains to prove Theorem~\ref{thm:convex-main-general}. To proceed, we need a few intermediate results. First, we must derive a sufficient simulation time. While we have already determined the rate of convergence in continuous time, the error bounds depend also on the value of the chosen Lyapunov function $\Ecal$ at the initial time. Our next lemma provides a bound on $L$, for a properly chosen $m_0, \omega_0$.

\begin{lemma}
\label{lem:initial-lyapunov}
    $f : \R{d} \to \R{}$ be a $G$-Lipschitz continuous, convex function and $\Xcal^{\star}$ its (non-empty, convex) set of global minima. Let there further be a known point $x_0$ and parameters $R, \Lambda \in \Rmbb_+$ such that, for some $x_\star \in \Xcal^{\star}$, it holds that $\norm{x_0 - x_\star} \leq R$, and $f(x_0) - f(x_\star) \leq \Lambda $. Let $H_\mathrm{QHD}$ be as in \eqref{eq:QHD_Hamiltonian} with schedules $c_t, m_t, \omega_t$ satisfying \ref{eq:ideal_scaling}. Finally, let $\Phi_T \in \Hcal$ be the time evolution under $H_\mathrm{QHD}$ of the symmetric Gaussian wavefunction of covariance $\sigma^2 \mathbbm{1}$ centered at $x_0$.
    \begin{equation*}
        \Phi_0(x) = (2\pi \sigma^2)^{-d/4} e^{-(x - x_0)^2/4 \sigma^2}
    \end{equation*}
    Then 
    \begin{equation*}
        \Ecal_0 \leq \frac{d}{(2 m_0 \sigma)^2} + 2 \lambda^2(d \sigma^2 + R^2) + \omega_0^2(GR + \Lambda).
    \end{equation*}
\end{lemma}
\begin{proof} 
    By ideal scaling, can apply Theorem~\ref{thm:QHD_convergence_rate} with initial Lyapunov function
    \begin{equation*}
        \Ecal_0 \coloneqq \frac12 \langle(-im_0^{-1} \nabla + \lambda(x - x_\star))^2\rangle_{\Phi_0} + \omega_0^2\langle f(x) - f(x_\star)\rangle_{\Phi_0}.
    \end{equation*}
    For the term $\omega_0^2 \langle f(x) - f(x_\star)\rangle_{\Phi_0}$, we have
        \begin{align*}
            \omega_0^2\langle f(x) - f(x_\star)\rangle_{\Phi_0} &=\omega_0^2\left(\langle f(x) - f(x_0)\rangle_{\Phi_0}  + \langle f(x_0) - f(x_\star)\rangle_{\Phi_0} \right)\\
                         &\leq \omega_0^2\left(\langle G\norm{x - x_0} \rangle_{\Phi_0} + \Lambda  \right)\\
                         &\leq \omega_0^2\left(GR + \Lambda  \right).
        \end{align*}
        For the term $\langle(-i m_0^{-1} \nabla + \lambda(x - x_\star))^2\rangle_{\Phi_0}$ we observe that,
        \begin{align*}
            \langle(-im_0^{-1} \nabla + \lambda(x - x_\star))^2\rangle_{\Phi_0} &= m_0^{-2}\langle-\Delta \rangle_{\Phi_0} + \lambda^2 \langle (x - x_\star)^2 \rangle_{\Phi_0} \\
            &\quad+ \lambda m_0^{-1}\left(\langle -i\nabla (x - x_\star) \rangle_{\Phi_0} + \langle -i(x - x_\star)\nabla  \rangle_{\Phi_0}\right)\\
            &= m_0^{-2}\langle -\Delta \rangle_{\Phi_0} + \lambda^2 \langle (x - x_\star)^2 \rangle_{\Phi_0}.
        \end{align*}
    To see that the cross terms vanish going into the last line, we observe that, if $\{A,B\} = AB + BA$, 
    \begin{align*}
        \langle\lbrace -i\nabla, (x-x_\star) \rbrace\rangle_{\Phi_0} &= \langle\Phi_0, (-i\nabla)(x-x_\star) \Phi_0\rangle + \langle\Phi_0, (x-x_\star) (-i\nabla)\Phi_0\rangle\\
        &= 2\mathrm{Re}\left( \langle -i \nabla \Phi_0, (x-x_\star) \Phi_0\rangle\right).
    \end{align*}
    This vanishes because $\Phi_0$ is real. Meanwhile, the expectation of $-\Delta$ on a multivariate Gaussian with covariance $\sigma^2 \mathbb{I}$ is given by $d/(2\sigma)^{2}$.
    Finally, for the $x$-term we have,
    \begin{align}
        \langle (x - x_\star)^2 \rangle_{\Phi_0} &\leq 2 \left(\langle (x - x_0)^2\rangle_{\Phi_0} + \langle (x_0 - x_\star)^2\rangle_{\Phi_0}\right) \leq 2(d \sigma^2 + R^2).
    \end{align}
    Putting these pieces together gives the bound on $\Ecal_0$ of the Lemma.
\end{proof}

Lemma~\ref{lem:initial-lyapunov} leads to convergence rates with no hidden factors, and it remains to choose parameters to maximize performance. Rather than optimizing the convergence rate alone, one should also consider the costs of simulation. Towards this end, we provide an lemma that bounds the query complexity in a schedule-independent way. This lemma will also be useful for our lower bounds later.
\begin{lemma} \label{lem:general_QHD_time_complexity}
    Suppose $c_t, m_t, \omega_t \in C^1([0,\infty))$ satisfy the particular ideal scaling conditions
    \begin{equation*}
        \dot{m}_t = \lambda c_t m_t, \qquad \dot{\omega}_t = \vartheta \lambda c_t\omega_t/2
    \end{equation*}
    for $\lambda \in \Rmbb_+$ and $\vartheta \in (0,1]$. Then
    \begin{equation*}
        \int_0^T c_t m_t \omega_t^2 dt = \frac{m_0 \omega_0^2}{(1 + \vartheta) \lambda}\left[(\omega_T/\omega_0)^{2(1+\vartheta)/\vartheta} - 1 \right].
    \end{equation*}
    For sufficiently large $\omega_T$, the right hand side is minimized when $\vartheta = 1$ (ideal scaling with equality), and for this value
    \[
    \int_0^T c_t m_t \omega_t^2 dt = \frac{m_0\omega_0^2}{2\lambda}(\omega_T^4/\omega_0^4 - 1).
    \]
\end{lemma}
\begin{proof}
    We first write the integral as
    \begin{align} \label{eq:make_I}
        \int_0^T c_t m_t \omega_t^2 dt &= m_0 \omega_0^2 \int_0^T c_t \frac{m_t \omega_t^2}{m_0 \omega_0^2} dt \equiv m_0 \omega_0^2 I(T).
    \end{align}
    From the fundamental theorem of calculus,
    \begin{align} \label{eq:2nd_It}
    \begin{aligned}
        \frac{m_t \omega_t^2}{m_0 \omega_0^2} &= \int_0^t \frac{d}{d\tau} \frac{m_\tau \omega_\tau^2}{m_0 \omega_0^2} d\tau + 1 \\
        &= \frac{\lambda}{m_0 \omega_0^2} \left(\int_0^t \dot{m}_\tau \omega_\tau^2 d\tau + \int_0^t 2 m_\tau \omega_t \dot{\omega}_\tau\right) + 1 \\
        &= (1+\vartheta) \lambda \int_0^t c_\tau  \frac{m_\tau \omega_\tau^2}{m_0 \omega_0^2}d\tau + 1 = (1+\vartheta) \lambda I(t) + 1
    \end{aligned}
    \end{align}
    where in the last line we employed our specific ideal scaling conditions. Comparing~\eqref{eq:2nd_It} with~\eqref{eq:make_I}, one has
    \begin{align*}
        I(T) = (1 + \vartheta) \lambda \int_0^T c_t I(t) dt + \int_0^T c_t dt
    \end{align*}
    for all $T\geq 0$. Taking a derivative of both sides with respect to $T$, we see that $I$ solves a simple differential equation
    \begin{align*}
        &\dot{I}(t) = (1 + \vartheta) \lambda c_t I(t) + c_t, \\
       \implies & \frac{\dot{I}(t)}{(1 + \vartheta) \lambda I(t) + 1} = c_t, \\
       \implies & \frac{d}{dt}\ln((1 + \vartheta) \lambda I(t) + 1) = (1 + \vartheta) c_t,
    \end{align*}
    whose solution is
    \begin{align*}
        I(T) = \frac{1}{(1 + \vartheta)\lambda }\left(e^{(1 + \vartheta) \int_0^T c_t dt} - 1\right).
    \end{align*}
    Using the 2nd ideal scaling condition, $\int_0^T c_t dt = \vartheta^{-1} \log(\omega_T^2/\omega_0^2)$. Thus,
    \begin{align*}
        \int_0^T c_t m_t \omega_t^2 dt = \frac{m_0 \omega_0^2}{(1 + \vartheta) \lambda}\left[ (\omega_T/\omega_0)^{2(1+\vartheta)/\vartheta} - 1 \right].
    \end{align*}
\end{proof}

As a technical aside, we must connect our bounds on expectation value of $f$ to the probability of measuring a near-optimal $x$. This can be accomplished with Markov's inequality.
\begin{lemma}\label{e:sim_to_measurement_guarantee}
    Let $\epsilon > 0$, and $\Pmbb$ be a probability measure over $\Rmbb^d$. Suppose that $x_\star \in \Xcal^{\star}$ is a minimizer of a measurable function $f: \R{d} \rightarrow \Rmbb$. If 
    \begin{equation*}
        \Embb [f(x)] - f(x_\star) \leq \epsilon/3,
    \end{equation*}
    then $\Pmbb[f(x) - f(x_\star) < \epsilon] \geq 2/3$.
\end{lemma}
\begin{proof}
    For $x \sim \Pmbb$, $f(x) - f(x_\star)$ is a non-negative random variable. By Markov's inequality,
    \begin{equation*}
        \Pmbb\big[f(x) - f(x_\star) \geq 3(\Embb[f(x)] -  f(x_\star)) \big]   \leq 1/3\\
    \end{equation*}
    which implies by our assumptions on $\Embb[f]$ that $\Pmbb[f(x) - f(x_\star) \geq \epsilon] \leq 1/3$.
\end{proof}

We are finally ready to prove Theorem~\ref{thm:convex-main-general}. We note that our proof will assume the applicability of our simulation algorithm to continuous systems. We justify this in Section~\ref{sec:justify-periodic}. The query bounds we derive will be valid whenever the simulation algorithm is applicable.

\begin{proof}[Proof of Theorem~\ref{thm:convex-main-general}]
    By construction (Algorithm~\ref{alg:optimize_QHD}) we simulate the dynamics over $B_\infty(x_0, 4R_\infty) \equiv x_0 +[-4R_\infty,4R_\infty]^d$. By assumption, $\norm{f}_\infty \leq \Lambda_\infty$. From the guarantees of Theorem~\ref{thm:master_simulation_thm}, we have
    \begin{align*}
        \bigg\lvert (2N)^{-d}\sum_{x_\jmbf \in \Gcal} {h}(8R_\infty x_\jmbf)\lvert \Psi_{T_\epsilon}(x_\jmbf)\rvert^2 - \int_S{h}(8R_\infty y)\lvert \Phi(y, T_\epsilon)\rvert^2 dy \bigg\rvert \leq\lVert h \rVert_{\infty}\epsilon_{\mathrm{alg}}.
    \end{align*}
    for any Lipschitz function $h$ with the same Lipschitz constant, $G$, as the potential. We select $h(x) = f(x + x_0)$, hence $\norm{h}_\infty \leq\Lambda_\infty$. The guarantee translates into
    \begin{align*}
        \bigg\lvert (2N)^{-d}\sum_{x_\jmbf \in \Gcal} {f}(x_0 + 8R_\infty x_\jmbf)\lvert \Psi_{T_\epsilon}(x_\jmbf)\rvert^2 - \int_S{f}(x_0 + 8R_\infty y)\lvert \Phi(y, T_\epsilon)\rvert^2 dy \bigg\rvert \leq\Lambda_\infty \epsilon_{\mathrm{alg}} = \epsilon/6.
    \end{align*}
    We note that the first term on the left hand side of the above bound represents the average potential at a point sampled from the state with discrete support obtained by simulation, and the right hand side represents the expectation of the potential at a random point sampled according to the wavefunction $\Phi_{T_\epsilon}$. Thus the point $\tilde{x}$ returned by the algorithm satisfies
    \begin{align*}
        \Embb[f(\tilde{x})] \leq\frac{\epsilon}{6} + \left(\Embb_{\sim |\Phi_{T_\epsilon}|^2} [f(x)] - f(x_\star)\right) \leq \frac{\epsilon}{3},
    \end{align*}
    where the last inequality follows from our choice of $T_\epsilon$ and Lemma~\ref{lem:initial-lyapunov}. By Lemma~\ref{e:sim_to_measurement_guarantee}, $\tilde{x}$ satisfies the output guarantee with probability at least $2/3$.

    Finally, we determine the cost of running the simulation. Let $b(t)$ be the time dependent coefficient for the potential $f$. By Lemma~\ref{lem:general_QHD_time_complexity}, we have
    \begin{align*}
        \norm{b}_1 = \int_0^{T_\epsilon} c_t m_t \omega_t^2 dt \leq\frac{m_0}{2\lambda\omega_0^2} \omega_{T_\epsilon}^4.
    \end{align*}
    From Theorem~\ref{thm:QHD_convergence_rate}, we must choose $T_\epsilon$ such that $\omega_{T_\epsilon}^2 = \Ocal(\Ecal_0/\epsilon)$. Using the bound on $\Ecal_0$ from Lemma \ref{lem:initial-lyapunov}, we find that
    \begin{equation}
        \norm{b}_1 = \Ocal\left(\frac{m_0 \Ecal_0^2}{\lambda \omega_0^2 \epsilon^2}\right) \lesssim \frac{1}{\epsilon^2}\left[\sqrt{\frac{m_0}{\lambda}} \Big(\frac{d}{\omega_0(2 m_0 \sigma)^2} + 2 \frac{(\lambda\sigma)^2 + (\lambda R)^2}{\omega_0} + \omega_0 ( GR + \Lambda)\Big)\right]^2
    \end{equation}
    We seek to minimize the square-bracket expression, at least approximately. Let $\Lambda' \equiv GR + \Lambda$. We seek to minimize $d$ dependence overall. Working from the last term, we try the ansatz $\omega_0 \sqrt{m_0/\lambda} = d^p \sqrt{GR} / \Lambda'$ for some $p > 0$. Solving for $\lambda$ and plugging in to the bracketed expression gives
    \begin{align}
    \begin{aligned}
        \sqrt{\frac{m_0}{\lambda}} \frac{\Ecal_0}{\omega_0} &\lesssim d^p \sqrt{GR} + \sqrt{\frac{m_0}{\lambda}}\left(\frac{d}{(m_0 \sigma)^2 \omega_0} + \frac{\lambda^2}{\omega_0}(d \sigma^2 + R^2)\right) \\
        &\lesssim d^p \sqrt{GR} + \frac{d^p \sqrt{GR}}{\Lambda'} \left(\frac{d}{(m_0 \omega_0 \sigma)^2} + \frac{\lambda^2}{\omega_0^2}(d \sigma^2 + R^2)\right).
    \end{aligned}
    \end{align}
    A natural choice for initial variance is $\sigma^2 = R^2/d$. Focusing on the terms in parentheses, it suffices that
    \begin{align}
    \begin{aligned}
        \frac{d^2}{(m_0 \omega_0 R)^2} + \frac{\lambda^2 R^2}{\omega_0^2} &\lesssim \Lambda' \\
        \frac{d^2}{(m_0 \omega_0 R)^2 \Lambda'} + \frac{1}{\Lambda'}\left(\frac{m_0 \omega_0 \Lambda'^2 R}{d^{2p} G R}\right)^2 &\lesssim 1 \\
        \frac{d^2}{x} + \frac{x}{d^{4p}} \left(\frac{\Lambda'}{GR}\right)^2 &\lesssim 1
    \end{aligned}
    \end{align}
    where $x \coloneqq (m_0 \omega_0 R)^2 \Lambda'$. To ensure each term is $O(1)$ in $d$, we require $x \in \Omega(d^2)$ and $p \geq 1/2$. Taking the minimal value $p = 1/2$ and $x = d$, this leads satisfies our requirements. Altogether (remembering the definition of $\Lambda'$),
    \begin{equation}
        \norm{b}_1 \lesssim \frac{d GR}{\epsilon^2} \left(1 + \frac{\Lambda}{GR}\right)^4.
    \end{equation}
    While there is freedom in at least one of the parameters still, we may for concreteness choose $\omega_0 = 1$, leading to the following remaining choices of parameters.
    \begin{equation}
        m_0 = \frac{d}{R\sqrt{\Lambda'}} \qquad \lambda = \frac{m_0 \Lambda'^2}{dGR} = \frac{\Lambda'^{3/2}}{G R^2}
    \end{equation}
    The resource requirements are then given by Theorem~\ref{thm:master_simulation_thm}. For example, the query complexity is $\norm{b}_1 \Lambda_\infty$ up to logarithmic factors.
\end{proof}

\subsubsection{Justifying the Applicability of the Simulation Algorithm}
\label{sec:justify-periodic}

For the initial state we will consider a smoothed-out version of a truncated Gaussian. 
Consider the Gaussian wavepacket as (square root of) our ``ideal'' initial state over $\Rmbb^d$:
\begin{align*}
G_{\sigma}({Rx + \mu}) = \frac{e^{-x^2/(4\sigma_1^2)}}{(\sqrt{2\pi}\sigma_1)^{d/2}},
\end{align*}
where $\sigma_1 < 1$, and we will use arithmetic to center it around the desired point $x_0$ (i.e. $\mu = x_0$).
However, of course we cannot exactly load this in the simulation box. Hence, we prove the following results that shows there is a sufficiently good proxy.

\begin{lemma}
\label{lem:initial_state_discretization guarantee}
Consider $\sigma_1 < 1$, $\mu$, and the Gaussian wavepacket $G_{\sigma}(R x + \mu)$. If
\begin{align*}
    \log_2(N) = \Ocal\left(d\log\left(\frac{(d^{9/4}R^{3})^{1/d}}{\sigma^{3/4}\epsilon^{3/d}}\right)\right)
\end{align*} 
then for $R = {\Omega}\left(\log(d/\epsilon)\right)$, there is a normalized wave function in $\Hcal_N$ with compact support in $S$ that is at most $\Ocal(\epsilon)$ in $L_2(\Rmbb^d)$ distance from $G_{\sigma}(R x + \mu)$.
\end{lemma}
\begin{proof}
We consider truncating the Gaussian to the box $S$ in the following way
\begin{align*}
G_{\sigma_1, R}({x}) = \frac{G_{\sigma_1}(R x + \mu)\mathbb{1}_S(x)}{\Zcal},
\end{align*}
where $\Zcal$ is the $L_2$ normalization constant. Note that by standard Gaussian concentration over the $\ell_\infty$ ball and $\sigma_1 < 1$ if $R = {\Omega}\left(\log(d/\epsilon)\right)$ for $\epsilon < 1/2$, then $\Zcal \geq \sqrt{1/2}$.

However, even though the truncated Gaussian retains the smoothness properties of the original Gaussian within the interior, it looses the fast Fourier decay. It is well known that multiplying a function by a rectangular window convolves the Fourier transform by a function with decay like $\Ocal(1/\norm{\nmbf})$, which ends up dominating the rate of decay. This is significantly worse than the Fourier spectrum of the original Gaussian, which is also a Gaussian. Hence, post truncation we convolve it by a rectangular mollifier $\mathcal{M}_{\delta, d}$ (Definition \ref{defn:mollifer}).

Due to the separable natures of the Gaussian and mollifiers, we only need to perform $d$, one-dimensional convolutions as a classical preprocessing step. This is leads to the function:
\begin{align*}
    G_{\sigma_1, R, \sigma}(x) = \int_SG_{\sigma_1, R}(y)\mathcal{M}_{\sigma, d}(x-y)dy.
\end{align*}
However, due to the lack of Lipschitzness, we cannot apply Lemma \ref{lem:mollification} directly.

Since the coordinate system will be centered around the Gaussian, the periodic extension of $G_{\sigma_1, R}$ will be continuous. Suppose that $\sigma = o\left(1/R\right)$, so that
\begin{align*}
    \lvert  G_{\sigma_1, R, \sigma}(x) -  G_{\sigma_1, R}(x) \rvert &\leq \sqrt{2}(\sqrt{2\pi}\sigma_1)^{-d/2}\int_S\lvert e^{-\lVert Rx \rVert^2} - e^{-\lVert Ry\rVert^2} \rvert \mathcal{M}_{\sigma, d}(x-y) dy \\
    &\leq \sqrt{2}(\sqrt{2\pi}\sigma_1)^{-d/2}e^{-\lVert Rx \rVert^2}\int_S\lvert 1 - e^{R^2\lvert \lVert y\rVert^2 - \lVert x \rVert^2\rvert} \rvert \mathcal{M}_{\sigma, d}(x-y) dy\\
    &=\sqrt{2}(\sqrt{2\pi}\sigma_1)^{-d/2} \Ocal\left(R^2\sigma^2d\right),
\end{align*}
so we need 
\begin{align*}
    \sigma = \Ocal\left(d^{-1/2}R^{-1}(\sqrt{2\pi}\sigma_1)^{d/4}\epsilon\right).
\end{align*}
Note due to the symmetry of the Gaussian, $G_{\sigma_1, R}$ has a continuous periodic extension, so the above is an $\lVert \cdot \rVert_{\infty}$ bound over $S$.

The Sobolev bound effectively follows the same proof as the Sobolev bound in \ref{lem:mollification}. Specifically, we have
\begin{align*}
    \lvert G_{\sigma_1, R, \sigma}\rvert_{\Hcal^{m}} &= \Ocal\left(\left(\frac{Cd}{\sigma}\right)^{3d}\right) = \Ocal\left(\left(\frac{CRd^{3/2}}{(\sqrt{2\pi}\sigma_1)^{d/4}\epsilon}\right)^{3d}\right).
\end{align*}
We then plan to load the state 
\begin{align*}
    \widetilde{\Phi_0} = \frac{I_N[G_{\sigma_1, R, \sigma}]}{\lVert I_N[G_{\sigma_1, R, \sigma}]\rVert},
\end{align*}
where we can use Lemma \ref{lem:gen_periodic_initial_condition} with $\Phi_0 = G_{\sigma_1, R}$. However, due to use of mollification we do not have a good bound on the Sobolev norms of constant order. 

Recall that Lemma \ref{lem:gen_periodic_initial_condition} guarantees that the required number of grid points per dimension grow as
\begin{align*}
    N = \Ocal\left( \frac{d^{-1/4}\lvert \Phi_0\rvert_{\Hcal^d}^{1/d} + \lvert \Phi_0 \rvert_{\Hcal^1}}{\epsilon} \right). 
\end{align*}
The $\lvert \Phi_0 \rvert_{\Hcal^1}$ comes from the Fourier series truncation component, which using Lemma \ref{lem:trunc_high_sobolev} we can swap out to get an overall required grid spacing of just
\begin{align*}
    N = \Ocal\left( \frac{\lvert \Phi_0\rvert_{\Hcal^d}^{1/d}}{\epsilon} \right) = \Ocal\left(\frac{C'' R^3d^{9/2}}{(\sqrt{2\pi}\sigma_1)^{3d/4}\epsilon^3}\right)
\end{align*}
which implies
\begin{align*}
    \log_2(N) = \Ocal\left(d\log\left(\frac{(d^{9/4}R^{3})^{1/d}}{\sigma_1^{3/4}\epsilon^{3/d}}\right)\right).
\end{align*}
This yields
\begin{align*}
    \lVert G_{R, \sigma_1} - \widetilde{\Phi_0}\rVert = \Ocal(\epsilon).
\end{align*}
Since we also know that
\begin{align*}
    \lVert G_{\sigma_1} - G_{\sigma_1, R}\rVert  =\Ocal(\epsilon),
\end{align*}
the result follows from the triangle inequality.
\end{proof}

The state constructed in Lemma~\ref{lem:initial_state_discretization guarantee} can be chosen to be $O(\epsilon)$-close in trace distance to the desired starting Gaussian $\Phi_0$ (see Algorithm~\ref{alg:optimize_QHD}), by choosing $\sigma = R_2/\sqrt{d}$ and $R \coloneqq 4R_\infty$. Thus our bounds on the Lyapunov functional at time $0$ remain valid for the discrete case. The next difficulty is that our potential function is not necessarily periodic and thus the guarantees of Theorem~\ref{thm:master_simulation_thm} do not directly apply. However, we show in Section~\ref{sec:lipschitz-guarantee} that the guarantees of Theorem~\ref{thm:master_simulation_thm} apply for potentials that are only Lipschitz given an assumption on ``low-leakage" evolution that we restate below.
\lowleakage*
\noindent Specifically, in order for the guarantees to apply, we must have the low-leakage evolution condition with $\delta_0 = 1/R$ where the simulation is an $\ell_\infty$ ball of radius $R$. Substituting our choices of parameters for the simulation, we require effectively the following condition on the continuous wavefunction $\widetilde\Phi_t$ (truncated to lie in $B_\infty(x_0,4R_\infty - \delta)$ and re-normalized) for all $0 < \delta < 8$,
\begin{align*}
    \int_{\left[B_\infty(x_0,4R_\infty - \delta)\right]^c} |\widetilde\Phi_t|^2 \leq\frac{\delta \epsilon^2}{2 d\Lambda_\infty (GR + \Lambda )}
\end{align*}
Intuitively, the low-leakage condition ensures that the simulated wavefunction places very low amplitude near the boundaries of the region of simulation. This allows the guarantees for periodic potentials to be extended to Lipschitz potentials, because the discontinuity at the boundary has little effect on the simulation process.

Also note that the $R$-restriction of our convex function $f$, is both almost-surely differentiable and convex and on its interior. Hence, the argument with Eq.~\eqref{eq:lyapunov_inequality} still goes through for the $R$-restriction.

Up to this point, our justifications are fully rigorous. We now argue that the wavefunction evolving under \ref{eq:QHD_Hamiltonian} is indeed likely to satisfy the low-leakage conditions. We first consider the initial state, $\Phi_0$ as defined in Algorithm~\ref{alg:optimize_QHD}, and suppose that $R_\infty = R_2$\footnote{This is the case whenever we are only given an upper bound on $\norm{x - x_0}_2$, as in Theorem~\ref{thm:convex-main-general}, and the $\ell_\infty$ bound is derived from this.}. By standard concentration arguments, the total probability placed by the distribution $|\Phi_0|^2$ outside the ball $B_\infty(x_0,4R_\infty)$ is $\Ocal((d/R_\infty)e^{-8d})$. Therefore, truncating the distribution applies a multiplicative correction of $1 - \Ocal((d/R_\infty)e^{-8d})$ to the density. In the limit of large $d$, the total probability that the truncated distribution places on $\left[B_\infty(x_0,4R_\infty - \delta)\right]^c$ is upper bounded by that placed by the untruncated distribution on $\left[B_\infty(x_0,4R_\infty)\right] /\left[B_\infty(x_0,4R_\infty - \delta)\right]$, multiplied by $1 + \Ocal((d/R_\infty)e^{-8d})$. Using the asymptotically tight Mills inequality for the tail bounds of standard normal random variables, and the union bound it can be verified that this expression is upper bounded by $\Ocal(\delta d e^{-d} / R_\infty)$. Thus, the low-leakage condition is satisfied whenever $e^{-d} < 2d^2 R_\infty \Lambda_\infty(GR + \Lambda )/\epsilon^2$. By a simple modification of the above argument, the condition can also be seen to apply whenever the real space probability distribution corresponding to the state is sub-gaussian and centered at a point in $B_\infty(x_0,2R_\infty)$. More generally, the condition should be expected to hold whenever the state is concentrated away from the boundary of $B_\infty(x_0,4R_\infty)$. Due to the convexity of the potential, and the Lyapunov argument (Theorem~\ref{thm:QHD_convergence_rate}), the expected function value is decreasing once $\beta_t > 0$. As a consequence, the state is eventually concentrated within the sub-level sets of $f$, which are nested due to the convexity of the potential. For a function with quadratic growth away from its global minima, it can be confirmed that the variance of the distribution is strictly decreasing.

Finally, we note that if the low-leakage condition fails to be satisfied, then the simulation algorithm is simply not applicable to QHD (as is also the case for any algorithm based on a spectral method) and the algorithm is not guaranteed to succeed. Whenever the simulation algorithm is applicable, however, the query complexity guarantees hold without reservation.

\subsection{Schedule Invariance: Lower Bounds for Convex Optimization via QHD}
\label{sec:schedule-invariance}

In this section, we investigate the dependence of the query complexity derived in previous sections on choice of dynamics within the QHD framework. In particular, is there room to reduce the complexity, by careful optimization of the schedule functions $c_t, m_t, \omega_t$? We answer this question in the negative, under some reasonable assumptions on the query complexity of Hamiltonian simulation and the initial state. The arguments will closely mirror those for deriving the upper bounds on $\norm{b}_1$. These lower bounds, however, are highly conditional on the convergence rate matching the one from the Lyapunov argument. With faster convergence rates than $\omega_T^{-2}$, lower complexities may indeed be achievable.

\asmnofastforward*

We use the following helper lemma, which is a specific version of the Heisenberg uncertainty principle in $d$ dimensions.
\begin{lemma}
    Let $\phi \in \Dom(x_j p_k)\cap\Dom(p_jx_k) \subseteq L_2(\Rmbb^d)$ for all $j,k\in[d]$, with $p = -i\nabla$. Suppose $\langle x \rangle_{\phi} = x_0$ for some $x_0 \in \Rmbb$ and suppose $\langle p \rangle_{\phi} = 0$. Let $\sigma^2 = \langle (x - x_0)^2 \rangle_\phi$ be the variance of $\phi$. Then the variance of the momentum can be bounded as $\langle p^2 \rangle_\phi = \Omega(d^2/\sigma^2)$.
\end{lemma}
\begin{proof}
    Observe that $\langle(x-x_0)^2\rangle_\phi = \sum_{i=1}^{d} \langle (x - x_0)_i^2 \rangle_\phi$ and $\langle p^2 \rangle_\phi = \sum_{i=1}^{d} \langle p_i^2 \rangle_\phi$. Using the Robertson form of the Heisenberg Uncertainty Principle in one dimension~\cite{hall2013quantum}, we notice that $\langle p_i^2 \rangle_\phi \geq \varsigma/\langle (x - x_0)_i^2 \rangle_\phi$ where $\varsigma$ is an absolute constant.

    We then note that,
    \begin{align*}
        \langle p^2 \rangle_\phi \cdot \langle (x - x_0)^2 \rangle_\phi &= \left( \sum_{i=1}^{d} \langle p_i^2 \rangle_\phi\right) \cdot \left( \sum_{i=1}^{d} \langle (x - x_0)_i^2 \rangle_\phi\right), \\
        &\geq \varsigma \left( \sum_{i=1}^{d} \langle (x - x_0)_i^{-2}  \rangle_\phi\right) \cdot \left( \sum_{i=1}^{d} \langle (x - x_0)_i^2 \rangle_\phi\right) \geq \varsigma d^2,
    \end{align*}
    where the last inequality arises from the well known AM-HM inequality. The result follows.
\end{proof}

\lowerboundqhdmain*
\begin{proof}
Assume, without loss of generality, that $\Phi_0$ has expected value $x_0$ in position. Using Lemma~\ref{lem:general_QHD_time_complexity} and Assumption~\ref{asm:no-fast-forward}, the total query complexity of simulation is given by
\begin{equation*}
        \Lambda_f \int_0^T c_t m_t \omega_t^2 dt = \Lambda_f\frac{m_0 \omega_0^2}{(1 + \vartheta)\lambda}\left[ (\omega_T/\omega_0)^{2(1+\vartheta)/\vartheta} - 1 \right] \geq \Lambda_f\frac{m_0 \omega_0^2 }{2\lambda}\left((\omega_T/\omega_0)^4 - 1\right).
\end{equation*}
Since the convergence in continuous-time follows from Theorem~\ref{thm:QHD_convergence_rate}, in order to ensure error $\leq\epsilon$, the schedule and simulation time must be concurrently chosen to ensure that $\Ecal_0 \omega_T^{-2} \leq\epsilon$. Note that throughout the argument, $\Ecal_0$ has a possible dependence on problem parameters including $d,\epsilon$. As a consequence, the total query complexity can be lower bounded by 
\begin{align}
\label{eq:query-lower-bound}
    \frac{\Lambda_f m_0 \omega_0^2 }{2\lambda}\left(\frac{\Ecal_0^2}{\epsilon^2 \omega_0^4} - 1\right).
\end{align}
We recall that for a state satisfying our conditions, and with initial variance $\sigma^2 = O(R^2)$ (all expectations with respect to $\Phi_0$):
\begin{align*}
    \Ecal_0 &\coloneqq \frac12 \langle(-im_0^{-1} \nabla + \lambda (x - x_\star))^2\rangle  + \omega_0^2\langle f(x) - f(x_\star)\rangle  \\
    &= \frac12 m_0^{-2}\langle p^2 \rangle  + \frac12 \lambda^2\langle (x - x_0)^2 \rangle  + \omega_0^2\langle f(x) - f(x_\star)\rangle  \\
    &= \Omega \left( \max\left(m_0^{-2}\langle p^2 \rangle , \lambda^2\langle (x - x_\star)^2 \rangle , \;\omega_0^2\langle f(x) - f(x_\star)\rangle \right)\right) \\
    &= \Omega\left( \max\left(d/(m_0\sigma)^2, \lambda^2\langle (x - x_\star)^2 \rangle , \;\omega_0^2\langle f(x) - f(x_\star) \rangle \right)\right).
\end{align*}
Let $l(d,\epsilon)$ be a desired lower bound on $\Ecal_0$. In order for this bound to be satisfied for a general $f$ satisfying only the conditions of Problem~\ref{prob:nonsmooth-convex}, we must ensure that:
\begin{itemize}
    \item $d /(m_0 \sigma)^2 = \Ocal(l(d,\epsilon))$, therefore $m_0  = \Omega(\sigma \sqrt{d/l(d,\epsilon)})$.
    \item $\lambda^2 \langle (x - x_\star)^2 \rangle  = \Ocal(l(d,\epsilon))$, therefore $\lambda^{-1} = \Omega(\sqrt{(R^2 + \sigma^2)/l(d,\epsilon)})$.
    \item $\omega_0^2\langle f(x) - f(x_\star) \rangle  = \Ocal(l(d,\epsilon))$, therefore $\omega_0^2 = \Ocal(l(d,\epsilon)/GR)$.
\end{itemize}
Plugging these requirements into Equation~\ref{eq:query-lower-bound} we obtain the lower bound
\begin{align*}
    \Omega\left( \Lambda_f \cdot \frac{1}{\sigma} \sqrt{\frac{d}{l(d,\epsilon)}} \cdot \frac{\sqrt{R^2 + \sigma^2}}{\sqrt{l(d,\epsilon)}} \cdot \frac{GR}{l(d,\epsilon)}\cdot\left(\frac{l(d,\epsilon)^2}{\epsilon^2} - \frac{l(d,\epsilon)^2}{G^2 R^2}\right)\right).
\end{align*}
We observe that the second term in the parentheses is dominated when $\epsilon \to 0$, and the remaining term is minimized when $\sigma = \Theta(R)$. Thus simplifying, we obtain a query complexity lower bound of $\Omega(\sqrt{d} GR \Lambda_f / \epsilon^2)$.
\end{proof}

Although the assumption of $\Psi_0(x)$ having real amplitudes appears natural, the reader might wonder whether other choices can avert the lower bound. A hint in this direction is that, domain questions briefly aside, the operator $(p + x)^2$ in the Lyapunov operator $\widehat{\Ecal}_0$ has eigenfunction $e^{i x^2/2}$ with eigenvalue 0, independent of $d$. This suggests there may be some proper, normalized wavefunction which minimizes the Lyapunov constant. This possibility may be worth exploring in future research; however, we argue that such an approach may amount to shifting complexity from time evolution to state preparation. The state in question appears to be highly oscillatory, whereas a real, log-concave initial state such as a Gaussian admits a number of well-studied methods for preparing.

\section{Application to Stochastic Optimization} \label{sec:stochastic-convex}
\sh{TODO for shouvanik: Take a pass over Problem 2 and add a failure probability, the guarantee should not be only in expectation.}
To this point, we have considered generic noisy evaluations, e.g., binary oracles such that
\begin{align*}
    O_f|x \rangle|0\rangle \rightarrow |x \rangle|\widetilde{f}(x)\rangle,
\end{align*}
where $\lVert \widetilde{f}(x) - f(x)\rVert_{\infty} < \epsilon_f$. One can extend this notion to stochastic convex functions, in which case the oracle acts as 
\begin{align} \label{eq:stochastic_oracle}
 O_f|x \rangle|0\rangle \rightarrow |x \rangle \int_{\xi \in \Xi} \sqrt{ p_x (\xi)} |f(x) + \xi_x\rangle d \xi,
\end{align}
where, for $x \in \R{d}$, $p_x(\cdot)$ denotes the probability density of $\xi_x$, which for each $x$ has mean zero and variance bounded by some absolute constant $\sigma$. We recall the problem of unconstrained stochastic convex optimization.
\probstochasticconvex*
Using a quantum algorithm for mean estimation \cite{montanaro2015quantum} in combination with our results for unconstrained convex optimization immediately yields the following. 
\thmconvexstochastic*
\begin{proof}
    Again, we normalize the input space and output space so that $f$ is 1-Lipschitz and the diameter of the ball is 1. Hence, if we aim to solve the original problem to precision $\epsilon$, we must solve the normalized problem to precision $\varepsilon = \epsilon/GR$. 
    
    Assuming oracle access $O_f$ to $f$ of the form in Eq.~\eqref{eq:stochastic_oracle}, the mean estimation algorithm of Montanaro~\cite{montanaro2015quantum} requires $\widetilde{\Ocal}(1/\varepsilon_f)$ applications of $O_f$ to estimate $\Embb[O_f (x)] = f(x)$ within error $\varepsilon_f$. Meanwhile, Theorem~\ref{thm:convex-main-specialized} asserts that $\varepsilon_f = \widetilde{\Ocal}(\varepsilon^3/d^{1.5})$ (recall $\varepsilon_f$ is the rescaled error), and the total number of queries is thus
    \begin{equation*} 
        \widetilde{\Ocal} \left( \frac{d^{1.5}}{\varepsilon^3} \cdot d^{1.5} \left( \frac{1}{\varepsilon} \right)^2 \right) = \widetilde{\Ocal} \left( d^{3} \left( \frac{GR}{\epsilon} \right)^5 \right),
    \end{equation*}
    which is the theorem statement. 
\end{proof}

\subsection{Query Complexity of the Li-Risteski Algorithm}
\label{sec:li-risteski-scaling}

As a comparison point, let's consider the classical algorithm from Li and Risteski~\cite{risteski2016algorithms} for approximately convex optimization. To our knowledge, this is the most competitive classical algorithm in our setting of noisy zeroth-order convex optimization. Although Li and Risteski only demonstrates polynomial scaling, here we supply the detailed scaling to obtain the first row of Table~\ref{tab:compare-noisy}. See \cite[Algorithm 1]{risteski2016algorithms} for a description of the algorithm. Since we are interested in the case of unconstrained optimization, we will make this simplification while analyzing the query complexity. We will scale the input and output space such that $G= 1$ and $R=1$, as a consequence we must perform the optimization up to error $\varepsilon \coloneqq \epsilon/GR$

The algorithm executes $T = \Theta(d^2/\varepsilon^4)$ iterations. The primary computational bottleneck in each step is estimating the expectation
\begin{align*}
    f_t = \Embb_{w \in S^d} \left[\frac{d}{r}f(x_t + rw) w\right] = \Embb_{w \in S^d} \left[\frac{d}{r}(f(x_t + rw) - f(x_t)) w\right],
\end{align*}
within accuracy $\Theta(\varepsilon)$ in $\ell_2$ norm. Here $x_t$ is a point in $B_2(x_0,R)$, $r = \Theta(\varepsilon)$ and $w$ is uniformly sampled from the $d$-dimensional hypersphere $S_d$. In the worst case,
\begin{align*}
    \Sigma = \mathrm{Var}_{w \in S^d} \left[\frac{d}{r}(f(x_t + rw) - f(x_t))w\right] \leq d^2 \cdot \mathrm{Var}_{w \in S^d} \left[\norm{w}w\right] = d^2 \mathrm{Var}_{w \in S^d} \left[w\right] = d \Imbb_d,
\end{align*}
where the inequality from symmetry and the Lipschitz continuity of $f$. Using results from Hopkins~\cite{hopkins2020mean} on optimal mean estimation in Euclidean norm, we observe that given some fixed desired success probability, the $\ell_2$ norm accuracy of the estimated mean decays with the number of samples $N$ as 
\[ 
    \Ocal\left(\sqrt{\mathrm{Tr}(\Sigma)/N} + \sqrt{\norm{\Sigma}/N}\right) = \Ocal(d/\sqrt{N}).
\]
In order to obtain $\Ocal(\varepsilon)$ error in $\ell_2$ norm, the number of samples taken per round must be $\sim d^2/\varepsilon^2$. Each sample requires a noisy evaluation of the function $f$, resulting in a final query complexity of $\Ocal(d^4/\varepsilon^6)$.

Because their algorithm tolerates high evaluation noise, it is natural to wonder if refinements of the complexity are possible assuming lower noise levels. We essentially rule out this possibility in Appendix~\ref{app:prospects-reduction}. In summary, our algorithm appears to offer genuine speedups over known classical algorithms for stochastic convex optimization.

\section*{Acknowledgments}
The authors thank  Aram Harrow, Jiaqi Leng, and Xiaodi Wu for insightful discussions about QHD and the simulation of real-space dynamics. They also thank their colleagues at the Global Technology Applied Research center of JPMorganChase for their support and helpful discussions. 

\section*{Disclaimer}
This paper was prepared for informational purposes by the Global Technology Applied Research center of JPMorgan Chase \& Co. This paper is not a product of the Research Department of JPMorgan Chase \& Co. or its affiliates. Neither JPMorgan Chase \& Co. nor any of its affiliates makes any explicit or implied representation or warranty and none of them accept any liability in connection with this paper, including, without limitation, with respect to the completeness, accuracy, or reliability of the information contained herein and the potential legal, compliance, tax, or accounting effects thereof. This document is not intended as investment research or investment advice, or as a recommendation, offer, or solicitation for the purchase or sale of any security, financial instrument, financial product or service, or to be used in any way for evaluating the merits of participating in any transaction.

\bibliographystyle{alpha}
\bibliography{sdpbib}

\appendix

\section{Block Encodings} \label{app:block_encodings}

It is often useful to embed non-unitary operations $A$, which may represent problem input data, in some larger unitary operation $U_A$ that can be naturally represented on a quantum computer. This simple idea leads to the surprisingly powerful technique of block encoding~\cite{low2019qubitization, chakraborty2018power}. Here we provide a simple definition suitable for our purposes.
\begin{definition}[Block encoding]
    A linear operation $A$ on a quantum register $r_1$ is said to be \emph{block encoded} by a unitary $U_A$ on a super-register $r_1 + r_2$ if there exists a quantum state $\ket{G}$ on $r_2$ such that
    \begin{equation}
        (I_{r_1} \otimes \bra{G}_{r_2}) U_A (I_{r_1} \otimes \ket{G}_{r_2}) = A
    \end{equation}
    up to rescaling by a positive constant.
\end{definition}
\noindent In other words $U_A$ block encodes $A$ if $A$ can be viewed as a sub-block of $U_A$ in matrix form. More general notions of \emph{approximate} block encoding are also possible to define. 

A particularly valuable instance is an LCU-block encoding. An operator
\begin{equation}
    A = \sum_{j=0}^{n-1} a_j U_j
\end{equation}
that is a positive linear combination of unitaries (LCU) $U_j$ can be block encoded using the so-called $\SEL$ ("select") and $\PREP$ ("prepare") unitaries
\begin{equation}
    \SEL \coloneqq \sum_{j=0}^{n-1} \ket{j}\bra{j} \otimes U_j, \qquad \PREP\ket{0} \propto \sum_{j=0}^{n-1} \sqrt{a_j} \ket{j}.
\end{equation}
Setting $\ket{G} = \PREP\ket{0}$, we find that
\begin{equation}
    (\bra{G}\otimes I) \SEL(\ket{G}\otimes I) \propto A.
\end{equation}
The proportionality constant of this block encoding is, in fact, essential to understanding the probability that a block encoded $A$ is successfully implemented. The ``amplitude" of success is given by $\norm{A\ket{\psi}}/\norm{a}_1$, where $\ket{\psi}$ is the input state. The success probability upon direct measurement is the square of this, but there are situations in which amplitude amplification allows for quadratic improvement.

Whenever $U_j$ themselves are block encodings of operations $A_j$ by some $j$-independent state $\ket{G}$, the LCU block encoding allows for a nesting structure to encode operations such as $\sum_j a_j A_j$.

\subsubsection*{Oracles}
There are several common quantum oracles in the literature that we adopt in this work. Here we define them here and reproduce some simple relations between them.

Given a value $y$ in binary, possibly dependent on a binary value $x$ of some quantum register, it is natural to define a \emph{binary} (quantum) oracle $O_y$, which acts as
\begin{equation}
    O_y \ket{x}\ket{z} = \ket{x}\ket{z\oplus y_x}
\end{equation}
where $\oplus$ denotes binary addition. Note that $O_y$ is unitary and self-inverse. This quantum oracle is most naturally connected to classical oracles for computing $y$, since any reversible classical circuit for computing $y$ can be converted into a quantum circuit for the same task. Tacitly understood is that $O_y$ acts linearly on superpositions of computational basis states.

Another oracle type, which is perhaps more analogous to probabilistic classical oracles, is a (quantum) \emph{amplitude} oracle. Given a vector of values $y = (y_1,\ldots,y_n)$, define $D_y$ as a operation that acts in the computational basis of a register of dimension at least $n$ as
\begin{equation}
    D_y \ket{j} = y_j \ket{j} 
\end{equation}
for each $j \in [n]$. The oracle $A_y$ can be prepared using a small number of calls to $O_y$ and a sequence of known elementary gates.
\begin{lemma} \label{lem:bitwise_to_amp_block_encoding}
    The oracle $D_y$ can be block encoded in a unitary consisting of two calls to each binary oracle
    \begin{equation*}
        O_y \ket{j}\ket{z} = \ket{j}\ket{z\oplus y_j}
    \end{equation*}
    using $w + 1$ ancillary qubits and $\Ocal(w^2)$ elementary gates, where $w$ is the number of bits used to represent the $y$ values.
\end{lemma}
\begin{proof}
    Let $\Lambda$ be a (strict) upper bound on all possible values of $\abs{y_i}$, and renormalize $D_y \mapsto D_y/\Lambda$ so that all values $y_i$ lie in $(-1,1)$. This conversion comes at no cost by choosing $\Lambda$ a power of 2. Implement the binary oracle storing the value $y$ in a $w$-qubit register.
    Perform an $\arccos$ binary transformation on the $y$-register, which takes $O(w^2)$ gates~\cite{haner2018optimizing}.
    Append an additional single-qubit register and perform a controlled $Ry$-rotation of the single qubit controlled on the $y$-register. Finally uncompute the $\arccos$ and the $y$ register using another binary oracle call.
    Then one can check that the resulting unitary implements the amplitude oracle conditional to postselecting 0 in the ancilla register.
\end{proof}
\noindent As a simple corollary of this fact, the amplitude oracle
\begin{equation}
    A_y \ket{0} = \sum_j y_j \ket{j}
\end{equation}
for some initial state $\ket{0}$ can be implemented using the same cost as above plus the resources required to prepare the uniform superposition $\sum_{j\in J} \ket{j}$ (up to normalization) over some subset of computational basis states indexed by $J$.

The below conversion is standard and follows from using phase estimation.
\begin{lemma}[Phase to Binary Oracle Conversion]
\label{lem:phase_to_binary}
Suppose $O_f$ is an $\epsilon_f$-accurate phase oracle for $f$ rescaled by $\Lambda$ (Definition \ref{defn:phase_oracle}). Then with $\Ocal\left(\log(1/\delta)/\epsilon_f\right)$ calls to $ O_{f}$ and its inverse, we produce a unitary that is a $(\Lambda\epsilon_f)$-accurate binary oracle oracle for $f$ with probability at least $1 - \delta$.
\end{lemma}

\section{Prospects for Reducing Classical Complexity}
\label{app:prospects-reduction}

    Here we discuss whether leading classical algorithms for noisy zeroth-order convex optimization can be improved if the noise tolerance is allowed to be lowered to that of the present quantum algorithm. It is not apparent to us that such an improvement is possible. 
    
    First, we note that the query complexity of the algorithm of Belloni et al~\cite{belloni2015escaping} is based on the analysis of a classical hit-and-run walk. The bottleneck is the classical complexity of the algorithm and is unchanged even if the queries are chosen to be noiseless. The robustness of the algorithm arises from the natural robustness of walk based algorithms, and the lack of any gradient estimation. Thus reducing the noise tolerance does not lead to any improvements in oracle complexity.
    
    Next, consider the algorithm of Ristetski and Li~\cite[Algorithm 1]{risteski2016algorithms}, which is based on a gradient estimator. It is evident by inspection of the algorithm that the query complexity is determined by a parameter $r$ that is the radius of smoothing. As usual, we rescale the problem by $GR$, letting $\varepsilon = \epsilon/GR$ and $\varepsilon_f = \epsilon_f/GR$. A calculation borrowed from~\cite{flaxman2005online} shows that the choice for $r$ must satisfy
    \begin{align*}
        f(x) - f(x_\star) - 2r \geq \frac{\varepsilon_f \sqrt{d}}{r} \text{ whenever } f(x) - f(x_\star) \geq \varepsilon.
    \end{align*}
    Plugging in the 2nd inequality to the first, it suffices that $r$ lies in the range 
    \begin{equation}
    r = \frac{\varepsilon}{4} (1 \pm \sqrt{1 - 8\varepsilon_f \sqrt{d}/\varepsilon^2}.
    \end{equation}
    Assuming a noise $\varepsilon_f =\Ocal(\varepsilon^3/d)$, we choose the largest $r$ consistent with these bounds, which minimizes sample complexity (see proof of Thm 3.2~\cite{risteski2016algorithms}). This is $r = \Theta(\varepsilon)$, the same choice of $r$ as the original work. Thus, following through the rest of the analysis (which is independent of the noise parameter) leads to no improvements.

\section{Technical Results and Proofs} \label{app:technical-lemmas}

The following is well-known but just proven for completeness and presented in a form  that is suitable for our analysis.
\begin{lemma}[Evolution Bound]
\label{lem:evolution_bound}
Let $f, g \in \Hcal$, and $\phi_f$ and $\phi_g$:
\begin{align*}
&i\partial_t\phi_f = -a(t)\Delta\phi_f + b(t)f\phi_f\\
&i\partial_t\phi_g = -a(t)\Delta\phi_g + b(t)g\phi_g,
\end{align*}
with the same initial conditions $\psi \in \Hcal^2$. Then,
\begin{align*}
    \lVert \phi_f(T, \cdot) - \phi_{g}(T, \cdot)\rVert \leq \lVert b\rVert_1 \lVert f- g\rVert_{\infty}
\end{align*}
\end{lemma}
\begin{proof}
    Switching to the interaction picture with $U_{\Delta}(t) = \exp(-i\int a(t)\Delta)$ gives
    \begin{align*}
    &i\partial_t\mathcal{U}_f = b(t)U_{\Delta}fU_{\Delta}^{\dagger}\mathcal{U}_f = \Hcal_f(t)\mathcal{U}_f\\
    &i\partial_t\mathcal{U}_g = b(t)U_{\Delta}gU_{\Delta}^{\dagger}\mathcal{U}_g = \Hcal_g(t)\mathcal{U}_g.
    \end{align*}
    Now let $\psi$ be the initial wave function. We then have
    \begin{align*}
    \lVert \phi_f - \phi_g \rVert &=\norm{ [\mathcal{U}_f - \mathcal{U}_g]\psi }\\ &=    \norm{ [I- \mathcal{U}_g^{\dagger}\mathcal{U}_f]\psi }  \\&= \norm{ \int_{0}^{T}\frac{\td}{\td\!t}\left(\mathcal{U}_g^{\dagger}\mathcal{U}_f\right) \td\!t \psi }\\
    &=\norm{ \int_{0}^{T}\left[\frac{\td}{\td\!t}\mathcal{U}_g \right]^{\dagger}\mathcal{U}_f + \mathcal{U}_g^{\dagger}\frac{\td}{\td\!t}\mathcal{U}_f \td\!t \psi }\\
    &=\norm{ \int_{0}^{T}i\mathcal{U}_g^{\dagger}[\Hcal_f - \Hcal_g]\mathcal{U}_f \td\!t \psi }\\
    &\leq \int_{0}^{T} b(t) \td\!t \norm{ [f - g]\mathcal{U}_f\psi }\\
    &\leq \int_0^T b(t)dt \lVert f- g\rVert_{\infty}
    \end{align*}
\end{proof}

\begin{lemma}[Evolution Bound with Low Leakage]
\label{lem:evolution_bound_leakage}
    Let $f, g \in \Hcal$, and $\phi_f$ and $\phi_g$ satisfy:
    \begin{align*}
    &i\partial_t\phi_f = -a(t)\Delta\phi_f + b(t)f\phi_f\\
    &i\partial_t\phi_g = -a(t)\Delta\phi_g + b(t)g\phi_g,
    \end{align*}
    with the same initial conditions $\psi \in \Hcal^2$. 
    Furthermore, suppose that $\lVert g \rVert_{\infty} \leq \lVert f \rVert_{\infty}$ and that the potential $f$ results in a Low-leakage Evolution (Definition \ref{defn:low-leakage-ev})
    Then,
    \begin{align*}
        \lVert \phi_f(T, \cdot) - \phi_{g}(T, \cdot)\rVert \leq \lVert b\rVert_1 \lVert f- g\rVert_{L_{\infty}(B_{\infty}(\frac12-\delta))} + \delta.
    \end{align*}
\end{lemma}
\begin{proof}
    The proof follows \ref{lem:evolution_bound} up to
    \begin{align*}
    \lVert \phi_f - \phi_g \rVert &\leq \lVert b\rVert_1 \norm{ [f - g]\mathcal{U}_f\psi }\\
    &\leq \lVert b\rVert_1 \left[\int_{B_{\infty}(\frac12-\delta)}(f(x)-g(x))^2\lvert \phi_f(x)\rvert^2 dx\right]^{1/2}\\
    &+\lVert b\rVert_1 \lVert f - g \rVert_{\infty}\left[\int_{[B_{\infty}(\frac12-\delta)]^c}\lvert \phi_f(x)\rvert^2\right]^{1/2}\\
    &\leq \lVert b\rVert_1 \lVert f- g\rVert_{L_{\infty}(B_{\infty}(\frac12-\delta))} + \delta
    \end{align*}
\end{proof}

    The following is useful for bounding aliasing error.
    \begin{lemma} \label{lem:alias_inner_prod}
        Let $\chi_\nmbf(x) = e^{2\pi i \nmbf \cdot x}$. For any $\nmbf, \mmbf \in \Zmbb^d$, we have 
        \begin{equation}
            \langle \chi_\nmbf, \chi_\mmbf \rangle_\Gcal = \mathbbm{1}(\nmbf - \mmbf \in 2N\Zmbb^d).
        \end{equation}
        That is, the bilinear form evaluates to $1$ if $\nmbf$ and $\mmbf$ are in the same coset of $2N\Zmbb^d$, else is zero.  As a consequence, Suppose that $\phi \in \Hcal$ and that the periodic extension of $\phi$ is in $C^{1}(\Rmbb^d)$, then
        \begin{equation*}
            \langle\chi_\nmbf, \phi\rangle_\Gcal = \sum_{\kmbf \in \Zmbb^d} \langle \chi_{\nmbf + 2N \kmbf}, \phi\rangle.
        \end{equation*}
    \end{lemma}
    \noindent Thus, intuitively, the bilinear form $\langle\cdot,\cdot\rangle_\Gcal$ cannot distinguish low frequency modes from high frequency ones that differ by units of $2N$ in each component.
    \begin{proof}
        From the definition of $\langle\cdot,\cdot\rangle_\Gcal$,
        \begin{align*}
            \langle\chi_\nmbf, \chi_\mmbf\rangle_\Gcal &= \frac{1}{(2N)^d} \sum_{\jmbf \in \Ncal} e^{2\pi i \, \jmbf \cdot(\mmbf-\nmbf)/(2N)} \\
            &= \prod_{\ell=1}^d \frac{1}{(2N)} e^{2\pi i j_\ell (m_\ell - n_\ell)/(2N)}.
        \end{align*}
        By basic properties of discrete Fourier transforms, each term in the product is zero unless $m_k = n_k + 2N k_\ell$ for some $k_\ell \in \Zmbb$. This amounts to the first claim of the lemma. 
        
        For the second claim, let $\phi_C$ be the Fourier-truncated function for wavenumber $\norm{\nmbf}_\infty \leq C$. We first wish to show that the map $\langle\chi_\nmbf, \cdot\rangle_\Gcal$ is a continuous complex-valued function in the sense that $\lim_{C\rightarrow\infty} \langle\chi_\nmbf, \phi_C\rangle_\Gcal = \langle\chi_\nmbf, \phi\rangle_\Gcal$. Indeed,
        \begin{align*}
            \abs{\langle\chi_\nmbf, \phi\rangle_\Gcal - \langle\chi_\nmbf, \phi_C\rangle_\Gcal} &= (1/2N)^d \abs{\sum_{\jmbf \in \Ncal} \overline{\chi}_\nmbf(x_\jmbf) (\phi(x_\jmbf) - \phi_C(x_\jmbf))} \\
            &\leq  \norm{\phi - \phi_C}_\infty.  
        \end{align*}
        Since we are assuming that $\phi$ is continuously-differentiable on the $n$-Torus, this implies that it is Lipschitz over $S$. Thus we have that for any $d$ and $x \in S$
        \begin{align*}
            \lim_{\delta \rightarrow 0^{+}} \lvert \phi(x + \delta) - \phi(x)\rvert \cdot \log^{d}(1/\delta) = 0,
        \end{align*}
        and so the multi-dimensional Dini--Lipschitz theorem \cite[Theorem 2.1.1]{zhizhiashvili2012trigonometric} implies that the rectangular Fourier partial sums $\phi_C$ converge to $\phi$ uniformly. Hence $ \norm{\phi - \phi_C}_\infty$ vanishes as $C\rightarrow\infty$.

        Finally, we evaluate $\langle\chi_\nmbf,\phi_C\rangle_\Gcal$ to be 
        \begin{align*}
            \langle\chi_\nmbf,\phi_C\rangle_\Gcal &= \sum_{\norm{\mmbf}_\infty \leq C} \langle\chi_\mmbf, \phi\rangle \langle\chi_\nmbf, \chi_\mmbf\rangle_\Gcal = \sum_{\substack{\norm{\mmbf}\leq C \\ \mmbf - \nmbf \in (2N\Zmbb)^d}} \langle\chi_\mmbf, \phi\rangle.
        \end{align*}
        Taking the $C\rightarrow\infty$ limit gives the result of the lemma.
    \end{proof}

    \subsection{Existence of Dynamics}
\label{subapp:existence_of_dynamics}

Here we prove Theorem \ref{thm:solution_existence}, which ensures the existence of a unique solution to the Schr\"odinger initial value problem and hence well-posed dynamics. The technical difficulty, as always in these contexts, is that the Hamiltonian in question is an unbounded operator and thus there are domain issues. We will ultimately rely on a modified version of a theorem from Reed and Simon's classic text~\cite{reed1975ii}. We begin with a simple helper lemma that can be found in a number of mathematical physics texts.
\begin{lemma} \label{lem:unbounded_plus_bounded}
    Let $A:\Dom(A)\rightarrow \Hcal$ be an unbounded self-adjoint operator on a separable Hilbert space $\Hcal$. Let $B$ be self-adjoint and bounded on $H$. Then $A + B$ is self-adjoint on $\Dom(A)$.
\end{lemma}
\noindent In particular, the Hamiltonian $H(t) = -a(t) \Delta + b(t) f(x)$ is self-adjoint for all $t\in [0,T]$ with domain $\Dom(-\Delta) = \Hcal^2$. Indeed $f$ is continuous on a compact domain and hence a bounded function. The key result is an adaption of Theorem X.71 from~\cite{reed1975ii} which is suitable for the high-dimensional setting.
\begin{theorem}[Existence of Dynamics]
Let $H(t) = -a(t) \Delta + b(t)f$ have domain $\Dom(-\Delta) = \Hcal^2$, with $f$ Lipschitz continuous and $a, b \in C^1([0, T])$. Let $\Phi_0 \in \Hcal^2$. Then, there is a strongly time-differentiable $\Phi(t)$ such that $i\partial_t\Phi(t) = H(t)\Phi(t)$ and $\Phi(0) = \Phi_0$. 
\end{theorem}
\begin{proof}
    By Lemma~\ref{lem:unbounded_plus_bounded}, $H(t)$ is self-adjoint for all $t \in [0,T]$. Moreover, since $-\Delta \geq 0$ (as an operator) and $b(t) f(x)$ is bounded, there exists a $D \geq 0$ such that $-\Delta + b(t)f + D \geq 1/2$ (as operators) for all $t$. Hence, $A(t) = -i(H(t) + D)$ generates a contraction semigroup for each $t$. Moreover, $A(t)$ satisfies the conditions of Theorem X.70 of~\cite{reed1975ii}. Thus, by this theorem, the propagator $U(t,0)$ generated by $A(t)$ is unitary, and for each initial $\Phi_0 \in \Hcal^2$, $\Phi(t) = U(t,0) \Phi_0$ is a solution to the Schr\"odinger equation.
\end{proof}

We would like to guarantee additional regularity on the true solution by fine-tuning the regularity of the initial state. It turns out that this is simple consequence of the original proof of Reed \& Simon, and has appeared in the time-independent case elsewhere \cite[Theorem 7.5]{brezis2011functional}. The key is to note that domain of our operator of interest:
\begin{align*}
    A(t) = -a(t)\Delta + V(x,t)
\end{align*}
is determined by $\Delta$ for bounded $V$, such that $A(t) \succ 0$ (which can be taken without loss of generality by applying a shift). Note that by the Meyers-Serrin Theorem, the functions $C^{\infty}(S)$ are dense (in $\Hcal^k$-norm) in $\Hcal^k$ for all $k$. Hence each $\Hcal^{k+2}$ is dense in $\Hcal^k$, where each is a Hilbert space. Hence by Proposition 9.23 \cite{hall2013quantum} and the potential being bounded in $L_{\infty}(S)$-sense, for each $t$,  $A(t) : \Hcal^{k+2} \rightarrow \Hcal^{k}$ is a self-adjoint operator on $\Hcal^{k}$. The Hille-Yosida theorem  gives that $iA(t)$ is also the generator of a contractive $C_0$-semigroup over $\Hcal^{k}$. Theorem 10.70 in \cite{reed1975ii} then implies that there is a unique strongly-differentiable solution to the dynamics in $\Hcal^{k+2}$ for all $t$. This leads to a proof of Theorem \ref{thm:solution_existence}.

Alternatively, if we have regularity guarantees on the potential, $V(x,t)$, specifically, $V(t) \in \Hcal^{s}$, then $\Phi(t) \in \Hcal^{s}$ when $\Phi_0 \in \Hcal^{s}$ immediately follows from Lemma \ref{eq:sob_seminorm_bound}.

\subsection{Proofs of State Interpolation Bounds} \label{subapp:state_interpolation_bound}
In this appendix we present proofs for the bounds on the interpolation error for a state $\phi$ in some Sobolev class, stated in Theorem~\ref{thm:total_interpolation_error_l2} of the main text. We also prove an analogous result for the interpolation error in the $\Hcal^2$ Sobolev seminorm. These results are critical to the proofs of our rigorously-provided guarantees on the performance of the multi-dimensional pseudo-spectral method. They may also be of independent interest.

We begin with a few lemmas that relate $\phi$ and $I_N\phi$ via the orthogonal projection $P_N$ onto $\Hcal_N$.
\begin{lemma}
    \label{lem:trunc_high_sobolev}
    For $m \geq 1$, if $\phi \in \Hcal^{m}$, then 
    \begin{align*}
        \lVert P_N\phi - \phi\rVert \leq  \frac{\lvert \phi \rvert_{\Hcal^m}}{N^{m}}
    \end{align*}
\end{lemma}
\begin{proof}
    Computing the norm in the Fourier basis, one finds
    \begin{align*}
        \lVert P_N\phi - \phi\rVert ^2 &= \sum_{\nmbf \notin \Ncal} \lvert \langle \chi_{\nmbf}, \phi\rangle \rvert^2.
    \end{align*}
    Inserting factors of $\norm{\nmbf}^{2m}$ and using the Sobolev bound,
    \begin{align*}
        \sum_{\nmbf \notin \Ncal} \lvert \langle \chi_{\nmbf}, \phi\rangle \rvert^2 &=\sum_{\nmbf \notin \Ncal} \norm{\nmbf}^{-2m} \norm{\nmbf}^{2m}\lvert \langle \chi_\nmbf, \phi\rangle \rvert^2\leq N^{-2m}\sum_{\nmbf \notin \Ncal} \norm{\nmbf}^{2m}\lvert \langle \chi_\nmbf, \phi\rangle \rvert^2\leq \frac{\lvert \phi \rvert_{\Hcal^m}^2}{N^{2m}}.
    \end{align*}
    Taking square roots of boths sides leads to the lemma statement.
\end{proof}

As mentioned in the main text, the truncation error is only one half of the story. When one approximates Fourier coefficients via the discrete Fourier transform, aliasing error is introduced due to the inability to distinguish certain high-frequency components. Recall that the Fourier interpolant $I_N[\phi]$ is the truncated Fourier polynomial obtained from $\phi$ via the discrete Fourier transform. The following bounds the error between the truncated Fourier series and the interpolant, which is the aliasing error. 

\begin{lemma}[Aliasing Error -- multi-dimensional Theorem 2.3 \cite{shen2011spectral}]
\label{lem:aliasing}
    Suppose $\phi \in \Hcal^m$ for positive integer $m > \max\{d/2,2\}$ and that the periodic extension of $\phi$ is in $C^{1}(\Rmbb^d)$. Then 
    $$
        \lVert I_N\phi - P_N\phi\rVert \lesssim \left(\frac{\pi}{4}\right)^{d/4} \frac{1}{\sqrt{(m-d/2) \Gamma(d/2)} N^m} \abs{\phi}_{\Hcal^m}^2.
    $$
\end{lemma}
\begin{proof}
    Since we assume $\phi$ has a $C^1(\Rmbb^d)$ periodic extension, Lemma~\ref{lem:alias_inner_prod} holds. From the definitions of $I_N$ and $P_N$, 
    \begin{align*}
         \lVert I_N\phi - P_N\phi\rVert^2 &= \sum_{\nmbf\in\Ncal} \lvert \langle \chi_\nmbf, \phi \rangle_\Gcal- \langle \chi_\nmbf, \phi \rangle \rvert^2 \\
         &= \sum_{\nmbf\in\Ncal} \Big\lvert\sum_{\kmbf \in \Zmbb^d\setminus\{\mathbf{0}\}} \langle\chi_{\nmbf + 2N\kmbf}, \phi\rangle\Big\rvert^2.
    \end{align*}
    We wish to bound this aliasing error using the semi-norm $\abs{\cdot}_{\Hcal^k}$. We begin with an application of Cauchy-Schwarz. For any $m \geq 1$,
    \begin{align*}
        \sum_{\nmbf\in\Ncal} \Big\lvert\sum_{\kmbf \in \Zmbb^d\setminus\{\mathbf{0}\}} \langle\chi_{\nmbf + 2N\kmbf}, \phi\rangle\Big\rvert^2 &= \sum_{\nmbf\in\Ncal} \Big\lvert\sum_{\kmbf \in \Zmbb^d\setminus\{\mathbf{0}\}} \norm{\nmbf + 2N\kmbf}^{-m} \norm{\nmbf + 2N\kmbf}^m \langle\chi_{\nmbf + 2N\kmbf}, \phi\rangle\Big\rvert^2 \\
        &\leq \sum_{\nmbf\in\Ncal} \left(\sum_{\kmbf\in\Zmbb^d\setminus\{\mathbf{0}\}} \frac{1}{\norm{\nmbf + 2N\kmbf}^{2m}}\right)\left(\sum_{\kmbf\in\Zmbb^d\setminus\{\mathbf{0}\}}\norm{\nmbf + 2N\kmbf}^{2m} \abs{\langle\chi_{\nmbf + 2N\kmbf}, \phi\rangle}^2 \right) \\
        &\leq \max_{\nmbf\in\Ncal}\left\{\sum_{\kmbf\in\Zmbb^d\setminus\{\mathbf{0}\}} \frac{1}{\norm{\nmbf + 2N\kmbf}^{2m}}\right\} \sum_{\nmbf\in\Ncal} \sum_{\kmbf\in\Zmbb^d\setminus\{\mathbf{0}\}}\norm{\nmbf + 2N\kmbf}^{2m} \abs{\langle\chi_{\nmbf + 2N\kmbf}, \phi\rangle}^2 \\
        &\leq\max_{\nmbf\in\Ncal}\left\{\sum_{\kmbf\in\Zmbb^d\setminus\{\mathbf{0}\}} \frac{1}{\norm{\nmbf + 2N\kmbf}^{2m}}\right\} \abs{\phi}_{\Hcal^m}^2.
    \end{align*}
    The maximum over $\nmbf$ is exhibited at $\nmbf = -N \mathbf{1}$, where $\mathbf{1}$ is the all ones vector. Hence,
    \begin{align*}
        \max_{\nmbf\in\Ncal}\left\{\sum_{\kmbf\in\Zmbb^d\setminus\{\mathbf{0}\}} \frac{1}{\norm{\nmbf + 2N\kmbf}^{2m}}\right\} &= \sum_{\kmbf\in\Zmbb^d\setminus\{\mathbf{0}\}} \frac{1}{\norm{2N\kmbf - N \mathbf{1}}^{2m}} \\
        &= \frac{1}{(2N)^{2m}} \sum_{\kmbf \in \Zmbb^d\setminus\{\mathbf{0}\}} \frac{1}{\norm{\kmbf - \frac12\mathbf{1}}^{2m}} \\
        &=  \Ocal\left(\frac{\Omega_{d}}{(2N)^{2m}} \sum_{ r \in \Zmbb\setminus\{0\}} \frac{r^{d-1}}{(r - \frac12)^{2m}}\right)\\
        &= \Ocal\left(\frac{\Omega_{d}}{(2N)^{2m}} \sum_{ r \in \Zmbb\setminus\{0\}} \frac{1}{(r - \frac12)^{2m-d+1}}\right)\\
        &= \Ocal \left(\frac{\Omega_{d}}{(2N)^{2m}}\left(2^{2m-d+1} + \int_{1/2}^\infty \frac{1}{k^{2m-d+1}}dk\right)\right).
    \end{align*}
    where $\Omega_d$ is the surface area of the ball in $d$ dimensions. In the last line, we  made use of the inequality 
    \begin{align*}
        \sum_{r=N}^{\infty} f(n) \leq f(N) + \int_{N}^{\infty}f(x)dx,
    \end{align*}
    valid for monotonically-decreasing $f$. The integral converges when $2m > d$, and is given by $2^{2m-d}/(2m-d)$. Using the known formula
    \begin{equation*}
        \Omega_d = \frac{2 \pi^{d/2}}{\Gamma(d/2)}
    \end{equation*}
    and simplifying, we have the result stated in the lemma. 
\end{proof}

We also derive an analogous result bounding the aliasing error in terms of $\Hcal^2$ seminorm.
\begin{lemma}[Aliasing Error in $\Hcal^2$ seminorm]
\label{lem:aliasing_in_sobolev_norm}
    Suppose $\phi \in \Hcal^m$ for positive integer $m\geq \max\{d/2,2\}$, and that the periodic extension of $\phi$ is in $C^{1}(\Rmbb^d)$. Then 
    \begin{align*}
        \lvert I_N\phi - P_N\phi\rvert_{\Hcal^{2}} \lesssim  \left(\frac{\pi}{4}\right)^{d/4} \frac{1}{\sqrt{(m-d/2) \Gamma(d/2)} N^m} \abs{\phi}_{\Hcal^{m+2}}.
    \end{align*}
\end{lemma}
\begin{proof}
We know that 
\begin{align*}
&-\Delta(I_N\phi(x)) = \sum_{\nmbf \in \Ncal}4\pi^2\langle \chi_{\nmbf}, \phi\rangle_{\Gcal}\lVert\nmbf\rVert^2\chi_{\nmbf}(x)\\
&-\Delta(P_N\phi(x)) = \sum_{\nmbf \in \Ncal}4\pi^2\langle \chi_{\nmbf}, \phi\rangle\;\lVert\nmbf\rVert^2\chi_{\nmbf}(x).
\end{align*}
Thus,
\begin{align*}
    \lvert I_N\phi - P_N\phi\rvert_{\Hcal^{2}} 
     = \lVert \Delta\left( I_N\phi - P_N\phi\right)\rVert^2 = \sum_{\nmbf\in\Ncal} 4\pi^2\lVert \nmbf \rVert^4\lvert \langle \chi_\nmbf, \phi \rangle_\Gcal- \langle \chi_\nmbf, \phi \rangle \rvert^2.
\end{align*}  
Hence, following the proof of Lemma~\ref{lem:aliasing} we get that $\lvert \phi \rvert_{\Hcal^m}$ is replaced with $\lvert \phi \rvert_{\Hcal^{m+2}}$.
\end{proof}

If we combine the truncation error and aliasing error, the triangle inequality (Equation~\eqref{eq:alias_truncation}) results in the following two bounds on the error between the interpolant and $\phi$.
\begin{theorem}
 Suppose $\phi \in \Hcal^m$ for positive integer $m > \max\{d/2,2\}$, and that the periodic extension of $\phi$ is in $C^{1}(\Rmbb^d)$. Then 
 \begin{align*}
        \lVert \phi - I_N\phi\rVert = \Ocal \left( \left(\frac{\pi}{4}\right)^{d/4} \frac{1}{\sqrt{(m-d/2) \Gamma(d/2)} N^{m}} \abs{\phi}_{\Hcal^m} + \frac{\lvert \phi \rvert_{\Hcal^1}}{N} \right).
    \end{align*}
\end{theorem}
\noindent Combining the aliasing error of Lemma \ref{lem:aliasing_in_sobolev_norm} with the truncation error of Lemma \ref{lem:trunc_high_sobolev} leads to an overall error bound.
\begin{theorem}
\label{thm:total_interpolation_error_sob}
 Suppose $\phi \in \Hcal^{m+2}$ for positive integer $m > \max\{d/2,2\}$, and  that the periodic extension of $\phi$ is in $C^{1}(\Rmbb^d)$. Then 
 \begin{align*}
        \lvert \phi - I_N\phi\rvert_{\Hcal^2} = \Ocal \left( \left(\frac{\pi}{4}\right)^{d/4} \frac{1}{\sqrt{(m-d/2) \Gamma(d/2)} N^{m}} \abs{\phi}_{\Hcal^{m+2}} + \frac{\lvert \phi \rvert_{\Hcal^3}}{N}\right).
    \end{align*}
\end{theorem}

\subsection{Proofs for Section \ref{subsec:error_analysis_hamiltonian}}

We start by proving Theorem \ref{thm:collocation_bound_real_space} from the main text, restated below.
\begin{theorem}[Multi-dimensional Generalization of Theorem 1.8 \cite{lubich2008quantum}]
Consider $N \in \Nmbb$. Let $\Phi(t)$ denote the solution to 
\begin{equation}
    i \partial_t \Phi(x, t) = [- a(t) \Delta + V(x,t)]\Phi(x, t),
\end{equation}
over $S = [-\frac12, \frac12]^d$,
with initial condition $\Phi(x, 0) = \Phi_0 \in \Hcal_{N}$. If $\forall t \in [0, T]$,  $\Phi(t) \in \Hcal^{m+2}$ for $m > \max(d/2, 2)$, and $\Phi(t)$ has a periodic extension, in $x$, that is in $C^{1}(\Rmbb^d)$, then the $L_2$-distance between the true solution $\Phi(t)$ and the pseudo-spectral approximation $\Psi(t)$ is bounded as
\begin{align*}
    \lVert \Phi(t) - \Psi(t)\rVert \lesssim\norm{a}_{1}\cdot \max_{t \in [0 ,T]}\left[\left(\frac{\pi}{4}\right)^{d/4}\frac{2\lvert \Phi(t) \rvert_{\Hcal^{m+2}} + \lvert \Phi(t)\rvert_{\Hcal^m}}{\sqrt{(m-d/2) \Gamma(d/2)} N^{m}} + \frac{2\lvert \Phi(t) \rvert_{\Hcal^3} + \lvert \Phi(t) \rvert_{\Hcal^1}}{N}\right].
\end{align*}
\end{theorem}
\begin{proof}
First, we observe that 
\begin{align*}
    I_N[V\Phi] = I_N[V I_N[\Phi]]
\end{align*}
because $V\Phi$ and $V I_N[\Phi]$ match on $\Gcal$ by the interpolation property of $I_N$. We can apply the interpolation operator to both sides of the \eqref{eq:real_schrodinger} and use linearity to get
\begin{align*}
&I_{N}\left[i\partial_t\Phi(x, t) = -a(t)\Delta \Phi(x, t) + V(x, t)\Phi(x, t)\right]\\
&\implies i\partial_tI_N[\Phi](x, t) = -a(t)\Delta I_N[\Phi](x,t) + I_N[VI_N[\Phi]] + a(t)[-I_N[\Delta \Phi]+\Delta I_N[\Phi]]\\
&\implies i\partial_tI_N[\Phi] = -a(t)\Delta I_N[\Phi] + I_N[VI_N[\Phi]] + \delta.
\end{align*}
Hence the Fourier interpolation of the true solution $\Phi$ satisfies an equation similar to $\Psi$, up to an error $\delta(x,t)$. Specifically from Lemma \ref{lem:interpolating_collocation} and $I_N[\Psi] = \Psi \in \Hcal_N$,
\begin{align*}
    i\partial_tI_{N}[\Psi] = -a(t)\Delta I_{N}[\Psi] + I_N[VI_N[\Psi]],
\end{align*}
which is a Sch\"odinger equation with Hamiltonian
\begin{align*}
    -a(t)\Delta + I_N[V \cdot].
\end{align*}

Next we bound the $L_2$ distance between $I_N[\Phi]$ and $\Psi$ using the interpolated equation for $\Psi$ from Lemma \ref{lem:interpolating_collocation}. By linearity, the error $e(x, t) = I_N[\Phi] - \Psi$ satisfies
\begin{align*}
i\partial_t e(x,t) = -a(t)\Delta(e)(x,t) + I_N[V e](x, t) + \delta(x, t).
\end{align*}
Note that since $\partial_te(x,t)$ is $L_2$ integrable in $x$, fundamental theorem of calculus and dominated convergence give
\begin{align*}
\partial_t\frac12\lVert e(\cdot,t)\rVert^2 = \text{Re}\langle e, \partial_t e\rangle = \text{Re}\langle e, -i\delta\rangle \leq \lVert e(\cdot, t)\rVert\cdot  \lVert \delta(\cdot, t)\rVert,
\end{align*}
where $\Re$ takes care of the skew-Hermitian part, and the norms are taken over the variable indicated by a "$\cdot$''.
Thus,
\begin{align*}
\lVert e(\cdot, T)\rVert &\leq \int_{0}^T \lVert \delta(\cdot, t)\rVert dt \\ &\leq \int_0^T a(t) \lVert \Delta I_N[\Phi](\cdot, t) - I_N[\Delta \Phi](\cdot, t)\rVert dt \\
&\leq \lVert a(t) \rVert_1 \cdot \max_{t \in [0, T]}\lVert \Delta I_N[\Phi](\cdot, t) - I_N[\Delta \Phi](\cdot, t)\rVert.
\end{align*}

Continuing to decompose the error,
\begin{align*}
\lVert \Delta I_N[\Phi](\cdot, t)- I_N[\Delta \Phi](\cdot, t)\rVert &\leq \lVert \Delta\left(I_N[\Phi](\cdot, t) - \Phi(\cdot, t)\right)\rVert + \lVert I_N[\Delta\Phi](\cdot, t) - \Delta\Phi(\cdot, t)\rVert \\
&=\lvert I_N[\Phi](\cdot, t) - \Phi(\cdot, t) \rvert_{\Hcal^2}  + \lVert I_N[\Delta\Phi](\cdot, t) - \Delta\Phi(\cdot, t)\rVert.
\end{align*}
We can bound the right hand side using Theorems \ref{thm:total_interpolation_error_l2} and \ref{thm:total_interpolation_error_sob}. From Theorem~\ref{thm:total_interpolation_error_l2},
\begin{align*}
    \lVert I_N[\Delta\Phi](\cdot, t) - \Delta\Phi(\cdot, t)\rVert &\lesssim\left(\frac{\pi}{4}\right)^{d/4} \frac{1}{\sqrt{(m-d/2) \Gamma(d/2)} N^{m}} \abs{\Delta \Phi(\cdot, t)}_{\Hcal^m} + \frac{\lvert\Delta \Phi(\cdot, t) \rvert_{\Hcal^1}}{N}\\
    &=\left(\frac{\pi}{4}\right)^{d/4} \frac{1}{\sqrt{(m-d/2) \Gamma(d/2)} N^{m}} \abs{\Phi(\cdot, t)}_{\Hcal^{m+2}} + \frac{\lvert\Phi(\cdot, t) \rvert_{\Hcal^3}}{N}
\end{align*}
while from Theorem~\ref{thm:total_interpolation_error_sob},
\begin{align*}
    \lvert I_N\Phi(\cdot, t) - \Phi(\cdot, t) \rvert_{\Hcal^2} \lesssim\left(\frac{\pi}{4}\right)^{d/4}\frac{1}{\sqrt{(m-d/2) \Gamma(d/2)} N^{m}} \abs{\Phi(\cdot, t)}_{\Hcal^{m+2}} + \frac{\lvert \Phi(\cdot, t) \rvert_{\Hcal^3}}{N}.
\end{align*}
Hence $\lVert I_N[\Phi](\cdot, t) - \Psi(\cdot, t)\rVert$ is bounded by the sum of the above bounds times $\lVert a \rVert_1$. We also know that by the Theorem \ref{thm:total_interpolation_error_l2}:
\begin{align*}
 \lVert I_N\Phi(\cdot, t) - \Phi(\cdot, t)\rVert &\lesssim\left(\frac{\pi}{4}\right)^{d/4} \frac{1}{\sqrt{(m-d/2) \Gamma(d/2)} N^{m}} \abs{\Phi(\cdot, t)}_{\Hcal^m} + \frac{\lvert\Phi(\cdot, t) \rvert_{\Hcal^1}}{N}.
\end{align*}
From here, the theorem statement follows by the triangle inequality.
\end{proof}

Next we prove the Sobolev growth bound Lemma \ref{lem:sobolev_growth_wavefunc}. For the proof, it will be helpful to consider the Sobolev seminorms via the fractional Laplacian, which is a natural generalization of Equation \ref{eq:laplacian_defn} to fractional powers:
\begin{align*}
    [\left(-\Delta\right)^{s}g](x) \coloneqq  \sum_{\nmbf \in \Zmbb^d} (4\pi)^s\lVert \nmbf\rVert^{2s}\langle g, \chi_{\nmbf}\rangle\;\chi_{\nmbf}(x), \qquad s\in (0, 1).
\end{align*}
We then define $D := \sqrt{-\Delta}$, and observe that
\begin{align*}
    \langle \psi, \phi \rangle_{\Hcal^k} &= \langle \psi, \phi\rangle\; + \langle D^k\psi, D^k\phi \rangle\;\\
    \lvert \psi \rvert_{\Hcal^{k}} &= \langle D^k \psi, D^k \psi\rangle\;.
\end{align*}
We will also use the following inequality.
\begin{lemma}[Tight $L_2$ Gagliardo-Nirenberg \cite{morosi2018constants}]
\label{lem:gns_inequality}
For all $\psi \in \Hcal^m$ and $\theta\in[0,1]$, we have
\begin{align*}
    \lVert D^{\theta m}\psi\rVert \leq \lVert \psi \rVert^{1-\theta}\lVert D^m \psi\rVert^{\theta}
\end{align*}
where $D = \sqrt{-\Delta}$.
\end{lemma}
We thus show the following under the assumption on the $t$-continuity of mixed-$x$ partial derivatives of $\Phi$, which is reasonable given the high regularity assumptions on $a(t)$ and $b(t)$.
\begin{lemma}[Sobolev Growth Bound -- Multidimensional Generalization of \cite{bourgain1999growth}]
    Suppose $S$ is $d$ dimensional and $\Phi(t)$ is the  classical solution to Problem \ref{prob:restricted_schr}.  Furthermore, suppose that $g(x) \in \Hcal^{(2m + 2)+\lceil \frac{d}{2} \rceil + 1}$,  $\Phi_0 \in \Hcal^{(2m + 2)+\lceil \frac{d}{2} \rceil + 1}$, and that $\forall \alpha \in \Nmbb^{d}$ with $\lVert \alpha \rVert_1 \leq 2m+2$, we have 
 $t \rightarrow \partial^{\alpha}_x\Phi(x, t)$ is continuous $\forall x \in S$. Then
    \begin{align*}
         \lvert \Phi(t)\rvert_{\Hcal^s} \leq 2{\lvert \Phi_0\rvert_{\Hcal^{s}}} + \left(\frac{2\sum_{k=1}^{s}\binom{s}{k}\lvert g\rvert_{\Hcal^k}\lVert b\rVert_1}{s}\right)^s
    \end{align*}
    for all $s \leq m$.
\end{lemma}
\begin{proof}
    Theorem \ref{thm:solution_existence} implies if $\Phi_0 \in \Hcal^{r}$, then $\forall t, \Phi(t) \in \Hcal^{r}$ for all $t \in [0, T]$. If $ r = (2m + 2)+\lceil \frac{d}{2} \rceil + 1$, then the Sobolev embedding theorem implies that $\partial^{\alpha}g(x)$, $\partial^{\alpha}_x\Phi(x ,t)$ exist in the classical sense and are Lipschitz continuous in $x$ for $\lVert \alpha\Vert_1 \leq 2m + 2$. Note also by Theorem \ref{thm:solution_existence}, $\partial_t\Phi(t)$ exists in the classical sense. Thus for $\lVert \beta \rVert_1 \leq 2m$,  $\partial^{\beta}_x\partial_t\Phi(x,t)$ exists and for each $x$ is continuous in $t$, since
    \begin{align}
    \label{eqn:mixed_derivs}
    \partial^{\beta}_x\partial_t\Phi(x,t) = \partial^{\beta}_x[ia(t)\Delta \Phi(x,t) - ib(t)g(x)\Phi(x,t)].
    \end{align}
    The Lipschitz continuity in $x$ implies that $\partial^{\beta}_x\partial_t\Phi(x,t)$ is jointly continuous in $x$ and $t$. Then Clairaut's theorem implies $\partial_x^{\beta}\partial_t\Phi(x, t) = \partial_t \partial_x^{\beta}\Phi(x, t)$, for $\alpha$ such that $\lVert \beta \rVert_1 \leq 2m$.
    
    We seek to bound the rate of change of $\lvert \Phi(\cdot, t)\rvert_{\Hcal^s}$, using the above equality: 
    \begin{align*}
    \partial_t\left(\frac12\lvert \Phi(\cdot, t)\rvert_{\Hcal^s}^2\right) &= 
    \frac12\partial_t\lVert D^s \Phi(\cdot, t)\rVert ^2 \\
    &= \frac12\partial_t \langle \Delta^{s}\Phi(\cdot, t) , \Phi(\cdot, t)\rangle\\
    &= \lvert \Re\langle \partial_tD^s\Phi(\cdot, t),  D^s \Phi(\cdot, t)\rangle\rvert \\
    &=\lvert \Re\langle D^s \partial_t\Phi(\cdot, t),  D^s \Phi(\cdot, t)\rangle\rvert\\
    &=\lvert \Re\langle D^s [ia(t)\Delta -ib(t)g(x)]\Phi(\cdot, t),  D^s \Phi(\cdot, t)\rangle\rvert\\
    &=\lvert \Re\langle ia(t)D^s\Delta \Phi(\cdot, t),  D^s \Phi(\cdot, t)\rangle + \Re\langle -ib(t)D^s[g(x)\Phi(\cdot, t)],  D^s \Phi(\cdot, t)\rangle\rvert.
    \end{align*}
    The third equality follows from the continuity of $\partial_t\Delta^s\Phi$, dominated convergence and symmetry of the derivatives.
    
    Note that
    \begin{align*}
        \langle D^s\Delta \Phi(\cdot, t),  D^s \Phi(\cdot, t)\rangle = \sum_{k=1}^{d} \langle \partial_{k}^2D^s\Phi(\cdot, t),  D^s \Phi(\cdot, t)\rangle,
    \end{align*}
    since $\partial_j^2$ is symmetric, the above term is real. Hence the first term above vanishes. In the below calculation, for simplicity we will ignore the power of $4\pi^2$ factors that appear when differentiating the Fourier series. We also have that
    \begin{align*}
        \lvert\text{Re}\langle -iD^s[g\Phi(\cdot, t)], D^s \Phi(\cdot, t) \rangle\rvert &= \lvert \text{Im}\sum_{\lVert \alpha \rVert_1 = s} \langle \partial^{\alpha}(g\Phi(\cdot, t)), \partial^{\alpha}\Phi(\cdot, t)\rangle\rvert\\
        &=\lvert \text{Im}\sum_{\lVert \alpha \rVert_1 =s}\sum_{ A \subseteq \Pcal(\alpha)\setminus\emptyset} \langle (\partial^Ag)\partial^{A^c}\Phi(\cdot, t), \partial^{\alpha}\Phi(\cdot, t)\rangle \rvert\\
        &=\lvert \text{Im}\sum_{\lVert \alpha \rVert_1 =s}\sum_{ A \subseteq \Pcal(\alpha)\setminus\emptyset} \sum_{\nmbf \in \Zmbb^d}\widehat{\Phi}_{\nmbf}\prod_{j \in \alpha}n_j\sum_{\rmbf\in\Zmbb^d}\widehat{g}_{\rmbf}\widehat{\Phi}(\cdot, t)_{\nmbf - \rmbf}\prod_{j \in A}r_j\prod_{j \in A^c}(n-r)_j\rvert\\
        &\leq \lVert D^s\Phi\rVert\left[\sum_{\nmbf \in \Zmbb^d}\sum_{\rmbf\in\Zmbb^d}\sum_{\norm{\alpha}_1 =s}\sum_{ S \subseteq \Pcal(\alpha)\setminus\emptyset} \prod_{j \in A}r_j^2\prod_{j \in A^c}(n-r)_j^2\lvert\widehat{g}_{\rmbf}\rvert^2\lvert\widehat{\Phi}_{\nmbf - \rmbf}\rvert^2\right]^{1/2}\\
        &\leq \lVert D^s \Phi \rVert\cdot \sum_{k=1}^{s}\binom{s}{k}\left(\sum_{\nmbf \in \Zmbb^d}\sum_{\rmbf \in \Zmbb^d}\norm{\nmbf-\rmbf}^{2(s-k)}\lVert \rmbf \rVert^{2k}\lvert \widehat{\Phi}_{\nmbf - \rmbf}\rvert^2 \lvert\widehat{g}_{\rmbf}\rvert^2\right)^{1/2}\\
        &\leq\sum_{k=1}^s\binom{s}{k} \lVert D^s\Phi \rVert\lVert D^{s-k}\Phi \rVert\lVert D^k g \rVert,
    \end{align*}
where $\Pcal$ denotes the power set of the vector $\alpha$ as a subset of $\Nmbb$, $A^c$ denotes the complement of $A$. The reason for excluding the empty set $\emptyset$ is because, as mentioned earlier, the $\Im$ causes this term to vanish. This is a subtle but important point, because otherwise we would get exponential growth below. Note that the various re-orderings of infinite sums are justified by the absolute convergence of the series involved. The last inequality follows from the convolution theorem and Cauchy-Schwarz.

We can also use Lemma \ref{lem:gns_inequality} to say that for a normalized wave function and $k\leq s$
\begin{align}
\label{eq:sobolev_inequality}
    \lVert D^{s-k} \Phi(t)\rVert \leq \lVert D^s\Phi(t)\rVert^{1-k/s}.
\end{align}
Together these give
\begin{align*}
\partial_t\lVert D^s \Phi(t)\rVert^2 \leq 2b(t)\left( \sum_{k=1}^s\binom{s}{k} \lVert D^k f\rVert\right)\cdot[\lVert D^s\Phi\rVert^2]^{1-1/2s}
\end{align*}

Let $C \coloneqq \sum_{k=1}^{s}\binom{s}{k}\lVert D^k f\rVert$. If $y = \lVert D^s \Phi(t)\rVert^2$, then we have the inequality $y' \leq 2Cb(t)y^{1-\frac{1}{2s}}$. Thus,
\begin{align*}
&\int_{0}^T \frac{y'}{y^{1-\frac{1}{2s}}} dy \leq 2C\lVert b\rVert_1\\
&2sy^{\frac{1}{2s}} - 2sy(0)^{\frac{1}{2s}} \leq 2C\lVert b\rVert_1\\
&y \leq 2y(0) + 2\left(\frac{C\lVert b\rVert_1}{s}\right)^{2s},
\end{align*}
where the last line used Jensen's inequality.
\end{proof}

Consider the rectangular mollifier of Definition~\ref{defn:mollifer}.
\begin{align*}
    \mathcal{M}_{\sigma, d}(x) = \begin{cases}
        (\vartheta\sigma)^{-d}\prod_{j=1}^{d}e^{-\frac{1}{1-x_j^2/\sigma^2}} & x \in (-\sigma, \sigma)^d\\
        0 & \text{otherwise}
    \end{cases},
\end{align*}
which is in $C^{\infty}(\Rmbb^d)$ and has $L_1$ norm equal to one. The value of $\vartheta < 1$ is the normalization constant of the Mollifier with $\sigma =1$. There is an obvious periodic extension, where we assume  $\sigma \leq \frac12$ consider $\mathcal{M}$ restricted to $S$ and periodically stitch it together to form $\mathcal{M}_{\sigma, P}$.

We now show the following result for the mollifier (also valid for the periodic version):
\begin{lemma}
\label{lem:mollifier_lem}
    The rectangular mollifier satisfies
    \begin{align*}
              \lVert \mathcal{M}_{\sigma, d} \rVert_{\Hcal^m} = \Ocal\left(\left(Cd/\sigma\right)^{3d}\right)
    \end{align*}
for $m \leq d + \Ocal(1)$, where $C$ is an absolute constant.
\end{lemma}
\begin{proof}
We  will use the known formula for the surface area of the $d$-dimensional unit sphere, $\Omega_d$:
\begin{equation*}
    \Omega_d = \frac{2 \pi^{d/2}}{\Gamma(d/2)}.
\end{equation*}
 For conciseness, we will also use $\asymp$ to denote $\Ocal$.  

Consider a single mollifier in one dimension:
\begin{align*}
    \mathcal{M}_{\sigma,1} = \frac{e^{-\frac{1}{1-x^2/\sigma^2}}}{\vartheta \sigma},
\end{align*}
then is known that the asymptotic decay of the Fourier coefficients via saddle-point approximation \cite{johnson2015saddlepointintegrationcinftybump} (plus the scaling property of Fourier transforms) is

\begin{align*}
    \lvert \langle \mathcal{M}_{\sigma, 1}, \chi_{n}\rangle\rvert \lesssim \sigma^{-3/4}\frac{e^{-\sqrt{\sigma n}}}{\vartheta\lvert n\rvert^{3/4}}.
\end{align*}
Since the rectangular mollifier is separable, we get that the Fourier decay is
\begin{align*}
     \lvert \langle \mathcal{M}_{\sigma, d}, \chi_{n}\rangle\rvert \lesssim\vartheta^{-d}\sigma^{-3d/4}\frac{e^{-\sum_j\sqrt{\sigma n_j}}}{\prod_{j}\lvert n_j\rvert^{3/4}} \lesssim\left(\sigma^{3/4}\vartheta\right)^{-d}{e^{-\sqrt{\sigma \lVert \nmbf \rVert}}} 
\end{align*}

For simplicity we will again ignore the factor of $(4\pi)^m$ that appears in $\lvert \mathcal{M}_{\sigma, d} \rvert_{\Hcal^m}^2$ as it can be absorbed into the constant $C$ in the lemma statement. We want to bound
\begin{align*}
    \lvert \mathcal{M}_{\sigma, d} \rvert_{\Hcal^m}^2&= \sum_{\nmbf \in \Zmbb^d} \lvert \langle  \mathcal{M}_{\sigma, d} , \chi_{\nmbf}\rangle \rvert^2 \lVert \nmbf \rVert^{2m},
\end{align*}
which we split in the following way
\begin{align*}
    \lvert \mathcal{M}_{\sigma, d} \rvert_{\Hcal^m}^2 = \sum_{\lVert \nmbf \rVert_{\infty} \leq d } \lvert \langle  \mathcal{M}_{\sigma, d} , \chi_{\nmbf}\rangle \rvert^2 \lVert \nmbf \rVert^{2m}  + \sum_{\nmbf \notin B_{2}(0, d) } \lvert \langle  \mathcal{M}_{\sigma, d} , \chi_{\nmbf}\rangle \rvert^2 \lVert \nmbf \rVert^{2m}.
\end{align*}
Note that $d$ is an asymptotic parameter, so we can apply the saddle-point approximation throughout the second sum.

We have the following trivial bound on $\lvert \langle  \mathcal{M}_{\sigma, d} , \chi_{\nmbf}\rangle \rvert^2 \leq 1$, so for the first ``finite'' term we use the trivial bound
\begin{align*}
    \sum_{\lVert \nmbf \rVert_{\infty} \leq d } \lvert \langle  \mathcal{M}_{\sigma, d} , \chi_{\nmbf}\rangle \rvert^2 \lVert \nmbf \rVert^{2m} \leq d^{d+2m}.
\end{align*}

The second infinite series will be handled using the above mentioned saddle-point approximation for the Fourier coefficients. This is valid since $d$ is an asymptotic parameter.
\begin{align*}
    \sum_{\nmbf  \notin B_{2}(0, d) } \lvert \langle  \mathcal{M}_{\sigma, d} , \chi_{\nmbf}\rangle \rvert^2 \lVert \nmbf \rVert^{2m}
    &\lesssim\sum_{\nmbf \notin B_{2}(0, d)}\left(\sigma^{3/4}\vartheta\right)^{-2d}{e^{-\sqrt{\sigma \lVert \nmbf \rVert}}} \lVert \nmbf \rVert^{2m}\\
    &\lesssim\Omega_{d}\sum_{r \geq d}\left(\sigma^{3/4}\vartheta\right)^{-2d}{e^{-\sqrt{\sigma r}}}r^{2m+d-1}.
\end{align*}
The term in the series is maximized at 
\begin{align*}
    r^* = \left\lfloor \frac{4(d+2m-1)^2}{\sigma} \right\rfloor,
\end{align*}
which will be greater than $d$.
Then we can split the sum into a monotonically-increasing finite part and monotonically-decreasing infinite part. Note that $L_2$ convergence of rectangular Fourier series also implies convergence when summing over the $\ell_2$ ball \cite[Theorem 4.2]{weisz2012summability}. We then apply the inequality
\begin{align*}
    \sum_{r=N}^{\infty} f(n) \leq f(N) + \int_{N}^{\infty}f(x)dx,
\end{align*}
for monotonically-decreasing $f$:
\begin{align*}
\sum_{r \geq d}{e^{-\sqrt{\sigma r}}}r^{2m+d-1}&= \sum_{r=d}^{r^*} e^{-\sqrt{\sigma r}}r^{2m+d-1} + \sum_{r^* +1}^{\infty} e^{-\sqrt{\sigma r}}r^{2m+d-1}\\
&\leq 2r^*e^{-\sqrt{\sigma r^*}}r^{*2m+d-1} + \int_{r^*}^{\infty} e^{-\sqrt{\sigma y}}y^{2m+d-1}dy.
\end{align*}
Now
\begin{align*}
   \int_{r^*}^{\infty} e^{-\sqrt{\sigma y}}y^{2m+d-1}dy &\leq \frac{1}{\sigma^{2m+d-2}}\int_{\sqrt{\sigma r^*}}^{\infty} e^{-x}x^{4m+2d-1}dx \leq \frac{1}{\sigma^{2m+d}}\Gamma(4m+2d),
\end{align*}
and
\begin{align*}
    e^{-\sqrt{\sigma r^*}}r^{*2m+d-1} &\leq \sigma^{-(2m+d-1)}e^{-(4m+2d-2)}(2d+4m-2)^{4m+2d-2} = \left(\frac{(2d+4m-2)}{\sigma^{1/2}e}\right)^{4m+2d-2}.
\end{align*}
Therefore
\begin{align*}
   \sum_{r \geq d}{e^{-\sqrt{\sigma r}}}r^{2m+d-1} &\leq 2r^*\left(\frac{(2d+4m-2)}{\sigma^{1/2}e}\right)^{4m+2d-2} + \frac{\Gamma(4m+2d)}{\sigma^{2m+d}}\\
    &\leq\sigma^{-(2m+d)}\left[2r^*\left(\frac{(2d+4m-2)}{e}\right)^{4m+2d-2} + \Gamma(4m+2d)\right]\\
    &\lesssim\frac{\Gamma(4m+2d)}{\sigma^{2m+d}},
\end{align*}
where Stirling's approximation was used.

So
\begin{align*}
     \sum_{ \nmbf  \notin B_{2}(0, d) } \lvert \langle  \mathcal{M}_{\sigma, d} , \chi_{\nmbf}\rangle \rvert^2 \lVert \nmbf \rVert^{2m} \lesssim\frac{(2\pi)^{d/2}\Gamma(4m+2d)}{\Gamma(d/2)\vartheta^{2d}\sigma^{2m+5d/2}}.
\end{align*}
For $ m \leq d + \Ocal(1)$, we get
\begin{align*}
    \sum_{\nmbf \notin B_{2}(0, d) } \lvert \langle  \mathcal{M}_{\sigma, d} , \chi_{\nmbf}\rangle \rvert^2 \lVert \nmbf \rVert^{2m}\lesssim\frac{C'\Gamma(6d)}{\Gamma(d/2)\vartheta^{2d}\sigma^{5d}} \lesssim\left(\frac{C''d^{11/2}}{\vartheta^2\sigma^{5}}\right)^d.
\end{align*}
Hence
\begin{align*}
    \lVert \mathcal{M}_{\sigma, d} \rVert_{\Hcal^m} = \Ocal\left(d^{3d/2}+ \left(\frac{C''
    d^{11/4}}{\vartheta\sigma^{5/2}}\right)^d\right) = \Ocal\left(\left(\frac{C^{'''}d^{3}}{\sigma^{3}}\right)^d\right).
\end{align*}
where all $C$'s are absolute constants.

\end{proof}

\begin{lemma}[Periodic Mollification]
Consider the $R$-restriction $g$ of a $G$-Lipschitz function in $p$-norm $f$, $p\geq 1$. Also define the set
\begin{align*}
 D := \begin{cases}
      S & \text{if}~$g$ ~\text{satisfies periodic boundary conditions,} \\ 
     B_{\infty}(0, \frac12 - \sigma) & \text{otherwise}
\end{cases}
\end{align*}
for any $0 < \sigma < \frac12$. 

There exists a $C^{\infty}$ function $g_{\sigma}$ with the following properties.
\begin{itemize}
    \item $\lvert g(x) - g_{\sigma}(x) \rvert = \Ocal\left(d^{1/p}GR\sigma\right)$ for all $x \in D$.
    \item  If $f$ is convex, $g_{\sigma}$ is convex on $D$.
    \item  For $m \leq d + \Ocal(1)$,
    $$\lvert g_{\sigma}\rvert_{\Hcal^m} =\Ocal\left(\lVert g \rVert_{\infty}\left(\frac{Cd}{\sigma}\right)^{3d}\right)$$
    where $C$ is an absolute constant.
\end{itemize}
\end{lemma}
\begin{proof}
Let $g_{P}$ denote the periodic extensions of $g$. We take the periodic convolution between $\mathcal{M}_{\sigma, P}$ and $g_{P}$:
\begin{align*}
g_{\sigma, P}(x) = \int_S g(y) \mathcal{M}_{\sigma, P}(x -y) dy.
\end{align*}
We first start without the assumption that $g$ satisfies the periodic boundary conditions. Note that $g_{P}$ is Lipschitz within $S$. Consider $x \in B_{\infty}(0, \frac12 - \sigma)$, by definition of $\sigma$, there is no wrap around that occur with $\mathcal{M}_{\sigma, P}$ for these values of $x$. Hence, we get
\begin{align*}
    \lvert g_{\sigma, P}(x) - g_{P}(x)\rvert^{p} &\leq \int_S \lvert g_{\sigma}(y) - g(x) \rvert^p \mathcal{M}_{\sigma, d}(x -y) dy \\ &\leq \int_S  (GR)^p \lVert x - y \rVert_p^p \mathcal{M}_{\sigma, d}(x -y) dy\\
    &\leq \int_S  d(RG)^p \lvert x - y \rvert^p \mathcal{M}_{\sigma, 1}(x -y) dy \\
    &\leq  d(GR)^{p}(2\sigma)^p.
\end{align*}

Since we are within one period, we can drop the $P$ subscript and get $\forall x \in B(0, \frac12-\sigma)$
\begin{align*}
    \lvert g(x) - g_{\sigma}(x) \rvert = \Ocal\left(d^{1/p}GR\sigma\right).
\end{align*}
However, if $g$ satisfies the periodic boundary conditions, then its periodic extension $g_P$ is a Lipschitz-continuous function over $\Rmbb^d$. Hence we then have that guarantee that $\forall x, y \in \Rmbb^d$
\begin{align*}
    \lvert g_{\sigma,P}(x) - g_{\sigma,P}(y)\rvert \leq GR\lVert (x - y)~\text{mod}~S \rVert_{p},
\end{align*}
and repeat the above argument without needing to consider a smaller ball inside $S$.

The Fourier coefficients of $g$ satisfy the bound $\lvert \langle g, \chi_{\nmbf}\rangle \rvert \leq \lVert g \rVert_{\infty}$, using Parseval and that $g$ is over the unit box. Since the periodic convolution of two functions multiplies their Fourier coefficients, it follows that
\begin{align*}
    \lvert g_{\sigma} \rvert_{\Hcal^m}^2&= (4\pi)^m\sum_{\nmbf \in \Zmbb^d} \lvert \langle  g_{\sigma}, \chi_{\nmbf}\rangle \rvert^2 \lVert \nmbf \rVert^{2m} \\ &=  (4\pi)^m\sum_{\nmbf \in \Zmbb^d} \lvert \langle  g, \chi_{\nmbf}\rangle \rvert^2 \lvert\langle  \mathcal{M}_{\sigma, d}, \chi_{\nmbf}\rangle \rvert^2 \lVert \nmbf \rVert^{2m}\\
    &\leq  (4\pi)^m\sum_{\nmbf \in \Zmbb^d} \lVert g \rVert_{\infty}^2 \lvert\langle  \mathcal{M}_{\sigma, d}, \chi_{\nmbf}\rangle \rvert^2 \lVert \nmbf \rVert^{2m}\\
    &= \lVert g \rVert_{\infty}^2 \lvert  \mathcal{M}_{\sigma, d}\rvert_{\Hcal^m}^2
\end{align*}
By Lemma \ref{lem:mollifier_lem}
\begin{align*}
    \lvert g_{\sigma} \rvert_{\Hcal^m} &= \Ocal\left(\lVert g \rVert_{\infty}\left(\frac{Cd}{\sigma}\right)^{3d}\right). 
\end{align*}
\end{proof}

\begin{theorem}[Pseudo-spectral Approximation for Lipschitz]
    Suppose $f : \mathcal{X} \rightarrow \Rmbb$ is a $G$-Lipschitz continuous function in $p$-norm with $[-R,R]^d\subseteq \mathcal{X} \subseteq \Rmbb^d$ and $p \geq 1$, and let $g$ be its $R$-restriction. Suppose that $g$ results in a $\Ocal\left(\frac{1}{R}\right)$ Low-Leakage Evolution (Definition \ref{defn:low-leakage-ev}), then the consequences of Theorem \ref{thm:collocation_trace_distance} apply with same specified value of $\log_2(N)$.
\end{theorem}
\begin{proof}
The proof follows the proof of Theorem \ref{thm:collocation_trace_distance} very closely. However, we no longer have the guarantee that the periodic extension of $g$ is continuous, and so only have the following guarantees:
by Lemma \ref{lem:mollification} there exists an $C^{\infty}(S)$ function $g_{\sigma} : [-\frac12, \frac12]^d \rightarrow \Rmbb$, where $\sigma = \frac{\epsilon}{d^{1/p}GR\lVert b \rVert_1}$,  such that
    \begin{align}
    \lvert g_{\sigma}(x) - g(x) \rvert &= \Ocal\left(\frac{\epsilon}{\lVert b \rVert_1}\right), \forall x \in B_{\infty}\left(\frac12-\epsilon\right) \subset S,\label{eq:infty_norm_g_bound_res}\\
    \lvert g_{\sigma}\rvert_{\Hcal^m} &= \Ocal\left(\lVert g \rVert_{\infty}\left(\frac{C\lVert b\rVert_1^2GRd^{1+2/p}}{\epsilon}\right)^{3d}\right), &\text{for} ~ m \leq d +\Ocal(1).
    \end{align}
    We still consider the PDEs in Equations \eqref{eq:first_pde}-\eqref{eq:fourth_pde}. Using \eqref{eq:infty_norm_g_bound_res}, our choice of $\sigma$ and Lemma \ref{lem:evolution_bound_leakage}, we get that first and last terms in the \eqref{eq:triangle_inequality} are $\Ocal(\epsilon)$. The bounding of the middle norm remains unchanged.
\end{proof}

\begin{lemma}[Pseudo-spectral Method Expectation Error General]
    Suppose we have the hypotheses of Theorem \ref{thm:collocation_trace_distance}.
    Let $h : \mathcal{X} \rightarrow \Rmbb$ be $G$ Lipschitz in $\ell_p$ norm with $\mathcal{X} \subseteq \Rmbb^d$ and $g$ be its $R$-restriction. Then
    \begin{align*}
    &\lvert \sum_{x_\jmbf \in \Gcal} g(x_\jmbf )\lvert \Psi(x_\jmbf )\rvert^2(2N)^{-d} 
     -  \int_S g(y) \lvert \Phi ({y})\rvert^2 dy\rvert \\&=\Ocal\left(\frac{\lVert g\rVert_{\infty}d\sup_{S}\lVert \nabla\left(\lvert \Phi_M \rvert^2)\right)\rVert_{\infty} + GRd^{1/p}}{N} + \lVert g\rVert_{\infty}\lVert \Psi - \Phi\rVert + \frac{\lVert g\rVert_{\infty}\lvert \Phi\rvert_{\Hcal^1}}{M}\right)
    \end{align*}
    for any $M \leq N$
\end{lemma}
\begin{proof}
    Let $\Phi_M := P_M\Phi$ for some $M \leq N$, and $\Zcal = \sum_{x_\jmbf \in \Gcal} \lvert \Phi_M(x_\jmbf)\rvert^2$.
    
    \begin{align*}
    \lvert \sum_{x_\jmbf \in \Gcal} g(x_\jmbf)\lvert \Psi({x_\jmbf})\rvert^2(2N)^{-d} 
     -  \int_S g(y) \lvert \Phi(y)\rvert^2 dy\rvert &\leq \lvert \sum_{x_\jmbf\in \Gcal} g(x_\jmbf)\lvert \Psi({x_\jmbf})\rvert^2(2N)^{-d}  - \sum_{x_\jmbf\in \Gcal} g(x_\jmbf)\frac{\lvert \Phi_M(x_\jmbf)\rvert^2}{\Zcal} \rvert \\
     &+\lvert \sum_{x_\jmbf\in \Gcal} g(x_\jmbf)\frac{\lvert \Phi_M(x_\jmbf)\rvert^2}{\Zcal}  - \sum_{x_\jmbf\in \Gcal} g(x_\jmbf)\lvert \Phi_M(x_\jmbf)\rvert^2(2N)^{-d} \rvert \\
     &+\lvert \sum_{x_\jmbf\in \Gcal} g(x_\jmbf)\lvert \Phi_M(x_\jmbf)\rvert^2(2N)^{-d} - \int_S g(y)\lvert \Phi_M(y) \rvert^2 dy \rvert\\
     &+\lvert \int_Sg(y)\lvert \Phi_M(y) \rvert^2 dy - \int_Sg(y)\lvert \Phi(y) \rvert^2 dy \rvert.
    \end{align*}
    
    For the first term:
    \begin{align*}
    \lvert \sum_{x_\jmbf\in \Gcal} g(x_\jmbf)\lvert \Psi({x_\jmbf})\rvert^2(2N)^{-d}  - \sum_{x_\jmbf\in \Gcal} g(x_\jmbf)\frac{\lvert \Phi_M(x_\jmbf)\rvert^2}{\Zcal} \rvert  &\leq  2\lVert g \rVert_{\infty}\left(\sum_{x_\jmbf\in \Gcal} \lvert\lvert \Psi({x_\jmbf})\rvert^2(2N)^{-d}  - \sum_{x_\jmbf\in \Gcal} \frac{\lvert \Phi_M(x_\jmbf)\rvert^2}{\Zcal}\rvert\right)\\
    & \leq 2\lVert g \rVert_{\infty}(2N)^{-d/2} \left(\sum_{x_\jmbf \in G}\left(\Psi({x_\jmbf}) - \frac{(2N)^{d/2}\Phi_M(x_\jmbf)}{\sqrt{\Zcal}}\right)^2\right)^{1/2} \\
    &\leq 2\lVert g \rVert_{\infty}(2N)^{-d/2} \left(\sum_{x_\jmbf \in \Gcal}\left(\Psi({x_\jmbf}) - \Phi_M(x_\jmbf)\right)^2 \right)^{1/2} \\&+ 2\lVert g \rVert_{\infty}(1- \sqrt{\Zcal}(2N)^{-d/2})\\
    &= 2\lVert g \rVert_{\infty}\left(\int_S(\Psi({y})  -\Phi_M(y))^2dy\right)^{1/2}\\&+ 2\lVert g \rVert_{\infty}(1- \sqrt{\Zcal}(2N)^{-d/2})\\
    &\leq2\lVert g \rVert_{\infty}\lVert \Psi - \Phi\rVert + 2\lVert g \rVert_{\infty}\lVert \Phi - \Phi_M\rVert  \\&+ 2\lVert g \rVert_{\infty}(1- \sqrt{\Zcal}(2N)^{-d/2}),
    \end{align*}
    where the last  equality  follows since $\Psi$ and $\Phi_M$ are both Fourier polynomials in $\Hcal_N$. These are all terms which we know how to bound. 
    
    The second term is bounded by
    \begin{align*}
    \big\lvert \sum_{x_\jmbf\in \Gcal} g(x_\jmbf)\frac{\lvert \Phi_M(x_\jmbf)\rvert^2}{\Zcal}  - \sum_{x_\jmbf\in \Gcal} g(x_\jmbf)\lvert \Phi_M(x_\jmbf)\rvert^2(2N)^{-d} \big\rvert \leq \lVert g \rVert_{\infty}(1- \Zcal(2N)^{-d}).
    \end{align*}
    The fourth is bounded by 
    \begin{align*}
    \lvert \int_Sg(y)\lvert \Phi_M(y) \rvert^2 dy - \int_Sg(y)\lvert \Phi(y) \rvert^2 dy \rvert \leq 2\lVert g\rVert_{\infty}\lVert \Phi_M - \Phi\rVert.
    \end{align*}
    
    The most involved term is the third, which is effectively a Riemann summation error:
    \begin{align*}
    \lvert \sum_{x_\jmbf\in \Gcal} g(x_\jmbf)\lvert \Phi_M(x_\jmbf)\rvert^2(2N)^{-d} - \int_Sg(y)\lvert \Phi_M(y) \rvert^2 dy \rvert &= \lvert\sum_{x_\jmbf \in \mathcal{G}}\int_{C_{x_\jmbf}}g(x_\jmbf)\lvert \Phi_M(x_\jmbf) \rvert^2- g(y)\lvert \Phi_M(y) \rvert^2 dy \rvert \\&\leq
    \lvert\sum_{x_\jmbf \in \mathcal{G}}\int_{C_{x_\jmbf}}g(x_\jmbf)\lvert \Phi_M(x_\jmbf) \rvert^2- g(x_\jmbf)\lvert \Phi_M(y) \rvert^2 dy \rvert \\
    &+\sum_{x_\jmbf \in \Gcal} \int_{C_{x_\jmbf}}\lvert g(x_\jmbf) - g(y)\rvert \lvert\Phi_M(y)\rvert^2 dy,
    \end{align*}
    where $C_{x_\jmbf} = x_\jmbf + [0, 1/2N]^{d}$. 
    
    Continuing,
    \begin{align*}
    \sum_{x_\jmbf \in \Gcal} \int_{C_{x_\jmbf}}\lvert g(x_\jmbf) - g(y)\rvert \lvert\Phi_M(y)\rvert^2 dy &\leq \sum_{x_\jmbf \in \Gcal}\int_{C_{x_\jmbf}} GRd^{1/p}\lVert x_\jmbf - y\rVert_{\infty} \rvert \Phi_M(y)\rvert^2dy\\
    &\leq \frac{GRd^{1/p}}{2N}\cdot \frac{\Zcal}{(2N)^{d}},
    \end{align*}
    and
    \begin{align*}
    \lvert\sum_{x_\jmbf \in \mathcal{G}}\int_{C_{x_\jmbf}}g(x_\jmbf)\lvert \Phi_M(x_\jmbf) \rvert^2- g(x_\jmbf)\lvert \Phi_M(y) \rvert^2 dy \rvert &\leq \lVert g \rVert_{\infty}\sup_{S}\lVert\nabla (\lvert \Phi_M\rvert^2)\rVert_\infty\sum_{x_\jmbf \in \mathcal{G}}\int_{C_{x_\jmbf}} \lVert x -y \rVert_1\\
    &\leq \frac{d\lVert g \rVert_{\infty}\sup_{S}\lVert\nabla (\lvert \Phi_M\rvert^2)\rVert_\infty}{N}.
    \end{align*}

    Recall that $\Phi_M$ is a truncation of the Fourier series of $\Phi$, i.e. $\Phi_M = \sum_{\nmbf \in \mathcal{R}} c_{\nmbf}\chi_{\nmbf}$,  $\mathcal{R} \subset \mathcal{N}$.
    The overall error is bounded by
    \begin{align*}
    &\lVert g\rVert_{\infty}^{-1}\lvert \sum_{x_\jmbf \in \Gcal} g(x_\jmbf)\lvert \Psi({x_\jmbf})\rvert^2(2N)^{-d} 
     -  \int_S g(y) \lvert \Phi ({y})\rvert^2 dy\rvert\\  &\leq \frac{d\sup_{S}\lVert \nabla\left(\lvert \Phi_M \rvert^2)\right)\rVert_{\infty} + GR d^{1/p}\lVert g\rVert_{\infty}^{-1}\Zcal(2N)^{-d}}{N} + 2 \lVert \Psi - \Phi\rVert + 4 \lVert \Phi - \Phi_M\rVert  + 3(1-\Zcal(2N)^{-d}). 
    \end{align*}
    
    We know that if $\lVert \Phi - \Phi_M\rVert = \delta$ and since $\mathcal{R} \subset \Ncal$ (i.e. $\Phi_M$ has support on at most as many Fourier modes as $\Psi$) and $\lVert \Phi \rVert = 1$, then
    \begin{align*}
    1 - \delta^2 &= \int_S \lvert \Phi_M(y)\rvert^2dy = (2N)^{-d}\sum_{x_\jmbf \in \Gcal} \lvert \Phi_M(x_\jmbf)\rvert^2 = \Zcal(2N)^{-d} \\
    \lVert \Phi - \Phi_M\rVert &= \delta.
    \end{align*}
    Hence we can drop the $\Zcal(2N)^{-d} < 1$ term. By Lemma \ref{lem:trunc_high_sobolev}, we know that $\delta \leq \frac{\lvert \Phi \rvert_{\Hcal^1}}{M}$. 
\end{proof}

\begin{lemma}
    Consider the setting of Problem~\ref{prob:restricted_schr}. 
    There exists a modified problem with initial condition $\widetilde{\Phi}_{0} \in \Hcal_N$ such that the solutions to Problem~\ref{prob:restricted_schr} at time $T$ with and without the modified initial condition  are $\Ocal\left(\epsilon\right)$ in $L_2$ distance. 
    
    The requirement on $N$ for the ``proxy'' initial wave function $\widetilde{\Phi}_0$ can be either of the following:
    \begin{align*}
        N = 
        \begin{cases}
            \Ocal\left(\frac{\lvert \Phi_0\rvert_{\Hcal_1}}{\epsilon}\right) & \text{ if}\quad \widetilde{\Phi}_0 = \frac{P_N\Phi_0}{\lVert P_N\Phi_0\rVert}, \\
            \Ocal\left( \frac{d^{-1/4}\lvert \Phi_0\rvert_{\Hcal^d}^{1/d} + \lvert \Phi_0 \rvert_{\Hcal^1}}{\epsilon} \right) & \text{ if}\quad \widetilde{\Phi}_0 = \frac{I_N\Phi_0}{\lVert I_N\Phi_0\rVert}.
        \end{cases}
    \end{align*}
\end{lemma}
\begin{proof}
    All we need to ensure is that $\lVert \widetilde{\Phi}_0 - \Phi_0\rVert =\Ocal\left(\epsilon\right)$. This is because a non-adaptive unitary evolution cannot modify the $L_2$ distance, we will still have that the time-evolved solutions are still $\Ocal(\epsilon)$ apart.
    
    For the first case via truncated Fourier series
    \begin{align*}
    \lVert \widetilde{\Phi}_0 - \Phi_0 \rVert &\leq  \lVert \widetilde{\Phi}_0 - P_N\Phi_0 \rVert + \lVert \Phi_0 - P_N\Phi_0 \rVert \leq \lvert 1 - \lVert P_N\Phi_0 \rVert \rvert + \lVert \Phi_0 - P_N\Phi_0 \rVert \leq \frac{2\lvert \Phi_0\rvert_{\Hcal^1}}{N},
    \end{align*}
    from Lemma \ref{lem:trunc_high_sobolev}. Hence $N = \Ocal\left(\frac{\lvert \Phi_0\rvert_{\Hcal_1}}{\epsilon}\right)$ suffices.
    
    For the second case, via interpolation
    \begin{align*}
    \lVert \widetilde{\Phi}_0 - \Phi_0 \rVert &\leq \lvert 1 - \lVert I_N\Phi_0\rVert \rvert + \lVert \Phi_0 - I_N\Phi_0 \rVert \leq 2\lVert \Phi_0 - I_N\Phi_0 \rVert.
    \end{align*}
    The second term can be bounded using Theorem \ref{thm:total_interpolation_error_l2}:
    \begin{align*}
    \lVert \Phi_0 - I_N\Phi_0 \rVert &\leq \left(\frac{\pi}{4}\right)^{d/4} \frac{1}{\sqrt{(d/2) \Gamma(d/2)} N^{d}} \abs{\Phi_0}_{\Hcal^d} + \frac{\lvert \Phi_0 \rvert_{\Hcal^1}}{N}=\Ocal\left(\frac{\abs{\Phi_0}_{\Hcal^d}}{(d^{1/4}N)^{d}}  + \frac{\lvert \Phi_0 \rvert_{\Hcal^1}}{N}\right),
    \end{align*}
    so
    \begin{align*}
        N = \Ocal\left( \frac{d^{-1/4}\lvert \Phi_0\rvert_{\Hcal^d}^{1/d} + \lvert \Phi_0 \rvert_{\Hcal^1}}{\epsilon} \right)
    \end{align*}
    suffices.
\end{proof}

\section{Improved Performance Under Regularity Assumptions}
\label{app:improved_regularity}

It is highly likely that the scaling of $N$ with $d$ is overly-pessimistic and is solely a limitation of analysis. In this section, we discuss some additional regularity assumptions on the potential that reduce the dependence of $N$ on $d$ from exponential to polynomial. We start with an alternative version of Theorem~\ref{thm:collocation_trace_distance} that showcases the dependence of $N$ on the regularity of the potential $f$ (or its $R$-restriction $g$). Then we proceed to provide a regularity assumption on the wavefunction evolution that also results in significantly improved qubit count.

\subsection{Potential Regularity}
\label{subapp:potential_regularity}

The following is a simple corollary of the proof of Theorem \ref{thm:collocation_trace_distance} without the introduction of the mollified function $g_{\sigma}$.

\begin{theorem}[Pseudo-spectral Approximation with Potential Regularity]
\label{thm:collocation_potential_regularity}
    Suppose $f : \mathcal{X} \rightarrow \Rmbb$ is function with $[-R,R]^d\subseteq \mathcal{X} \subseteq \Rmbb^d$ and let $g$ be its $R$-restriction such that $g \in \Hcal^{d+2}$. Let $\Psi(T)$ be the pseudo-spectral approximation to $\Phi(T)$ at time $T$, where $\Phi$ solves the~\ref{eq:real_schrodinger} equation with $V(x,t) = b(t) f(x)$ and $b:[0,\infty) \rightarrow \Rmbb_{\geq 0}$ Lesbesgue measurable. Then 

    \begin{align*}
         \lVert \Phi - \Psi\rVert \lesssim\norm{a}_{1}\left[\left(\frac{\pi}{4}\right)^{d/4}\frac{\lvert \Phi_0\rvert_{\Hcal^{d+2}} + \left(\frac{2\sum_{k=1}^{d+2}\binom{d+2}{k}\lvert g\rvert_{\Hcal^k}\lVert b\rVert_1}{d+2}\right)^{d+2}}{\sqrt{(d/2) \Gamma(d/2)} N^{d}} + \frac{\lvert \Phi_0\rvert_{\Hcal^{3}} + \left(\frac{2\sum_{k=1}^{3}\binom{3}{k}\lvert g\rvert_{\Hcal^k}\lVert b\rVert_1}{3}\right)^3}{N}\right].
    \end{align*}
\end{theorem}
\begin{proof}
    Basically follows the proof of Theorem \ref{thm:collocation_trace_distance} without the use of $g_{\sigma}$. We also use Lemmas \ref{lem:sobolev_growth_wavefunc}  and Lemma \ref{thm:collocation_bound_real_space}.
\end{proof}

The purposes of the mollification was to bound the $\Hcal^{k}$ norms in the above expression. However, we may already have some guarantees on the regularity of $g$. Specifically, it may be that some of the Sobolev norms vanish, for instance if $f$ is polynomially-like with constant degree. This is one way in which we can remove the linear dependent of $\log_2(N)$ on $d$. For example, consider the following class of functions.

\begin{definition}[$k$-Polynomially Enveloped Function]
\label{defn:polynomially-enveloped}
A $C^\infty$ function $f : [-R, R]^d \rightarrow \Rmbb$ is said to be $k$ polynomially enveloped if there exists a degree $k$ polynomial $h : [-\frac12,\frac12]^d \rightarrow \Rmbb$ such that $\forall \alpha \in \Nmbb^{d}$
\begin{align*}
    \lVert \partial^{\alpha} f\rVert_{\infty} \leq \lVert \partial^{\alpha}h(2Rx)\rVert_{\infty},
\end{align*}
where the degree and coefficients of $h$ does not depend on $d$.
\end{definition}

Under the above assumption, we get the following improved guarantees based on Theorem \ref{thm:collocation_potential_regularity}.

\begin{corollary}[Pseudo-Spectral Space Complexity -- Polynomially-Enveloped Potential]
Suppose we are in the same setting as Theorem \ref{thm:collocation_trace_distance} but additionally that $g$ is polynomially enveloped. Then we can achieve the same guarantees as Theorem \ref{thm:collocation_trace_distance} but with
\begin{align*}
      N = \Ocal\left(\frac{(2R)^{k/2}d^{5k/2-1/4}\lVert b \rVert_1 + \lvert \Phi_0\rvert_{\Hcal^{d}}^{1/d}}{\epsilon}\right).
\end{align*}
\end{corollary}
\begin{proof}

Like usual consider $g(x) = f(2Rx)$ for the $R$-truncation of $f$, which obviously retains our stated assumptions. Recall the ``real-space'' definition of Sobolev seminorms from Eq.~\eqref{eq:real_space_sob_seminorm}:
\begin{align*}
    \lvert g \rvert_{\Hcal^m}^2 = \sum_{\lVert \alpha \rVert_1 = m} \int_S (\partial^{\alpha}g(x))^2 dx.
\end{align*}

By our assumption, we know that $\partial^{\alpha}g(x)$ is almost surely bounded by $\partial^{\alpha}h$, where
\begin{align*}
    h({x}) = C\sum_{\lVert \beta \rVert_1 \leq k} (2R)^{\lVert \beta \rVert_1}{x}^{\beta}
\end{align*}
is a $d$-dimensional $k$-degree polynomial over $\Rmbb^d$ and $C$ is a constant. Hence $\lvert g \vert_{\Hcal^m} = 0$ for all $m > k$, since  $\{ \partial^\alpha : \lVert \alpha \rVert_1 = m\}$ annihilates all monomials that are not degree at least $m$. A simple upper bound gives $\max_{m \leq k }\left(\lvert h \rvert_{\Hcal^m}\right) = \Ocal((\sqrt{2R})^kd^{3/2k})$.

If we combine the above with Lemma \ref{lem:sobolev_growth_wavefunc}, then we get
\begin{align*} \lvert \Phi(\cdot, t)\rvert_{\Hcal^{d+2}} &\leq {\lvert \Phi_0\rvert_{\Hcal^{d+2}}} + \left(\frac{2\sum_{r=1}^{d+2}\binom{d+2}{r}\lvert g\rvert_{\Hcal^r}\lVert b \rVert_1}{d+2}\right)^{d+2} \\
 &={\lvert \Phi_0\rvert_{\Hcal^{d+2}}} + \left(\frac{2\sum_{r=1}^{k}\binom{d+2}{r}\lvert g\rvert_{\Hcal^r}\lVert b \rVert_1}{d+2}\right)^{d+2}\\
 &=\Ocal\left(\lvert \Phi\rvert_{\Hcal^{d+2}} + [(\sqrt{2R}(d)^{5/2})^k\lVert b \rVert_1]^{d}\right),
\end{align*}
where we used that $k$ is independent of $d$. Applying Theorem \ref{thm:collocation_bound_real_space}, we therefore conclude
\begin{align*}
    N = \Ocal\left(\frac{(2R)^{k/2}d^{5k/2-1/4}\lVert b \rVert_1 + \lvert \Phi_0\rvert_{\Hcal^{d+2}}^{1/d}}{\epsilon}\right)
\end{align*}
suffices.
\end{proof}
The above implies that the total number of qubits scales as 
\begin{align*}
    d\log_2(N) =\Ocal\left(d\log(\left[Rd\lVert b \rVert_1 +\lvert \Phi_0\rvert_{\Hcal^{d+2}}^{1/d}\right]/\epsilon)\right).
\end{align*}
This quantity is trivially valid in the setting of polynomial potentials, which could appear in applications of quantum simulation to nonconvex optimization. There are technical challenges in applying this to piecewise-polynomial cases due to the presence of Dirac deltas for higher-order distributional derivatives near points of nonsmoothness, i.e. infinite Sobolev norms.

\subsection{Wavefunction Regularity}
\label{subapp:wavefunction_regularity}
The assumption on the wavefunction that we propose that leads to improved regularity is that the evolution remains sandwiched between Gaussians.
\begin{definition}[Gaussian Enveloped Wavepacket]
\label{defn:gaussian-enveloped}
The solution $\Phi(t)$ to the \ref{eq:real_schrodinger} equation is said to be \emph{Gaussian enveloped} if $\forall t \in [0, T]$
\begin{align*}
    \lvert \Phi(t)\rvert_{\Hcal^m} \leq \lvert G_{\sigma_t} \rvert_{\Hcal^m}
\end{align*}
for all $m \in \Nmbb$, with
\begin{align*}
    G_{\sigma_t}({x}) = \frac{e^{-\frac{\lVert Rx-\mu\rVert_2^2}{4\sigma_t^2}}}{(\sqrt{2\pi}\sigma_t)^{d/2}}
\end{align*}
and $\sigma_t = \Omega(1/\sqrt{d})$.
\end{definition}

Since we have an assumed regularity of the wavefunction, we can bypass the use of the Sobolev growth bound. This leads to the following improved result.
\begin{theorem}[Pseudo-Spectral Space Complexity -- Gaussian-Enveloped Wavepacket]
\label{thm:collocation_trace_distance_gaussian}
Suppose we are in the same setting as Theorem \ref{thm:collocation_trace_distance} but additionally $\Phi$ is Gaussian enveloped. Then we can achieve the same guarantees as Theorem \ref{thm:collocation_trace_distance} but with
\begin{align*}
     N = \Ocal\left(d^{1.25}R/\epsilon\right).
\end{align*}

\end{theorem}
\begin{proof}

The partial derivatives of a Gaussian are related to the Hermite polynomials $H_m$:
\begin{align*}
\frac{d^{m}}{dx^m}e^{-(x-y)^2/2\sigma^2} = (-1)^m\sigma^{-m}H_m\left(\frac{x-y}{\sigma}\right)e^{-(x-y)^2/2\sigma^2},
\end{align*}
and so our assumption implies that our wavefunction $\Phi : S \rightarrow \Cmbb$ satisfies the following nice property:
\begin{align*}
\sum_{\lvert \alpha\rvert = m}\int_S\lvert \partial^{\alpha}\Phi({x})\rvert^2 \leq  \sum_{\lvert \alpha\rvert = m}R^{2m}\int_S\left(\frac{e^{-\frac{\lVert R{x}-{\mu}\rVert_2^2}{4\sigma^2}}}{(\sqrt{2\pi}\sigma)^{d/2}} \sigma^{-\lvert \alpha\rvert}\prod_{j=1}^{d}H_{\alpha_j}(\frac{Rx_j-\mu_j}{\sigma})\right)^2,
\end{align*}
which is effectively the assumption that the wavefunction sits between Gaussian wavepackets for all time.

Then we will have that 
\begin{align*}
\lvert \Phi \rvert_{\Hcal^m}^2 &= \sum_{\lVert \alpha\rVert_1 = m} \langle \partial^{\alpha}\Phi, \partial^{\alpha}\Phi\rangle\\
&= \sum_{\lvert \alpha\rvert = m} \int_S \lvert\partial^{\alpha}\Phi\rvert^2 dx\\
&\leq \sum_{\lVert \alpha\rVert_1 = m}R^{2m}\int_S \frac{e^{-\frac{\lVert R{x} -{\mu}\rVert_2^2}{2\sigma^2}}}{(\sqrt{2\pi}\sigma)^{d}} \sigma^{-2\lvert \alpha\rvert}\prod_{j=1}^{d}[H_{\alpha_j}(\frac{Rx_j-\mu_j}{\sigma})]^2\\
&=\Ocal\left(R^{2m}\sigma^{-4m}\int_{\Rmbb^d} \lVert {x} - {\mu}\rVert_2^{2m}\frac{e^{-\frac{\lVert {x}-{\mu}\rVert_2^2}{2\sigma^2}}}{(\sqrt{2\pi}\sigma)^{d}}\right)\\
&=\Ocal\left(d\left(\frac{4m^2R^2}{\sigma^2}\right)^{m}\right).
\end{align*}
Therefore, 
\begin{align*}
\lvert \Phi \rvert_{\Hcal^m} =  \Ocal\left(\sqrt{d}\left(\frac{2mR}{\sigma}\right)^{m}\right).
\end{align*}

Under the assumed uniform lower bound on $\sigma$, it follows:
\begin{align*}
\lvert \Phi \rvert_{\Hcal^d} =  \Ocal\left(\sqrt{d}\left(2d^{3/2}R\right)^{d}\right).
\end{align*}
Upon inspecting the denominator of the leading term in Eq.~\eqref{eq:collocation_bound}, the decay looks like
\begin{align*}
\left(\frac{\sqrt{d}\left(2d^{3/2}R\right)^{d}}{\sqrt{\Gamma(d/2)}N^{d}}\right) = \Ocal\left(\frac{\sqrt{d}\left(2d^{3/2}R\right)^{d}}{[(d/2)^{1/4}N]^{d}}\right),
\end{align*}
meaning that $N = \Ocal\left(d^{5/4}R/\epsilon\right)$ suffices.

\end{proof}

One may note that  Theorem \ref{thm:energy_error_bound} is not ``compatible'' with the reduced qubit count provided by a polynomially-enveloped potential or Gaussian-enveloped wavepacket, in the sense that these assumptions do not result in an improved qubit-count for Riemann summation error. However, the following can be modified result can be. Specifically, if we have an analytic wave function with fast Fourier decay, then we can get back to linear in $d$ qubits for both guarantees of Theorem \ref{thm:master_simulation_thm}.

\begin{lemma}[Paley-Wiener]
\label{lem:paley-wierner}
    Suppose $f:(\Rmbb/2R\Zmbb)^d\rightarrow \Rmbb$ is real analytic. Then there exists $\kappa \in \Rmbb_+$ such that
    \begin{equation*}
        \abs{\hat{f}_\nmbf} \leq e^{-\kappa \norm{\nmbf}_1} \sqrt{(2\pi)^d} \max_{z\in ([-R,R] + i[-\kappa,\kappa])^d} \abs{f(z)}      
    \end{equation*}
\end{lemma}
\begin{proof}
    The $\nmbf$th Fourier coefficient is given by
    \begin{align*}
        \hat{f}_\nmbf = \langle\chi_\nmbf, f\rangle &= \frac{1}{\sqrt{(2R)^d}} \oint_{[-R,R]^d} e^{-i\pi \nmbf\cdot x/R} f(x) dx \\
        &= \frac{1}{\sqrt{(2\pi)^d}} \oint_{\Tmbb^d} e^{-i \nmbf\cdot y} \tilde{f}(y) dy
    \end{align*}
    where $\tilde{f}(y) = f(R y /\pi)$ is real analytic. For each $y\in[-\pi,\pi]^d$, let $N_y$ be a neighborhood in $\Cmbb^d$ on which $\tilde{f}$ is complex analytic. Then $N = \cup N_y$ is an open neighborhood containing $[\pi,\pi]^d$. By the compactness of $[-\pi,\pi]^d$ and the tube lemma, there exists a tube $U \times [-\pi,\pi]^d$ where $U$ is a neighborhood of $0$ on the imaginary part of $\Cmbb^d$. Thus, there exists a $\kappa > 0$ such that $\tilde{f}$ is analytic at $z$, where
    \begin{equation*}
        z_j = y_j - i \kappa \sgn(n_j).
    \end{equation*}
    Here $\sgn$ is the sign function with $\sgn(0) = 0$ by convention. From the Cauchy Integral theorem for multivariate holomorphic functions, we perturb the integration along the variable $z$ to obtain
    \begin{align*}
        \hat{f}_\nmbf = e^{- \kappa \norm{\nmbf}_1} \frac{1}{\sqrt{(2\pi)^d}} \oint_{\Tmbb^d} e^{-i\nmbf\cdot y} \tilde{f}(z) dy.
    \end{align*}
    Taking absolute values, using the triangle inequality, and maximizing over possible integration contours, one obtains
    \begin{align*}
        \abs{\hat{f}_\nmbf} \leq e^{-\kappa \norm{\nmbf}_1} \sqrt{(2\pi)^d} \max_{z\in ([-R,R] + i[-\kappa,\kappa])^d} \abs{f(z)}
    \end{align*}
    which is the lemma statement.
\end{proof}
\noindent We remark that a suitable $\kappa$ might be more easily established in practical instances, or it may simply be assumed that there is complex analyticity out to $\pm \kappa$ on the imaginary axis. 

\begin{theorem}[Pseudo-spectral Method Expectation Error -- Analytic Wave function]
\label{thm:energy_diff_analytic}
    Suppose we have the hypotheses of Theorem \ref{thm:master_simulation_thm}, along with the guarantee that $d\log_2(N)$ is chosen such that the Theorem guarantees that
    \begin{align*}
        \lVert \Phi - \Psi \rVert = \Ocal\left(\epsilon\right),
    \end{align*}
    where $\Psi$ is the output of the pseudo-spectral method and $\Psi$ the true solution. If $\Phi$ is also analytic, then also have the guarantees of Theorem \ref{thm:energy_error_bound} but $N$ only needs to be 
    \begin{align*}
    N =\Ocal\left(\frac{(4\pi\sqrt{d})^d\lvert\Phi\rvert_{\Hcal^1}\max_{z\in ([-R,R] + i[-\kappa,\kappa])} \abs{\Phi(z)}}{\kappa^{d-1}\epsilon^2}\right).
    \end{align*}
    Hence if $\kappa = \Omega(\sqrt{d})$, then $d\log_2(N) = \Ocal(d\log(\lvert\Phi\rvert_{\Hcal^1}\max_{z\in ([-R,R] + i[-\kappa,\kappa])} \abs{\Phi(z)}/\epsilon))$.
\end{theorem}
\begin{proof} 
    Like before
    \begin{align*}
    \lVert \nabla(\lvert \Phi_M\rvert^2) \rVert_{\infty} \leq 2\lVert \Phi_M \rVert_{\infty} \lVert \nabla \Phi_M\rVert_{\infty} \leq 2\lVert \Phi_M \rVert_{\infty} \sum_{\nmbf \in \mathcal{R}}\lVert \nmbf \rVert_{\infty}\lvert c_{\nmbf}\rvert \leq 2M\left(\sum_{\nmbf \in \mathcal{R}} \lvert c_{\nmbf}\rvert\right)^2 .
    \end{align*}
    Now suppose $\Phi$ is analytic. Then, applying Lemma~\ref{lem:paley-wierner} we have that
    \begin{equation*}
        \abs{\hat{\Phi}_\nmbf} \leq e^{-\kappa \norm{\nmbf}_1} \sqrt{(2\pi)^d} \max_{z\in ([-R,R] + i[-\kappa,\kappa])^d} \abs{\Phi(z)} =: Ce^{-\kappa \norm{\nmbf}_1} \sqrt{(2\pi)^d} 
    \end{equation*}
    and hence,
    \begin{align*}
    \sum_{\nmbf \in \mathcal{R}} \lvert c_{\nmbf} \rvert &\leq C\sqrt{(2\pi)^d}\sum_{\nmbf \in \mathcal{R}}  e^{-\kappa \norm{\nmbf}_1}  \\
    &\leq C\sqrt{(2\pi)^d}\sum_{\nmbf \in \mathcal{R}}  e^{-\kappa \norm{\nmbf}_2}\\f
    &\lesssim   C\sqrt{(2\pi)^d} \int_{\Rmbb^d} e^{-\kappa \norm{\mathbf{y}}_2} d\mathbf{y} \\
    & \lesssim   C\frac{(2\pi)^{d} \Gamma(d)}{\kappa^{d-1}\Gamma(d/2)} \\
    & \lesssim   C\frac{(4\pi\sqrt{d})^{d} }{\kappa^{d-1}}.
    \end{align*}
    From here, the result follows from Lemma \ref{lem:energy_error_bound_gen}.
\end{proof}

\end{document}